\documentclass[conference]{IEEEtran}
\ifCLASSINFOpdf
\else
\fi
\usepackage{color}
\usepackage{xcolor}
\usepackage{tabu}
\usepackage{colortbl}
\usepackage{graphicx}
\makeatletter
\newcommand*\bigcdot{\mathpalette\bigcdot@{.5}}
\newcommand*\bigcdot@[2]{\mathbin{\vcenter{\hbox{\scalebox{#2}{$\m@th#1\bullet$}}}}}
\makeatother
\usepackage{amsmath}
\usepackage{booktabs}
\usepackage{threeparttable}
\usepackage{tikz}
\usepackage{array,tabularx}
\usepackage[most]{tcolorbox}
\usepackage{PTSansNarrow}
\usepackage{amssymb}
\usepackage{amsthm}
\usepackage{algorithm}
\usepackage{algorithmicx}  
\usepackage{algpseudocode} 
\usepackage{subfigure}
\usepackage{enumerate}
\usepackage{enumitem}
\setenumerate[1]{itemsep=0pt,partopsep=0pt,parsep=\parskip,topsep=5pt}
\setitemize[1]{itemsep=0pt,partopsep=0pt,parsep=\parskip,topsep=5pt}
\newtheorem{theorem}{Theorem}
\newtheorem{lemma}{Lemma}

\newtheorem{definition}{Definition}
\newcommand{\etal}{\emph{et al.} }

\newcommand{\p}{\mathcal{P}}
\usepackage[breaklinks]{hyperref}
\usepackage{lipsum} % for dummy text
\usepackage{multicol}
\usepackage{cases}

\newcommand{\new}[0]{\color{black}}
\newcommand{\revise}[0]{\color{black}}
\newcommand{\rrevise}[0]{\color{black}}
\newcommand{\rrrevise}[0]{\color{black}}

\usepackage{empheq}

\usetikzlibrary{matrix}

\begin{document}
%
% paper title
% can use linebreaks \\ within to get better formatting as desired
\title{Manifoldchain: Maximizing Blockchain Throughput via Bandwidth-Clustered Sharding}

% author names and affiliations
% use a multiple column layout for up to three different
% affiliations
% \author{\IEEEauthorblockN{Michael Shell}
% \IEEEauthorblockA{Georgia Institute of Technology\\
% someemail@somedomain.com}
% \and
% \IEEEauthorblockN{Homer Simpson}
% \IEEEauthorblockA{Twentieth Century Fox\\
% homer@thesimpsons.com}
% \and
% \IEEEauthorblockN{James Kirk\\ and Montgomery Scott}
% \IEEEauthorblockA{Starfleet Academy\\
% someemail@somedomain.com}}

% conference papers do not typically use \thanks and this command
% is locked out in conference mode. If really needed, such as for
% the acknowledgment of grants, issue a \IEEEoverridecommandlockouts
% after \documentclass

% for over three affiliations, or if they all won't fit within the width
% of the page, use this alternative format:
% 
\author{
\IEEEauthorblockN{
Chunjiang Che\IEEEauthorrefmark{1},
Songze Li\IEEEauthorrefmark{2}\IEEEauthorrefmark{3},
Xuechao Wang\IEEEauthorrefmark{1},
}
\IEEEauthorblockA{
\IEEEauthorrefmark{1}
The Hong Kong University of Science and Technology (Guangzhou)
}
\IEEEauthorblockA{
\IEEEauthorrefmark{2}Southeast University
% Engineering Research Center of Blockchain Application,\\
% Supervision And Management (Southeast University),\\
% Ministry of Education
}

\IEEEauthorblockA{
\IEEEauthorrefmark{3}
Engineering Research Center of Blockchain Application, Supervision and Management (Southeast University), \\
Ministry of Education
}
}

% use for special paper notices
%\IEEEspecialpapernotice{(Invited Paper)}

% \IEEEoverridecommandlockouts
% \makeatletter\def\@IEEEpubidpullup{6.5\baselineskip}\makeatother
% \IEEEpubid{\parbox{\columnwidth}{
%     Network and Distributed System Security (NDSS) Symposium 2024\\
%     26 February - 1 March 2024, San Diego, CA, USA\\
%     ISBN 1-891562-93-2\\
%     https://dx.doi.org/10.14722/ndss.2024.23xxx\\
%     www.ndss-symposium.org
% }
% \hspace{\columnsep}\makebox[\columnwidth]{}}

% make the title area
\maketitle

\begin{abstract}
Bandwidth limitation is the major bottleneck that hinders scaling throughput of proof-of-work blockchains. To guarantee security, the mining rate of the blockchain is determined by the miners with the lowest bandwidth, resulting in an inefficient bandwidth utilization among fast miners. We propose Manifoldchain, an innovative blockchain sharding protocol that alleviates the impact of slow miners to maximize blockchain throughput. Manifoldchain utilizes a bandwidth-clustered shard formation mechanism that groups miners with similar bandwidths into the same shard. Consequently, this approach enables us to set an optimal mining rate for each shard based on its bandwidth, effectively reducing the waiting time caused by slow miners. Nevertheless, the adversary could corrupt miners with similar bandwidths, thereby concentrating hashing power and potentially creating an adversarial majority within a single shard. To counter this adversarial strategy, we introduce \textit{sharing mining}, allowing the honest mining power of the entire network to participate in the secure ledger formation of each shard, thereby achieving the same level of security as an unsharded blockchain.
Additionally, we introduce an asynchronous atomic commitment mechanism to ensure transaction atomicity across shards with various mining rates. Our theoretical analysis demonstrates that Manifoldchain scales linearly in throughput with the increase in shard numbers and inversely with network delay in each shard. We implement a full system prototype of Manifoldchain, comprehensively evaluated on both simulated and real-world testbeds. These experiments validate its vertical scalability with network bandwidth and horizontal scalability with network size, achieving a substantial improvement of 186\% in throughput over baseline sharding protocols, for scenarios where bandwidths of miners range from 5Mbps to 60Mbps.
\end{abstract}
% IEEEtran.cls defaults to using nonbold math in the Abstract.
% This preserves the distinction between vectors and scalars. However,
% if the conference you are submitting to favors bold math in the abstract,
% then you can use LaTeX's standard command \boldmath at the very start
% of the abstract to achieve this. Many IEEE journals/conferences frown on
% math in the abstract anyway.

% no keywords

% For peer review papers, you can put extra information on the cover
% page as needed:
% \ifCLASSOPTIONpeerreview
% \begin{center} \bfseries EDICS Category: 3-BBND \end{center}
% \fi
%
% For peerreview papers, this IEEEtran command inserts a page break and
% creates the second title. It will be ignored for other modes.
%%\IEEEpeerreviewmaketitle

\section{Introduction}

% {\reviewer
% Reviewer 4: Is the sharing mining approach specific to Manifoldchain? Seems like it generally solves the issue of compromised shard in a sharded blockchain system. Is that true?

% Reviewer 4: One problem with cross-shard transaction is that some shard may become unavailable. How does Manifoldchain solve that issue?
% }

Blockchain technology, pioneered by Nakamoto's proof-of-work (PoW) longest-chain protocol~\cite{nakamoto2008bitcoin}, has garnered substantial attention in recent years. PoW blockchains are distinct in offering unique properties like dynamic availability, unpredictability, and security against adaptive adversaries — merits that have been validated both theoretically~\cite{garay2017bitcoin,garay2020does} and practically~\cite{charts}. However, a fundamental challenge remains its inherent poor scalability, hindering its broader adoption and real-world applicability. A major impediment to scale PoW blockchains lies in the necessity for every miner to replicate the communication, storage, and state representation of the entire ledger. Blockchain sharding protocols have emerged as a promising solution to address this challenge. Sharding protocols allocate miners to distinct shards, where they independently mine separate chains within their respective shards~\cite{luu2016secure,kokoris2018omniledger,wang2019monoxide,zamani2018rapidchain,rana2022free2shard,gencer2016service}. This innovative approach enables a linear scaling of system throughput in proportion to the network size.

However, a noteworthy shortcoming of most existing sharding protocols is the lack of attention to the bandwidth heterogeneity. {\new Bandwidth variance can be substantial in real-world scenarios. For example, some Bitcoin miners operate within the Tor~\cite{tor_wiki} network to safeguard their privacy, where only half of the available bandwidth in the Tor network is utilized~\cite{loesing2009measuring, 10.1145/2946802, 6407715}, resulting in significantly smaller bandwidth compared to normal miners.}
%Another challenge for a PoW blockchain to scale to its physical limit is the ``forking'' issue caused by bandwidth constraints. 
Miners with limited bandwidth resources, often referred to as \textit{stragglers}, contribute to increased network delays, which result in forking of blockchains, i.e., two distinct blocks extend the same preceding block. %For instance, due to network delay, it may take a prolonged period of time for a newly generated block to be received by other miners.
% may endure a prolonged interval from its creation to its reception by another miner. 
%During this time, other miners may have already mined another block that extends the same parent block. 
Forking significantly undermines blockchain security by diminishing the effective hashing power of honest miners, and elevating the probability of adversaries successfully tampering with ledger information. Consequently, fast miners are compelled to reduce their mining rates in order to mitigate forking, limiting overall transaction throughput. Most of the existing sharding protocols~{\new \cite{luu2016secure, kokoris2018omniledger, zamani2018rapidchain, gearbox, reticulum, zhang2023frontrunning, wang2019monoxide}} typically employ a uniform shard formation (USF) mechanism across shards, resulting in each shard accommodating both stragglers and fast miners. Therefore, the throughput in each shard is still limited by bandwidths of stragglers.

We introduce \textit{Manifoldchain} as a comprehensive solution to tackle this challenge. The key insight behind Manifoldchain is to cluster miners with similar bandwidths into the same shard, thereby segregating fast miners from stragglers. Consequently, an optimal mining rate can be set for each shard, tailored to bandwidth resources of miners within. This approach empowers fast miners to fully harness their computation resources, enabling them to propose blocks at a much higher rate without being impeded by stragglers. {\new Manifoldchain employs a bandwidth-clustered shard formation (BCSF) mechanism. Specifically, miners package their bandwidth information, public key, and IP address as a ``Pseudo-identity'' (PID), and extend a credential chain with transactions encapsulating these PIDs. PID owners confirmed on the credential chain are subsequently distributed across distinct shards corresponding to their bandwidths.} Once distributed, they begin processing transactions and growing ledgers within their respective shards. 

% {\old To address the potential issue of bandwidth deception, miners employ the Proof-of-Backhaul (PoB) protocol \cite{DBLP:journals/corr/abs-2210-11546} to authenticate the accuracy of the declared bandwidths.}{\new deleted} 
%Miners who successfully obtain the required credentials sign their blocks using their respective secret keys, and these blocks can be verified using the public keys stored in the credential-chain.

A pivotal challenge in realizing such a sharding protocol is ensuring security within each shard in the presence of a mildly adaptive adversary (a general adversary model adopted by most sharding protocols). It may seek to corrupt miners with similar bandwidths, thereby concentrating adversarial hashing power within a single shard and potentially yielding an adversarial ratio exceeding 50\%. To tackle this challenge, we introduce \emph{sharing mining}, which utilizes cryptographic sortition to determine whether a miner will produce an \emph{exclusive block} to a specific shard, or an \emph{inclusive block} that can be appended to chains in all shards. Sharing mining enables Manifoldchain to aggregate the honest hashing power from the entire network to secure each individual shard, while simultaneously benefiting from boosted throughput due to various mining rates across those shards.   
% as a strategic solution. Sharing Mining serves to augment the honest hashing power of each shard and distribute it evenly across different shards. This mechanism enables miners to simultaneously extend multiple longest chains in all shards using a single PoW solution. 

Nevertheless, as our design avoids full blocks in one shard from being sent to miners in other shards, it may lead to data availability and validity issues. We utilize \textit{coded Merkle tree} (CMT)~\cite{DBLP:conf/fc/YuSLAKV20} and \textit{fraud proof}~\cite{al2018fraud} to verify data availability/validity without downloading the full blocks. Additionally, sharing mining incorporates strategies like \emph{predictive mining} and \emph{fork pruning} to parallelize mining and verification processes. Particularly, predictive mining empowers miners to mine on multiple unverified parent blocks while concurrently requesting data availability and validity proofs. Subsequently, miners can prune forks upon receiving unavailability or invalidity proofs. 

Another challenge arises from the various mining rates across different shards when processing cross-shard transactions (\textit{cross-txs}), which may lead to inconsistent cross-shard asset transfers. The standard atomic commitment protocol \textit{Two-Phase Commit} (2PC)~\cite{gray2005notes} employed by existing sharding protocols such as OmniLedger~\cite{kokoris2018omniledger}, is unsuitable for Manifoldchain. For instance, 2PC requires the verification process of a cross-tx to pause until confirmed vote messages are received from all involved shards, introducing prohibitive latency for confirming cross-txs. To address this challenge, we adopt an asynchronous commitment mechanism by eliminating the locking phase of 2PC. Specifically, coins are directly spent instead of being locked. Upon any unsuccessful expenditure, the coins are refunded to the payers.

% Another challenge in designing Manifoldchain is to develop an asynchronous verification mechanism and a consistent confirmation rule for handling cross-shard transactions (cross-txs) under variable mining rates, especially with multiple inputs and outputs. We introduce \emph{attestable atomicity}, 
% % as a solution to tackle this concern. This approach 
% which incorporates a ``lock-review-unlock'' mechanism to guarantee atomicity in cross-tx confirmation. It also enables asynchronous verification of cross-txs by associating each such transaction with a corresponding \textit{testimony}. Furthermore, in the context of an Unspent Transaction Output (UTXO) model, achieving confirmation becomes increasingly intricate as mining rates fluctuate among shards. In this case, cross-tx outputs may reach confirmation faster than their corresponding inputs, thereby giving rise to potential inconsistency in cross-tx confirmation. Consequently, attestable atomicity introduces an innovative confirmation rule that employs a two-phase approach—pre-confirmation and final-confirmation—to ensure atomic cross-tx confirmation.

We have conducted a comprehensive theoretical analysis of Manifoldchain, emphasizing its security and overall system throughput.  
Our analysis rigorously establishes the persistence and liveness properties of Manifoldchain, offering a precise upper bound on the error probability, a significant advance over the asymptotic results commonly found in previous works~\cite{garay2015bitcoin,dembo2020everything}. This precise security characterization plays a crucial role in guiding the optimal configuration of protocol parameters. Notably, we address the determination of mining rates across different shards through an optimization problem, delivering an optimal solution for these settings. Finally, we formally prove that the optimal configuration of mining rates allows Manifoldchain to achieve a linear horizontal scaling with the number of miners, and a linear vertical scaling with the reciprocal of network delay in each shard. %This remarkable scalability attribute contributes to the overall scalability of Manifoldchain, highlighting its potential to accommodate increased demands without compromising security.
% {\reviewer Reviewer 1: POS is a promising protocol to replace the POW protocol due to the power consumption issues, such as in Ethereum. However, this paper introduces so many additional mining decision-making and verification mechanisms, will it bring greater power consumption to the high power consumption of POW?} Moveover, Manifoldchain is applicable to PoS and other PoXs protocols. \sz{This is quite abrupt here. Not related to the theoretical analysis.}

Furthermore, to manifest the theoretical promise of Manifoldchain, we have implemented it concisely in about 15,000 lines of Rust code (available open source~\cite{codeAnon}). Our implementation has undergone comprehensive evaluation across diverse scenarios, encompassing a realistic testbed hosted on Amazon EC2, as well as a simulated testbed on a local machine. {\new The experimental results from both realistic and simulated testbeds demonstrate that Manifoldchain outperforms baseline sharding protocols in several key aspects: (1) Manifoldchain delivers approximately $5\times$ throughput in the fastest shard configured with the highest mining rate, and approximately $3\times$ average throughput; (2) Manifoldchain achieves around a $3\times$ increase in throughput with the same increase in the number of miners; (3) Manifoldchain significantly boosts throughput with increased bandwidth resources, while baseline sharding protocols show no improvement.
% Besides, the experimental results on simulated testbed demonstrate Manifoldchain performs better over both horizontal scalability and vertical scalability compared with Monoxide.
}

\section{Related Work}
%Blockchain protocols derived from Bitcoin, which utilize Nakamoto's consensus, are known to experience limited scalability. 
% \subsection{Vertically Scaling Protocols}
\noindent\textbf{Vertically scaling protocols.} The majority of research efforts focused on scaling blockchain performance vertically, namely designing consensus protocols with inherently high performance. Bitcoin-NG \cite{DBLP:conf/nsdi/EyalGSR16} achieves high throughput by decoupling Bitcoin's execution into two distinct primitives: leader election and transaction proposal. 
% In Bitcoin, leader election and transaction serialization are tightly coupled processes that occur simultaneously upon the arrival of a honest block. Transaction serialization demands a greater communication resources than leader election, hence impeding the progress of leader election and resulting in a prolonged system freeze between leader elections. In contrast, Bitcoin-NG ensures that the leader election is forward-looking, and the system serializes transactions asynchronously. Specifically, 
It decouples a full block into two parts: a key block for leader election and a microblock that contains transactions. The protocol relies on only key blocks with significantly reduced block size to expedite leader election, while allowing the elected leader to generate microblocks as fast as possible according to its computing/network resource. Similarly, Prism \cite{DBLP:conf/ccs/BagariaKTFV19} embraces the same decoupling principle, allowing for concurrent leader election and transaction proposal processes. {\new However, these protocols maintain a single ledger, making it difficult to scale with the network size.}

% However, in these protocols, all nodes maintain a single common ledger, meaning the overall performance is limited by the slower nodes in the network, commonly referred to as stragglers.

% Adopting this decoupling ideology, Prism \cite{DBLP:conf/ccs/BagariaKTFV19} divides blocks into proposer blocks, transaction blocks, and voter blocks. The proposer blocks and transaction blocks in this context resemble the key blocks and microblocks in Bitcoin-NG, while voter blocks exhibit a notable distinction, which leverage the Practical Byzantine Fault Tolerant (pBFT) consensus algorithm to expedite the confirmation process. Both Bitcoin-NG and Prism enable consensus achievement and transaction processing to occur simultaneously 
% % \xw{--> ``proposal''. There are many other misused quotation marks like this.} 
% by decoupling the leader election from transaction serialization. 

% \subsection{Horizontally Scaling Protocols}
\noindent\textbf{Horizontally scaling protocols.} To overcome the constraint encountered by vertically scaling protocols, several protocols aim to improve performance through horizontal scaling, which involves adding more nodes to distribute the load more effectively, thus boosting the system's throughput. Sharding stands out as the principal method for achieving Horizontal scaling. It divides the blockchain nodes into shards, each of which is only responsible for processing a distinct subset of the transactions. %facilitating the utilization of the available computational and storage resources. 
As a result, sharding enables the overall throughput to scale with the number of nodes.
% be scaled almost linearly with the number of honest miners in the network. 

{\new
Numerous works are dedicated to designing general sharding protocols in the permissionless setting. Elastico~\cite{luu2016secure} uniformly partitions the mining network into small committees, within which miners run pBFT~\cite{castro1999practical} to reach consensus on a disjoint set of transactions. %As an initial design, it does not support cross-transaction atomicity and still necessitates broadcasting of full blocks across shards. 
Building upon this approach, Omniledger~\cite{kokoris2018omniledger} addresses Elastico's limitations of not supporting cross-tx atomicity, by introducing a cross-tx verification mechanism called Atomix, which enables miners to verify cross-txs without storing the full blocks from foreign shards. RapidChain~\cite{zamani2018rapidchain} further enhances the %these advancements. It increases the 
fault tolerance threshold from up to 1/4 in Omniledger to 1/3, boosts throughput within each shard through block pipelining, and reduces communication overhead via efficient routing. %In addition to their inherent low fault tolerance, 
{\revise Nonetheless, these pBFT-based sharding protocols cannot achieve \textit{full sharding} (as formally defined in Definition~\ref{full_sharding_definition}). Specifically, the aforementioned sharding protocols necessitate that the number of honest miners in each shard scales with the shard size. This scaling demand requires a significantly larger shard size to ensure an honest (super)majority within each shard, in accordance with the law of large numbers. In contrast, a full sharding protocol requires only a constant number of honest miners within each shard to ensure security.}
% {\old A PoW-based full sharding protocol Monoxide is proposed in \cite{wang2019monoxide}. Monoxide supports a fault tolerance of up to 1/2 via so-called Chu-ko-nu mining, which ensures post-partition security and rendering the task of attacking any specific shard as challenging as targeting the entire network.}
{\revise Wang \etal introduced Monoxide~\cite{wang2019monoxide}, the first full sharding protocol, which employs Chu-ko-nu mining to make attacking any specific shard as difficult as targeting the entire network. In the best-case scenario, this design can ensure security as long as each shard contains at least one honest miner.} Nevertheless, Monoxide has its own shortcomings: (1) rigorous security analysis is not provided; (2) its proposed cross-tx verification mechanism, eventual atomicity, cannot support many-to-many transaction models.
%Recently, the upgrade to the existing Ethereum network, Ethereum 2.0, was completed, also incorporating sharding techniques to enhance scalability and security. 

Some other sharding solutions aim to tackle sub-problems within a sharding system. GearBox~\cite{gearbox} and Reticulum\cite{reticulum} aims to attain the smallest shard size, thereby maximizing parallelism.
%GearBox~\cite{gearbox} employs the safety-liveness dichotomy to attain the smallest shard size, thereby maximizing parallelism. Nevertheless, it necessitates frequent respawning of new shards during runtime or overlapping shard memberships, potentially resulting in additional overhead. Reticulum\cite{reticulum} addresses this challenge by implementing a two-phase-voting mechanism and distributing miners across different shards with a two-layer structure. In contrast to these endeavors focused on optimizing parallelism,
Haechi\cite{zhang2023frontrunning} focuses on fortifying resilience against front-running attacks, wherein adversaries manipulate transaction execution order to achieve unfair finalization. 
%To mitigate this vulnerability, Haechi introduces an ordering phase between transaction processing and execution, ensuring that the execution order of transactions aligns with the processing order, thereby achieving finalization fairness.
ByShard\cite{DBLP:journals/vldb/HellingsS23} extends the system-specific
specialized sharding protocols to a application-agnostic solution. All these works address various issues distinct from ours, yet Manifoldchain, as a general sharding solution, is compatible with them to address corresponding challenges.
}

% \subsection{Bandwidth Considered Protocols}
% Many blockchain research studies \cite{DBLP:journals/corr/abs-2203-06357, DBLP:conf/ccs/Gazi0R22, DBLP:conf/aft/LiG021, DBLP:conf/eurocrypt/PassSS17} focusing on 
\noindent\textbf{Bandwidth considered protocols.} It has been theoretically demonstrated that the network delay, denoted by $\Delta$, plays an important role in determining an appropriate mining difficulty in synchronous blockchain protocols~\cite{DBLP:journals/corr/abs-2203-06357, DBLP:conf/ccs/Gazi0R22, DBLP:conf/aft/LiG021, DBLP:conf/eurocrypt/PassSS17}.
% security analysis demonstrate that the bounded network delay, denoted as $\Delta$, plays an important role in determining an appropriate mining difficulty in synchronous network. 
Specifically, the block generation rate is constrained by $\Delta$, as arbitrarily accelerating block generation can lead to excessive forking, thereby wasting honest mining power and compromising the network's security.
% , making the network more vulnerable to the adversary. 
Bitcoin-NG and Prism fundamentally enhance vertical scalability by mitigating the impact of $\Delta$ via functionality decoupling. However, they neglect the heterogeneity of blockchain nodes. Specifically, stragglers contribute to higher $\Delta$ values and subsequently become the bottlenecks for system throughput. To mitigate this issue, Yang \textit{et al.} proposed DispersedLedger \cite{DBLP:conf/nsdi/YangPAKT22}, a partial synchronous BFT protocol designed to achieve near-optimal throughput even in scenarios with heterogeneous network bandwidths. %Adopting the same decoupling ideology, 
DispersedLedger enables nodes to agree on proposals at a rapid rate without  downloading the entire block. Instead, nodes are able to retrieve the serialized transactions at their own paces. For instance, DispersedLedger nodes agree on Verifiable Information Dispersal blocks. These blocks utilizes erasure codes to store transactions across $N$ nodes, ensuring that they can be retrieved later in the presence of Byzantine behavior. {\new However, these non-sharding solutions inherently lack horizontal scalability.}

\noindent\textbf{Our protocol.} We propose Manifoldchain, a permissionless full sharding protocol with salient attributes as follows:
\begin{itemize}
    \item Manifoldchain is the first sharding protocol that takes the bandwidth heterogeneity into account. Other aforementioned protocols, either overlook the bandwidth heterogeneity in the sharding setup\cite{luu2016secure, kokoris2018omniledger, zamani2018rapidchain, gearbox, reticulum, zhang2023frontrunning}, or only consider them in the non-sharding setup\cite{DBLP:conf/nsdi/EyalGSR16, DBLP:conf/nsdi/YangPAKT22}. 
    \item Manifoldchain achieves the highest fault tolerance threshold. Compared with pBFT-based sharding protocols\cite{luu2016secure, kokoris2018omniledger, zamani2018rapidchain, gearbox, reticulum, zhang2023frontrunning} whose fault tolerance threshold is up to 1/3, Manifoldchain supports a fault tolerance of up to 1/2 and achieves full sharding as Monoxide. Furthermore, it only requires that there is at least one honest miner in each shard. Compared with Monoxide, it addresses remaining limitation by providing a tight security analysis that guides the selection and optimization of protocol parameters. Additionally, Manifoldchain can support many-to-many cross-tx verification. Table~\ref{comparison_table} provides an overview of Manifoldchain in comparison to other sharding protocols.
    \item Manifoldchain is a general sharding principle that is compatible with many other sharding solutions. For instance, while GearBox and Reticulum improve horizontal scalability by maximizing number of shards, Manifoldchain offers a new angle to boost throughput within individual shards by allowing normal miners and stragglers to operate independently. Insights from previous works can be combined with ours to further improve both horizontal and vertical scalability.
\end{itemize}

\begin{table*}[]
\centering
\caption{Comparison of Manifoldchain with SOTA sharding protocols}
\label{comparison_table}
\begin{threeparttable}
\begin{tabular}{l|ccccccc}
\hline
                             & {\revise Elastico~\cite{luu2016secure}}                 & {\revise Omniledger~\cite{kokoris2018omniledger}}               & {\revise RapidChain~\cite{zamani2018rapidchain}}               & {\revise Gearbox~\cite{gearbox}}                  & {\revise Reticulum~\cite{reticulum}}                                   & {\revise Monoxide~\cite{wang2019monoxide}}                 & \textbf{Manifoldchain}  \\ \hline
Global honest threshold\tnote{1}      & $\frac{3}{4}N$           & $\frac{3}{4}N$           & $\frac{2}{3}N$           & $\frac{2}{3}N$           & $\frac{2}{3}N$                      & $\frac{1}{2}N$           & $\mathbf{\frac{1}{2}N}$ \\
Intra-shard honest threshold\tnote{2} & $\frac{2}{3}\frac{N}{m}$ & $\frac{2}{3}\frac{N}{m}$ & $\frac{2}{3}\frac{N}{m}$ & $\frac{2}{3}\frac{N}{m}$ & $\frac{2}{3}\frac{N}{m}$  & $\geq 1$\tnote{3} & $\mathbf{1}$          \\
Supports many-to-many tx?    & Yes                      & Yes                      & Yes                      & Yes                      & Yes                                            & No                       & \textbf{Yes}            \\
Supports full sharding?      & No                       & No                       & No                       & No                       & No                                              & Yes                      & \textbf{Yes}            \\
Scalable with bandwidth?     & No                       & No                       & No                       & No                       & No                                              & No                       & \textbf{Yes}            \\ \hline
\end{tabular}
\begin{tablenotes}    
    \footnotesize        
    \item[1] The minimum number of honest miners necessary to ensure the security of the entire network, where $N$ is the total number of miners.
    \item[2] The minimum number of honest miners necessary to ensure the security of a specific shard, where $m$ is the number of shards.
    \item[3] In best-case scenario, Monoxide can achieve the same intra-shard honest threshold as Manifoldchain, but it lacks a formal security analysis.
\end{tablenotes}
\vspace{-5mm}
\end{threeparttable}
\end{table*}

% \lipsum[1-3] % Dummy text for demonstration purposes

% \begin{tcolorbox}[colback=blue!10!white,colframe=blue!50!black,title=Spanning Two Columns]
% \begin{multicols}{2}
% \lipsum[4-6] % Dummy text for demonstration purposes
% \end{multicols}
% \end{tcolorbox}

% \lipsum[7-10] % Dummy text for demonstration purposes

\section{Background and Model}
\label{background}
% This section introduces some basic backgrounds of Nakamoto's Consensus and Blockchain sharding protocol.

\subsection{Nakamoto Consensus Protocol}\label{nakamotoconsensus}

% A consensus protocol facilitates agreement among distributed participants on specific data or decisions. Contemporary consensus protocols can be primarily categorized into two main types: pBFT\cite{castro1999practical} and Nakamoto Consensus (NC)\cite{nakamoto2008bitcoin}.

% The pBFT consensus protocol operates by arbitrarily selecting a miner to function as the leader and broadcast a proposal. This proposal gains acceptance if it obtains more than 2/3 of the votes from other miners. In contrast, NC eschews the concept of a designated leader. A miner broadcasts a block, %which concretizes the concept of a proposal, 
% once it successfully discovers a solution to the PoW puzzle. To express agreement with a proposal, a miner works on a PoW solution, utilizing the previous block as its input. Proposals that accumulate the highest number of approvals are considered finalized data or decisions. %Bitcoin as an instance employs NC to enables distributed miners to achieve consensus on a unified ledger without an intermediary financial institution.

Nakamoto consensus (NC), adopted by Bitcoin, operates on the PoW mechanism and the longest chain rule.  It can be described succinctly as follows: at any moment, an honest miner adopts the longest chain available to it and aims to mine a new block extending this chain; {\new a block is considered confirmed once it is sufficiently deep within the chain.}

\noindent\textbf{Proof of Work.} 
% The PoW algorithm determines the eligibility of a miner to generate a block. In the context of Bitcoin, a block refers to a candidate proposal under voting election. Upon finding a solution to the PoW puzzle, a miner can package a subset of transactions into a block, and broadcast it to other miners. 
Finding a valid PoW solution necessitates locating a $\mathtt{nonce}$ value to ensures the output of the SHA256 hash function is less than the preset mining difficulty $\sigma$. We adopt the standard random oracle model~\cite{garay2015bitcoin}, utilized by an algorithm $\mathbf{PoW}^{\sigma}(\mathtt{parent\_hash}, \mathtt{info}, \mathtt{nonce})$ for finding a valid PoW solution, {\new as depicted in Appendix~\ref{pseudocode_pow}. } Here $\mathtt{parent\_hash}$ and $\mathtt{info}$ are the hash of the parent block and block content respectively. The blockchain's immutability stems from the fact that any attempt to modify a single block within the chain inevitably triggers alterations to its hash value as well as to the hash values of all subsequent blocks.

\noindent\textbf{Unspent Transaction Output (UTXO).} %In the ecosystem established by Bitcoin, the UTXO signifies a user's balance. 
A Bitcoin transaction consists of multiple inputs and outputs, each represented as an UTXO. For a Bitcoin transaction $\rm \mathtt{tx}:\{I_1,...,I_k;O_p,...,O_h\}$ with $\rm \sum_{i=1}^k \mathcal{A}(I_k) = \sum_{j=1}^h \mathcal{A}(O_j)$(where $\mathcal{A}(\cdot)$ represents the amount of coins), $k$ inputs are consumed and $h$ outputs are generated.  
% multiple inputs and outputs, denoted as 
% % is characterized as the transfer of assets and encompasses multiple inputs and outputs, facilitating the division and consolidation of value. Formally, a transaction can be denoted as 
% $\rm \mathtt{tx}:\{I_1, I_2, ..., I_k;O_p, O_2, ..., O_h\}$, where $\rm \sum_{i=1}^k I_k = \sum_{j=1}^h O_j$ represents a transaction consisting of $k$ inputs and $h$ outputs. 
Transaction outputs may function as inputs for subsequent transactions; however, an output can be utilized only once. In other words, UTXOs embody the available outputs that can be employed to generate new transactions, effectively representing the user's balance in Bitcoin network.

\subsection{Blockchain Sharding Protocol}

A blockchain sharding protocol partitions miners into different shards, each processing a disjoint set of transactions. Overall, a specific sharding protocol involves two phases: shard formation phase and ledger generation phase. In the shard formation phase, sharding protocols utilize a shard formation mechanism to partition miners into distinct shards. A standard approach is the USF mechanism, which randomly shuffles miners and uniformly distributes them into shards. This initial design ensures that adversarial hashing power is evenly distributed among multiple shards, thereby fortifying each shard against targeted attacks. In the ledger generation phase, an intra-shard consensus protocol like NC or pBFT is employed by miners to reach agreement on the transaction order and update the ledger state within each shard. %NC and pBFT serve as two main consensus protocols for intra-shard consensus. 
The choice between NC and pBFT is scenario-dependent: pBFT is preferred for its faster transaction confirmation times, whereas NC is favored for its adaptability to variable shard sizes. Additionally, cross-shard atomicity must be ensured to guarantee the atomic transfer of assets across shards. Specifically, the outcome of a cross-shard transaction—whether committed or aborted—is consistent across all involved shards.

{\revise
\noindent\textbf{Full Sharding Protocols.} Let $N$ represent the number of miners and $m$ represent the number of shards, we define a full sharding protocol as follows:
\begin{definition}\label{full_sharding_definition}
A sharding protocol achieves full sharding if there exists a constant $n_0$, such that for all $m$ and $N$ satisfying $0 < m \leq N$, security holds as long as each shard contains at least $n_0$ honest miners.
\end{definition}
Intuitively, a full sharding protocol requires only a constant number of honest miners in each shard to ensure security. Early sharding protocols, such as Elastico~\cite{luu2016secure}, Omniledger~\cite{kokoris2018omniledger}, and Rapidchain~\cite{zamani2018rapidchain}, require a linearly scaling number of honest miners in each shard (generally over two-thirds of the miners) and are categorized as non-full sharding protocols. Wang \etal introduced the first full sharding protocol Monoxide~\cite{wang2019monoxide}. In the best-case scenario, it requires only one honest miner in each shard. Our proposed protocol, Manifoldchain, is a more complete full sharding protocol equipped with formal security analysis and a many-to-many transaction design.
}

\vspace{-1mm}
\subsection{Network Model}
%\vspace{-1mm}
% {\reviewer The assumption of a synchronous network is quite strict. I think this is too strong of an assumption for a blockchain system. The authors indicate that the delay of a message is primarily attributed to transmission and propagation delay, but queuing delay is also very important especially when the network is congested. A robust system should be able to handle partial synchrony.}

% Addressing network delays in the NC proves to be considerably more difficult than in the pBFT consensus protocol. It is possible to transform any pBFT consensus protocol into a secure protocol within a $\Delta$-delay network by mandating that all honest miners pause (without taking any action) for  $\Delta$ before responding to any message, effectively simulating synchronous rounds. However, this method is entirely ineffective when applied to NC.

Given that NC-style protocols are known to lack security in a partial synchronous network~\cite{pass2017analysis, DBLP:journals/corr/abs-2203-06357} where there is an unknown bound on the network delays, we adopt the standard synchronous model. This model constrains the adversary to delaying messages from honest nodes by no more than a known maximum delay, denoted as $\Delta$. 
Previous works~\cite{garay2015bitcoin,pass2017analysis, DBLP:journals/corr/abs-2203-06357,dembo2020everything,garay2020does} have consistently overlooked the network heterogeneity, instead assuming the worst-case end-to-end delay to be $\Delta$. This simplification leads to an elegant trade-off between security and throughput in NC, as shown by Pass et al.~\cite{DBLP:conf/eurocrypt/PassSS17}. Let $\rho \in (\frac{1}{2}, 1]$ be the fraction of honest hashing power and $N$  the total number of miners. Pass et al. have theoretically established that to prevent consistency violations, the difficulty parameter $p$ must be set below $\frac{1}{(1-\rho)N\Delta}$, where $p = \sigma/2^{256}$ is the probability of a successful outcome from a single query to the random oracle $\mathbf{PoW}^{\sigma}(\cdot)$. Consequently, as the difficulty parameter $p$ determines the block generation rate, the throughput is directly proportional to $p$ and inversely affected by an increase in delay $\Delta$.
 
In this work, we adopt a standard network model adopted by many other works~\cite{DBLP:conf/nsdi/EyalGSR16, DBLP:conf/ccs/BagariaKTFV19, DBLP:conf/nsdi/YangPAKT22, DBLP:conf/fc/SompolinskyZ15, DBLP:conf/p2p/DeckerW13}. Specifically, miners have various bandwidths, represented as $\mathtt{C}_i$ for miner $i$. The network delay consists of two components: transmission delay $\Delta_t$ and propagation delay $\Delta_p$. Transmission delay refers to the duration required to push all bits of a message into the link, calculated as $\Delta_t = \frac{m}{\mathtt{C}}$ for a message of $m$ bits and a link bandwidth of $\mathtt{C}$ bits per second. Propagation delay $\Delta_p$, on the other hand, measures the time it takes for a bit to travel from one end of the link to the other, dependent on the distance and the speed of light or electrical signal propagation. %{\old In practical settings, especially with large message sizes and high-speed networks, the transmission delay typically overshadows the propagation delay \xw{reference?}.} 
Consequently, let $\Delta_\mathtt{C} = \mathtt{B}/\mathtt{C} + \Delta_p$ be the network delay for a node with bandwidth $\mathtt{C}$, where $\mathtt{B}$ is the maximum block size. We denote by $\Delta = \mathop{\max}_{i}\Delta_{\mathtt{C}_i}$ the maximum delay in the network, primarily determined by the straggler with the smallest bandwidth.

\vspace{-2mm}
\subsection{Threat Model} 
We consider both static and dynamic settings in this paper, tailored to address distinct types of adversaries. {\revise In the static setting, we consider a \textit{static} adversary~\cite{9479723} who corrupts a fixed set of miners when the protocol starts. In contrast, an \textit{adaptive} adversary can change which miners to corrupt during protocol execution. In the dynamic setting, we consider a \textit{mildly} adaptive adversary\cite{kokoris2018omniledger}, which requires a certain delay to transfer corruption. Other types of adaptive adversaries, including weakly and strongly adaptive adversaries~\cite{DBLP:conf/tcc/WanXDS20, DBLP:journals/dc/KlonowskiKM19, DBLP:journals/dc/AbrahamCDNPRS23}, lack this delay and can corrupt all miners in a shard at any time, are not considered in typical sharding protocol designs~\cite{kokoris2018omniledger, zamani2018rapidchain, luu2016secure, DBLP:conf/crypto/KiayiasRDO17}. A detailed comparison of the mildly adaptive adversary with other adaptive adversaries is provided in Appendix~\ref{adaptive_adversary_appendix}.
% In the dynamic setting, we consider the mildly adaptive adversary, a general model commonly adopted by other sharding protocols. The weakly and strongly adaptive adversaries, without a certain delay in corrupting miners, can corrupt all miners in a shard at any time. Defending against such adversaries in a sharding protocol is challenging and is considered an intriguing direction for future work. 
% Despite varying terminology—mildly adaptive in Omniledger~\cite{kokoris2018omniledger}, slowly-adaptive in Rapidchain~\cite{zamani2018rapidchain}, round-adaptive in Elastico~\cite{luu2016secure}, and delayed-adaptive in other protocols~\cite{DBLP:conf/crypto/KiayiasRDO17}—these models assume the same capability of the adversary. 
We present the detailed adversary model as follows:}

\begin{itemize}
    \item In the static setting, the set of participating miners, totaling $N$, remains unchanged with no new entries or exits during the protocol's operation. This scenario assumes a static adversary who, prior to the protocol's start, can corrupt up to $(1 - \rho)N$ miners. The adversary does not have prior knowledge about the miners' participation or specific attributes (particularly their bandwidth information). Once the protocol begins, the adversary cannot modify the selection of corrupted miners. 
    \item In the dynamic setting, miners can join or leave the network during the shard formation phase, but their participation stabilizes during the ledger generation phase. {\rrevise Formally, we denote by $\alpha_i$ the fraction of the new miners among all miners in the $i$-th shard formation phase.}
    % where $0 < \alpha_i \leq 1$ and $\alpha_i \rho N \geq 1$.
    % {\old This scenario accommodates a mildly adaptive adversary.} 
    Specifically, in the initialization phase and each shard formation phase, the adversary can corrupt up to $1-\rho$ fraction of participating miners or new miners without prior knowledge of their bandwidth specifics. Additionally, the adversary has the capability to dynamically change which miners to corrupt but any transition to corrupt additional miners demands a minimum time investment equivalent to the duration of each ledger generation phase. Regardless of the adversary's corruption strategy, the maximum fraction of corrupted miners remains $1-\rho$.
\end{itemize}

Both static and adaptive adversary can arbitrarily determine the participating strategy of the corrupted miners, and fully control their associated communication and computation resources.  
We assume there exits an underlying bandwidth distribution of honest miners, and the new honest miners joining in each shard formation phase adhere to this distribution. Prior to protocol execution, we conduct bandwidth estimation, which can be achieved by leveraging existing techniques~\cite{10.1145/3012426.3022184, 10.1145/1064212.1064242}, even under Byzantine settings~\cite{DBLP:journals/corr/abs-2210-11546}.

\section{Manifoldchain}\label{solution}

\noindent\textbf{Overview.} We introduce Manifoldchain, the first blockchain sharding protocol that simultaneously scales system throughput vertically and horizontally. The key idea is using the BCSF mechanism to allocate miners with similar bandwidths to the same shard. This mechanism clusters stragglers into some specific shards, mitigating their impact %of significantly slow miners 
on the performance of other shards. As a result, we can tailor the mining rate for each shard based on its specific network delay,
%effectively liberating the shard's throughput from constraints imposed by stragglers. 
and fast miners can create blocks and process transactions at an accelerated pace, leading to an overall increase in throughput. Fig.~\ref{fig:idea} demonstrates our key insight. Here, we emphasize two primary challenges that are intrinsic to the design of such a blockchain sharding protocol.

The first challenge is to maintain security within each shard. 
%while also considering the bandwidth issues. 
% {\old The adversary may corrupt miners with similar bandwidths, concentrating all adversarial hashing power in a single shard, resulting in an adversarial majority in that shard.}
The corrupted miners may misstate similar bandwidths, concentrating adversarial hashing power in a single shard, resulting in an adversarial majority in that shard. 
% {\old We propose sharing mining to address this challenge, which allows miners to simultaneously extend multiple longest chains in all shards using a single PoW solution. Sharing mining enables honest miners to distribute their hashing power across all shards, consequently ensuring the security within those shards with adversarial majority.} 
We propose sharing mining to ensure the security within those shards with an adversarial majority. It is implemented by allowing miners to simultaneously extend multiple longest chains in all shards using a single PoW solution, which enables honest miners to contribute their hashing power across all shards. However, in sharing mining, only the block header - rather than the full block - is sent to miners in other shards. The adversary can potentially create equivocation on the highest block through sending heads of invalid blocks to other shards, splitting the honest hashing power within and outside of a shard.
% This makes the protocol vulnerable to a specific type of attack that we refer to as Hashing Power Splitting Attack (HPSA). Specifically, an attacker with overwhelming hashing power may mine a longer chain, containing invalid transactions but appearing valid from the perspective of honest miners in other shards, as they only receive the block headers. This effectively splits the hashing power between miners within a specific shard and those outside it, potentially leading to severe security vulnerabilities. We present a detailed description of HPSA in the Appendix \ref{HPSA}. 
We address this issue by employing predictive mining and fork pruning. Predictive mining enables miners to mine on multiple potential parents and try to request data availability and validity proof at the same time, while fork pruning serves to prune forks with invalid availability or validity proofs.

Another challenge is to develop an asynchronous atomic commitment mechanism to ensure the atomic multi-input-multi-output UTXO transfer across shards under various mining rates. %A UTXO comprises multiple inputs and outputs. The atomic commitment of UTXO manifests in the consistency of inputs and outputs. If all inputs are considered as spent, all outputs are successfully created; conversely, if any input is considered as invalid, no outputs are generated, and other inputs remain unspent. 
The standard atomic commitment 2PC~\cite{gray2005notes} incorporates Voting Phase and Decision Phase to handle this problem. Specifically, a coordinator node prompts participants to vote, and if all vote to commit, broadcast a commit message; otherwise, broadcast an abort message if any votes to abort. In Manifoldchain, the coordinator node cannot begin Decision Phase until all participants' vote messages are confirmed. The participants may be stragglers from slows shards, and their long confirmation time results in an unacceptable cross-tx latency. To address this challenge, we implement an asynchronous commitment mechanism by eliminating the locking phase. Initially, coins are expended directly rather than being locked if input shards vote to commit. Upon receiving all messages voting to commit, the corresponding coins will be subsequently generated in the output shards. Conversely, if any shard votes to abort or any vote message becomes outdated, the associated coins will not be generated, and the coins will be refunded to the input shards with prior successful expenditure.

% For instance, the outputs of a cross-tx within fast shards may be confirmed prior to the confirmation of the corresponding inputs within slow shards, leading to an inconsistent cross-tx confirmation. We propose attestable atomicity to address this challenge. Attestable atomicity utilizes a two-phase process, consisting of a pre-confirmation and final-confirmation stage, to ensure atomic cross-tx confirmation. Specifically, the outputs of a cross-tx are considered pre-confirmed once their associated blocks are succeeded by $\kappa$ descendant blocks, and are regarded as final-confirmed upon confirmations of all corresponding inputs. %Besides, Attestable Atomicity introduces a \textit{lock-review-unlock} mechanism designed to ensure atomicity in cross-transaction commitment. In particular, this entails an initial locking of inputs, which is then followed by the unlocking of these inputs upon the completion of the review process for the associated outputs. 
% Additionally, attestable atomicity pairs each cross-tx with a corresponding testimony, thereby enabling asynchronous verification of the cross-tx within each individual shard. This setup allows mining processes in faster shards to proceed without being impeded by slower shards.

\begin{figure}
  \centering
  \includegraphics[width=8cm]{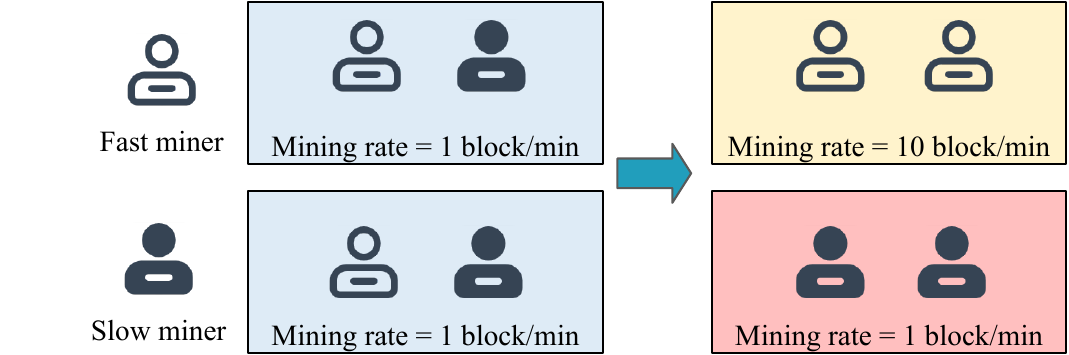}
  \vspace{-2mm}
  \caption{Basic insight of Manifoldchain. The hindrance imposed by bandwidth heterogeneity is alleviated by gathering the stragglers. }
  \label{fig:idea}
  \vspace{-6mm}\hspace{-10mm}
\end{figure}

\subsection{Bandwidth-Clustered Shard Formation}\label{formation}

We propose the BCSF protocol such that miners with similar bandwidths are allocated into the same shard. In BCSF, all miners maintain a \textit{credential-chain} based on NC, incorporating miner PID in place of transactions into the chain. A PID contains a miner's public key, the signed bandwidth using its secret key, and his IP address. The mining process of the credential-chain mirrors that of NC: miners package identities into blocks and strive to find a valid PoW solution. Upon receiving a block, each miner undertakes two verification processes: the first verifies the correctness of the PoW solution, while the second verifies the signature of the bandwidth using the public key provided in the identity. Given the liveness parameter $u$ of the credential-chain, which determines the maximum time it takes for a transaction to be confirmed, and the start time $t$ of the shard formation phase, we ensure that all PIDs of honest miners are confirmed by $t + u$. The owner of a confirmed PID within the time interval $[t, t+u)$ acquires the authorization to participate in the ledger generation phase.

We represent the information of these authorized miners in a two-dimensional space. The domain of the X-axis spans from 0 to $2^{256}$, while the Y-axis is defined within the range of $(\mathtt{C}_{min}, \mathtt{C}_{max})$, where $\mathtt{C}_{min}$ and $\mathtt{C}_{max}$ represent the minimum and maximum bandwidth in the network respectively. The coordinate of a miner, represented as $(h, \mathtt{C})$, is composed of the hash value of its PID $h$ and bandwidth $\mathtt{C}$. Subsequently, different shard formation mechanisms specify how to partition the space and distribute miners into different shards.

% \begin{figure}
%   \centering
%   \includegraphics[width=9cm]{./img/shard_formation.png}
%   \caption{USF mechanism and BCSF mechanism. The former focuses solely on ensuring security, whereas the latter takes both security and heterogeneity into consideration.}
%   \label{fig:formation}
% \end{figure}

\vspace{-1mm}
\begin{figure}[h]
\label{fig:bcsf}
  \centering
  \includegraphics[width=6cm]{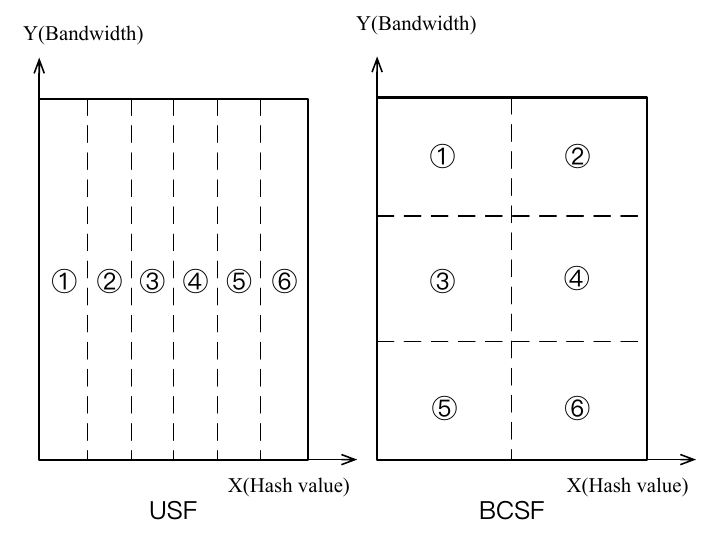}
  \vspace{-4mm}
  \caption{USF vs BCSF. The first one uniformly distributes miners across shards while the second one clusters miners with similar bandwidths.}
  \label{fig:formation_protocol}
  %\vspace{-1mm}
\end{figure}

\noindent\textbf{Shard formation mechanism}. We abstract the USF mechanism as follows: the space is divided into $m$ regions along the X-axis, each corresponding to one shard. Given that hash values can be treated as random numbers, the USF mechanism effectively ensures a uniform distribution of miners across shards; BCSF incorporates an additional consideration of nodes' bandwidth information: the space is partitioned into multiple regions along both the X-axis and Y-axis, facilitating the grouping of miners with similar bandwidths. An example for $m=6$ shards is shown in Fig. \ref{fig:formation_protocol}.

\noindent\textbf{Secure shard partition}. We denote by $S_X$ and $S_Y$ the numbers of segmented regions along X-axis and Y-axis,  respectively. Increasing both $S_X$ and $S_Y$ improves overall throughput by generating more shards. However, it's crucial not to overly expand the number of shards, as doing so may result in shards devoid of honest miners, especially if the total shard count surpasses the number of honest participants. Let $R_Y(\underline{y}, \overline{y})$ represent a Y-region containing points with Y-coordinates within the range $(\underline{y}, \overline{y}]$. A shard formation mechanism segments the space into $S_Y$ Y-regions. 

We can set the Y-region ranges to achieve a uniform distribution of miners across these regions. Particularly, given an estimated bandwidth distribution,
%with a Probability Density Function (PDF) $pdf(\cdot)$, 
we select $S_Y-1$ separation points $y_1,\ldots, y_{S_Y-1}$,
%These separation points partition the space into $S_Y$ Y-regions denoted as $R_Y(y_0, y_1)$, $R_Y(y_1, y_2)$,..., $R_Y(y_{S_Y-1}$, $y_{S_Y})$, 
such that 
%which satisfy 
    \begin{equation}
        \begin{split}
           P(y_0, y_1)=P(y_1, y_2)= \cdots=P(y_{S_Y-1}, y_{S_Y}),
        \end{split}        
    \end{equation}
where $y_0=\mathtt{C}_{min}$, $y_{S_Y} = \mathtt{C}_{max}$, and $P(a, b)$ denotes the probability that a miner's bandwidth is within $(a, b]$. 

During the ledger generation phase, miners with the authorization sign the generated blocks using their secret keys and subsequently broadcast these blocks to other miners. The recipients validate the incoming blocks using the public keys stored within the confirmed PIDs. {\revise For rotation, miners update bandwidth information by broadcasting new PIDs during each shard formation phase. They then access updated shard allocation information from the credential chain and dynamically re-shard the network.}

{\rrrevise
\noindent {\bf Selection of $S_X$ and $S_Y$.} We demonstrate how to select appropriate values for $S_X$ and $S_Y$ to ensure that each shard contains at least one honest miner, as illustrated in Lemma~\ref{honest_presence_proof}. Initially, we set $S_Y$ as large as possible to effectively separate fast and slow miners, while minimizing bandwidth variance within each Y-region. Given a fixed $S_Y$, we set $S_X$ as large as possible to increase the number of shards, as long as the error probability in (\ref{honest_presence_probability}) is negligible. For instance, if all miners have either high or low bandwidth, we set $S_Y=2$ and increase $S_X$ until the error probability exceeds a predefined threshold.
}

\noindent\textbf{Adversary-proofness of BCSF.} Two malicious deviations can occur during the shard formation phase. Firstly, a malicious miner might attempt to enter a specific shard by manipulating its PID, both the hash value and bandwidth information, submitted to the credential-chain. However, by carefully choosing the values of $S_X$ and $S_Y$, BCSF ensures the presence of at least one honest miner in every shard, irrespective of the adversary's choices of shards. This condition is sufficient to establish the security of Manifoldchain, as elaborated in Section~\ref{main_security_analysis}. At a high level, the security and performance of each shard remain intact despite the presence of malicious miners with incompatible bandwidths, thanks to the Byzantine fault tolerance inherent in our protocol. Secondly, a malicious miner may submit multiple PIDs to the credential-chain. However, during the ledger generation phase, such a miner must either distribute its hashing power among various PIDs or concentrate it on a single one. Either case does not change the aggregate adversarial mining power across shards, rendering the strategy ineffective. Note that we rely on the liveness property of NC to ensure that all honest miners' PIDs are confirmed on the credential-chain before the shard formation phase concludes. Attacks at the network layer, such as spamming, fall outside this paper's scope.

\subsection{Block \& Chain Structures}\label{blk_arch}
Our block structure decouples the traditional Bitcoin full block into two types of blocks: the {\it consensus block} and the \textit{transaction block}. %Compared to a Bitcoin full block that requires hundreds of kilobytes to store transactions, 
A consensus block has a fixed size of approximately 100 bytes, introducing negligible overhead for both communication and storage. Consensus blocks are broadcast to all miners in the network, while transaction blocks are only transmitted within their own shards. To support the implementation of sharing mining, we categorize consensus blocks into exclusive blocks and inclusive blocks. %In this section, our primary emphasis lies in elucidating the block structure. Further insights into the design will be expounded upon in Section \ref{sharing_mining}.

Distinct from a Bitcoin full block, a consensus block could have multiple parents. More specifically, an exclusive block extends chains within its affiliated shard and can have multiple parents, whereas an inclusive block, with the ability to have parents from multiple shards, extends chains of all shards. This design, allowing for multiple parents, facilitates predictive mining where miners may have to mine on multiple unverified blocks. This will be further elaborated in Section~\ref{sharing_mining}. We employ the idea of 2-for-1 PoW~\cite{garay2015bitcoin} to mine both exclusive and inclusive blocks concurrently. Initially, miners engage in the process of mining a consensus block and they do not know the type of the consensus block until the puzzle is solved. When a consensus block is mined, it is considered mined of either exclusive block or inclusive block depending on the region the block hash falls. {\new Specifically, A miner mines a consensus block when it discovers a nonce that satisfies $\mathbf{PoW}^{\sigma}(\bigcdot, \mathtt{nonce}) < \sigma$. Further, we can set a threshold $\sigma' < \sigma$. If $\mathbf{PoW}^{\sigma}(\bigcdot, \mathtt{nonce}) < \sigma'$, the consensus block is deemed an inclusive block; otherwise, it is classified as an exclusive block.} We present all the segments of a consensus block as follows:

\begin{itemize}
    \item $\mathtt{shard\_index}$: an index denoting the specific shard to which the block is affiliated.
    \item $\mathtt{verified\_parent}$: hash of the highest verified block within the affiliated shard.
    \item $\mathtt{inter\_parents}$: a set composed of all unverified blocks' hashes extending the highest verified block within the affiliated shard.
    \item $\mathtt{global\_parents}$: a set composed of all inter parents across all shards.
    \item $\mathtt{timestamp}$: the time when the block is generated.
    \item $\mathtt{nonce}$: the resolved solution to a PoW puzzle within the context of sharing mining.
    \item $\mathtt{tx\_merkle\_root}$: the root of a Merkle tree generated from all transactions in a transaction block.
    \item $\mathtt{tmy\_merkle\_root}$: the root of a Merkle tree generated from all testimonies in a transaction block.
\end{itemize}

Adopting the UTXO model, a transaction in Manifoldchain consist of multiple inputs and outputs. Each input contains a data field called $\mathtt{payer\_addr}$, signifying the address of the payer. Similarly, each output contains a data field named $\mathtt{receiver\_addr}$, denoting the receiver's address. Basically, transactions are categorized into two main types: domestic transactions (domestic-txs) and cross-txs. If all payers and receivers involved in a transaction belong to the same shard, the transaction is termed a domestic-tx; conversely, if participants span multiple shards, it is referred to as a cross-tx. A transaction block resembles transactions is referenced to a consensus block through the $\mathtt{tx\_merkle\_root}$. Within a transaction block each cross-tx is accompanied by an additional testimony. This structure offers an efficient method for verifying cross-txs, presented in Section \ref{cross-tx method}. 

\vspace{-1mm}
\subsection{Sharing Mining}\label{sharing_mining}

% We employ the BCSF mechanism to group miners with similar bandwidths into the same shard. While honest hashing power is evenly distributed across all shards, the adversary could concentrate its hashing power on specific shards, potentially leading to an adversarial majority in those shards.

Using the BCSF shard formation mechanism, honest hashing power is evenly distributed across all shards, but adversary could concentrate its hashing power on specific shards. To ensure security in shards with adversarial majority, we propose sharing mining to allow honest miners to contribute their hashing power to all shards. An initial design of sharing mining, inspired by Chu-ko-nu mining, allows a miner to extend chains in all shards simultaneously using a single PoW solution. Instead of mining on one parent block, a miner works on creating a consensus block using the hash of highest blocks from all shards. Adopting 2-for-1 PoW, the consensus block is then classified as either an exclusive block or an inclusive block. In the case of an exclusive block, it extends the longest chain within the miner's corresponding shard. Conversely, in the case of an inclusive block, it simultaneously extends multiple longest chains across all shards. Honest miners treat exclusive blocks and inclusive blocks equivalently, adhering to the same verification and fork selection mechanisms. Fig.~\ref{fig:Chu-ko-nu} shows how sharing mining works. 
% We demonstrate how honest hash power in one shard contributed to maintain security in other shards. Consider a scenario where the ratios of adversarial hashing power invested into mining inclusive blocks across $m$ different shards are denoted as ${\alpha_0, \alpha_1, ..., \alpha_{m-1}}$. Here, it holds that $\frac{1}{m}\sum_0^{m-1}\alpha_i = 1-\rho$. The mining rate of inclusive blocks within each shard is denoted as $\lambda_I$, while the mining rate of exclusive block within shard $i$ is denoted as $\lambda_E$. Given that an inclusive block is extended on the longest chains of all shards, adversaries can increment a chain at a rate of $\alpha_i\cdot \lambda_E+ \lambda_I\cdot\sum_0^{m-1}\alpha_i = \lambda m(1-\rho)$, while honest miners can extend a chain at a rate of $(1-\alpha_i)\lambda_E + \lambda_I\cdot\sum_0^{m-1}1-\alpha_i = \lambda_I m\rho$. Consequently, the ratio of adversarial hashing power within shard $i$ can be expressed as $\frac{\lambda m(1-\rho)}{\lambda m(1-\rho) + \lambda m\rho} = 1- \rho$, aligning consistently with the global ratio of adversarial hashing power observed across all shards. \xw{The analysis does not make much sense to me. Why does the adversary need to split its power on inclusive blocks?}{\old rewrite, explain the security is hold even all adversarial miners in one shard}

\begin{figure}
  \centering
  \includegraphics[width=7cm]{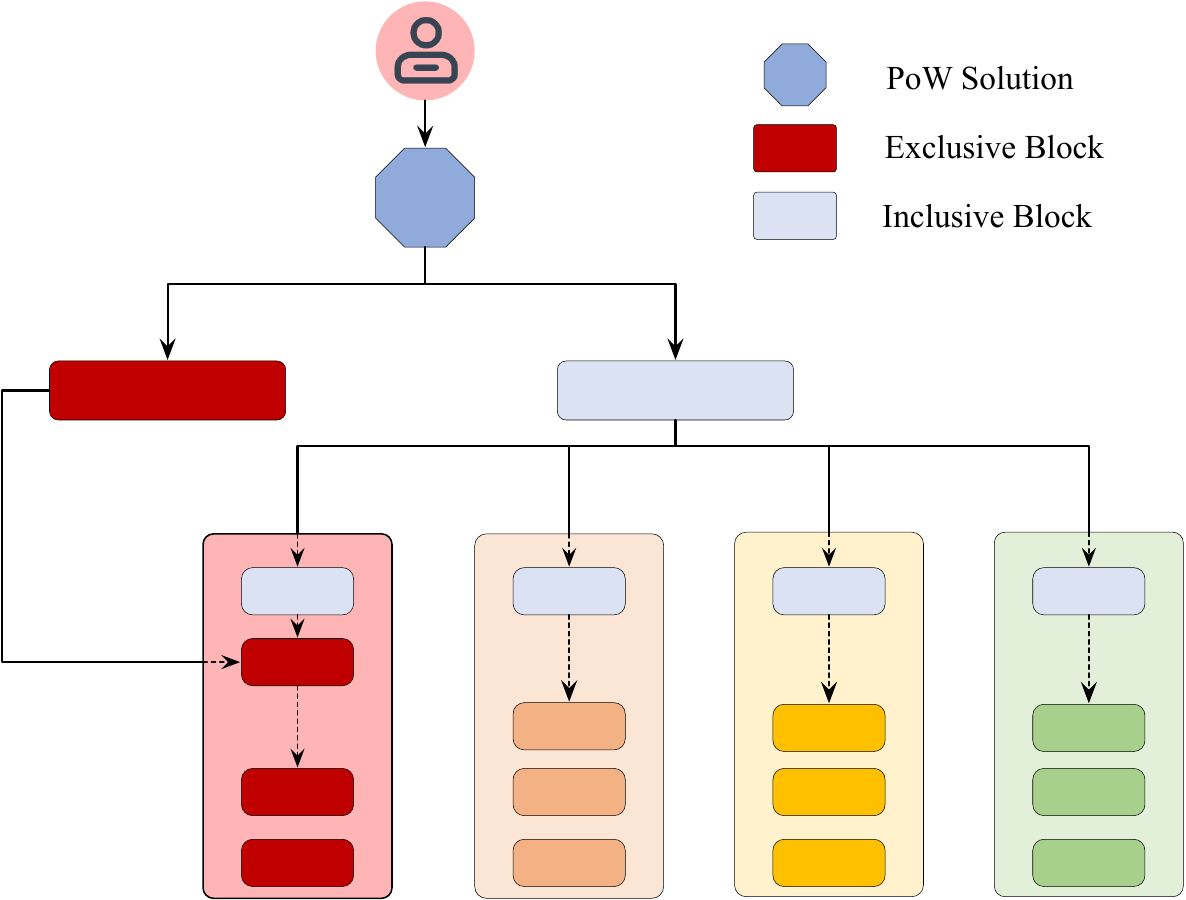}
  \vspace{-2mm}
  \caption{Overview of sharing mining. An exclusive block extends the blockchain within a specific shard, whereas an inclusive block extends all blockchains across shards.}
  \label{fig:Chu-ko-nu}
  \vspace{-6mm}
\end{figure}

\noindent {\bf Hashing power splitting attack.} As the transaction block associated with a consensus block is not sent to miners outside the shard, it causes the protocol vulnerable to a Hashing Power Splitting Attack (HPSA). Specifically, an attacker with overwhelming hashing power within a shard may mine a longer chain, containing invalid transactions but appearing valid from the perspective of honest miners in other shards, as they cannot verify the validity of the block without the transactions. This can cause honest miners to mine on different consensus blocks, leading to split of honest mining power and hence reduced security. A detailed description of HPSA is given in the Appendix \ref{HPSA}. 
% However, the initial design does not consider the data availability and transaction validity issues. In a sharding protocol, miners only store transactions within their shards. Consequently, they only verify the PoW solution of the inclusive blocks from other shards but lack verification of the corresponding transaction blocks. 

To counter HPSA, Manifoldchain requires each miner to verify the following properties of each block from foreign shards, without storing the corresponding transaction block.
% Basically, the validity of a transaction block involves two properties:
\begin{itemize}
    \item \textbf{Transaction validity} ensures that all transactions packaged within a block adhere to the appropriate format and preclude any double-spending events.
    \item \textbf{Data availability} requires a full block including both consensus block and associated transaction block, appears in each honest miner's view within its shard. 
\end{itemize}

% We refer the miners in a specific shard as in-shard miners, while other miners are referred as out-shard miners. Sharing mining provides a proof mechanism to enable out-shard miners to verify transaction blocks proposed by in-shard miners if at least one honest miner exists in each shard. 
We provide corresponding verification mechanisms for these two properties in scenarios where there is at least one honest miner in each shard. To ensure transaction validity, we employ fraud proof~\cite{al2018fraud} to enable in-shard miner to inform out-shard miner of a block's violation of validity rules. {\revise When an honest in-shard miner encounters an invalid block, it generates a fraud proof and broadcast it among all out-shard miners. Upon receiving a valid fraud proof, honest out-shard miners reject the block.} 

For data availability, we utilize CMT~\cite{DBLP:conf/fc/YuSLAKV20} to enable miners to verify data availability without downloading the full block. Specifically, CMT adds redundancy to transaction blocks via erasure codes. Any honest miner can verify full data availability through the mere download of a block hash commitment with a size of $\mathcal{O}(1)$ byte, combined with a random sampling of $\mathcal{O}(\log(B))$ bytes. Therefore, out-shard miners can validate the data availability of an in-shard transaction block by requesting just $\mathcal{O}(\log(B))$ samples. {\revise When an honest out-shard miner receives a foreign block, it requests block samples from the source shard. If it receives any error-coding proofs or fails to receive all requested samples within the maximum network delay $\Delta$, it rejects the block.} Detailed verification mechanisms are provided in Appendix~\ref{validity_availability_verification}.

% Verifying data availability and transaction validity requires at least $\Delta$ seconds to complete. During this time, there might be forks extending from the same parent as the unverified blocks. We refer to forks containing unverified blocks as unverified forks, 

% and these unverified forks can be utilized by adversary to perform HPSA (detailed attack strategy is presented in Appendix \ref{HPSA}). 

% This makes the protocol vulnerable to a specific type of attack that we refer to as Hashing Power Splitting Attack (HPSA). Specifically, an attacker with overwhelming hashing power may mine a longer chain, containing invalid transactions but appearing valid from the perspective of honest miners in other shards, as they only receive the block headers. This effectively splits the hashing power between miners within a specific shard and those outside it, potentially leading to severe security vulnerabilities. We present a detailed description of HPSA in the Appendix \ref{HPSA}. 

Verifying data availability and transaction validity can take up to the maximum network delay $\Delta$ to complete, diminishing the benefit of Manifoldchain by clustering miners with similar bandwidths. 
% During this time, there might be forks extending from the same parent as the unverified blocks. We refer to forks containing unverified blocks as unverified forks, 
To maintain the throughput gain from downloading fraud and data availability proofs from foreign shards, we further optimize sharing mining, by proposing novel techniques of \emph{predictive mining} and \emph{fork pruning}. The key idea is for the miner to simultaneously mine on multiple blocks, with some of them waiting to be verified by messages from other shards. By doing this, a miner can start mining on a foreign block as soon as receiving the block, without needing to wait for verification. 
% The initial version extends on the longest chain within each shard, whereas the optimized version operates on all unverified forks, allowing simultaneous execution of PoW and verification processes. 
Upon completion of the verification process, a pruning procedure is conducted to maintain the uniqueness of the ledger. Consequently, any forks that do not extend the the longest verified chain are discarded.

\noindent\textbf{Predictive mining}. An exclusive block has multiple parents within one shard, comprising the last blocks of the longest verified chain and all unverified forks that extend on the longest verified chain. This exclusive block temporarily extends all chains that are possibly the longest verified chain. An inclusive block has multiple parents across all shards. Within each shard, similar to exclusive block, its parents comprise the last blocks of the longest verified chain and all unverified forks in that shard. As illustrated in Fig.~\ref{fig:predictive_pruning}, three unverified blocks extend from the second verified block. By employing predictive mining, the subsequent block will consider all three unverified blocks, along with the second verified block, as parents and extend its chain on top of them.

\begin{figure}
  \centering
  \includegraphics[width=8cm]{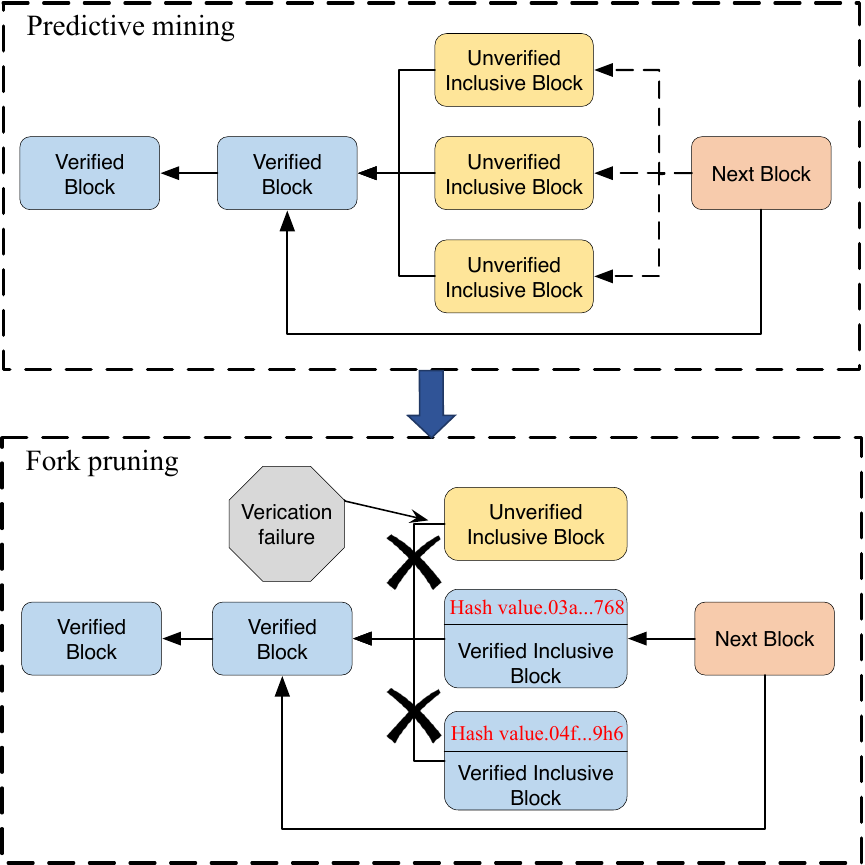}
  \vspace{-2mm}
  \caption{Key insight of predictive mining and fork pruning. Blocks can extend potential parents and invalid forks are subsequently pruned.}
  \label{fig:predictive_pruning}
  \vspace{-5mm}
\end{figure}

We present the procedures for mining exclusive and inclusive blocks, along with their verification process as follows:
\begin{itemize}
    \item \textbf{Mining}: a miner gathers all requisite segments for generating a consensus block. Initially, it package valid transactions and testimonies into a transaction block and generate the corresponding $\mathtt{tx\_merkle\_root}$ and $\mathtt{tmy\_merkle\_root}$. As for other segments, $\mathtt{verified\_parent}$ is designated as the hash value of the last block of the longest verified chain; $\mathtt{inter\_parents}$ is a vector containing $\mathtt{verified\_parent}$ and the hashes of the last blocks of all unverified forks that are at least as long as the longest verified chain; The $\mathtt{global\_parents}$ is formed through the aggregation of the $\mathtt{inter\_parents}$ across all shards. The miner works on finding a solution for $\textbf{PoW}^{\sigma}$$(\mathtt{parent\_hash}$, $\mathtt{info}$,$ \mathtt{nonce})$. In this context, $\mathtt{parent\_hash}$ refers to the hash of the three segments comprising $\mathtt{verified\_parent}$, $\mathtt{inter\_parents}$, and $\mathtt{global\_parents}$. All other segments, excluding $\mathtt{nonce}$, undergo hashing to formulate $\mathtt{info}$. The generated consensus block is subsequently categorized as an exclusive block or inclusive block based on its hash value, and it is linked to the transaction block. Both exclusive blocks and inclusive blocks are sent to all miners, while transactions blocks are only broadcast within their corresponding shards. %\sz{these blocks do not contain transactions, right? How transaction blocks are mined?}
    \item \textbf{Verification}: when a miner receives an exclusive or inclusive block from other shards, it first checks the correctness of the PoW solution by recomputing the hash value. If the solution is correct, the miner examines each parent included in the block. If a parent is verified invalid or does not exist in all possible forks, the miner gives up extending the block on this parent. Otherwise, it extends the block on the parent, even if it is not part of the longest chain (storing an exclusive or inclusive block incurs negligible storage as it does not contain any transactions). Conversely, if a miner receives an exclusive or inclusive block belonging to its shard, in addition to the aforementioned verification procedure, it also verifies the validity of the linked transaction block. For each domestic transaction, it conducts standard verification akin to Bitcoin's process. For each cross-shard transaction, it follows the verification method outlined in Section \ref{cross-tx method}. The miner rejects blocks containing any invalid transactions.
\end{itemize}

\noindent\textbf{Fork pruning}. 
Once an unverified block is successfully verified with both transaction validity and data availability, some forks are temporarily considered invalid and subsequently excluded from the parent blocks during the mining of the next block. Pruning is employed under the following two situations:

\begin{enumerate}
    \item \textbf{Verification failure}. If a block fails the data availability verification (a miner has not received enough samples to complete the availability verification of a block within $2\Delta$ seconds since receiving the block) or the transaction validity verification fails (a miner receives a valid $\mathtt{proof\_of\_invalidity}$ of that block within $2\Delta$ seconds since receiving the block), any forks extending from this block are pruned. 
    \item \textbf{Sibling block with lower hash value is verified}. In cases where multiple blocks on the same level attain verification, the block possessing the lowest hash value stays, while its sibling blocks are subsequently pruned. 
    % \item \textbf{Higher block is verified}. Upon the verification of a block, all blocks with heights lower than its height are pruned, which necessitated that each valid fork must grow on top of the longest verified chain. For instance, in the fork situation depicted in Fig. \ref{fig:fork_pruning}, if block 4 is verified, the fork $2 \Leftarrow 5$ is pruned. 
\end{enumerate}

% {\reviewer 
% Reviewer 1: The basic idea behind this paper is to reduce forks to guarantee security and the mining speed is limited by slow miners. However, the sharing mining strategy requires multiple chains to be extended simultaneously, thus introducing extra forks.

% Reviewer 5: The main technical concern I have is: does shared mining increase the chance of forking due to the fact that it allows extending multiple chains? If so, it somewhat conflicts with the motivation of the paper -- avoiding unnecessary forking.
% }
% {\old Notice that all forks only exit for a maximum interval of $\Delta$ seconds. After this interval, only one fork remains, indicating that if we solely consider the verified portion of a chain, our sharing mining protocol aligns with the insights of the NC protocol.}

\noindent\textbf{Forking rate about sharing mining}.
Even though sharing mining enables both exclusive blocks and inclusive blocks to extend on multiple parents, it does not produce any extra unexpected forks. (1) An exclusive block with multiple parents results in extra forks, but these forks are expected and only exist for a maximum interval of $\Delta$ seconds. After this interval, only one fork remains, indicating that if we solely consider the verified portion of a chain, our sharing mining protocol aligns with the insights of the NC protocol. (2) An inclusive block does not lead to any additional forks. Upon finding a PoW solution for an inclusive block, a miner generates both the inclusive block and its associated transaction blocks. The inclusive block solely contains essential header information and is negligibly small compared to the transaction block. Only the inclusive block is transmitted to other shards, while the transaction block remains within the shard. The transmission delay for inclusive blocks is negligible, thereby preventing the occurrence of extra forks when transmitting inclusive blocks across shards.

\subsection{Cross-tx Atomicity}\label{cross-tx method}

In the context of a sharding protocol, miners across different shards must collectively reach consensus on whether a cross-tx should be committed or not. This scenario aligns with the classic atomic commitment problem in distributed databases~\cite{DBLP:conf/pods/KeidarD95, DBLP:conf/itng/ParkLY13}. As discussed in Section~\ref{formation}, transactions are either domestic-txs or cross-txs. A cross-tx involves multiple inputs and outputs originating from different shards. %In this context, 
Here ``input'' signifies the expenditure of coins, while ``output'' denotes the creation of coins. We refer to shards in which coins are spent as ``input shards'', and shards in which coins are created as ``output shards''. The atomic commitment across shards requires that
%manifests in the consistency of inputs and outputs. Precisely, 
if all inputs are confirmed as spent, all outputs should be successfully created. Conversely, if any input is confirmed as invalid, all outputs should not be generated and other inputs are not spent eventually. The most straightforward and widely recognized approach for achieving atomic commitment is 2PC, a method employed by existing blockchain protocols such as RSCoin\cite{DBLP:conf/ndss/DanezisM16} and OmniLedger. Specifically, the coordinator node prompts participants to vote, and if all vote to commit, broadcast a commit message; otherwise, broadcast an abort message if any votes to abort. 

New challenge arises from sharding setting with various mining rates when naively applying 2PC to Manifoldchain: The vote messages may become outdated when forking occurs. For instance, a shard initially votes to commit a cross-tx, but the inputs may become invalid 
%disappear from the longest chain 
upon a forking occurrence. An intuitive solution is to delay the decision phase until the confirmation of all vote messages, which is adopted by OmniLedger. However, this solution requires miners in fast shards to pause the decision phase before receiving all confirmed vote messages, resulting in unacceptable cross-tx latency and contradicting our initial motivation.
% \item Corrupted miners may send misleading vote messages to interfere with the commitment, miners must verify these vote messages to check whether the input coins are successfully locked. In the context of sharding protocol, this can not be directly achieved as miners lack access to the detailed ledger information in other shards.
To address this challenge, we implement an asynchronous commitment mechanism by eliminating the locking phase. Initially, coins are expended directly rather than being locked if input shards vote to commit. Upon the successful expenditure of all inputs, the corresponding coins will be subsequently generated in the output shards. Conversely, if any input shard votes to abort or any vote message becomes outdated, the associated coins will not be generated, and the coins will be refunded to the input shards with prior successful expenditure. 

We introduce testimony to enable miners within the output shards to ascertain the successful expenditure of inputs. We implement the execution of a cross-tx by broadcasting it to each input shard and output shard. Cross-txs within the input shards are designated as input-txs, whereas those situated within the output shards are labeled as output-txs. %Input-txs correspond to the ``expenditure'' operation of coins, while output-txs correspond to the ``receipt'' operation. 
Each output-tx is associated with a testimony to enable miners to verify the validity of its corresponding input-txs asynchronously. A testimony is composed of multiple units, each corresponding to an input.
% with each testimony unit corresponding to an input of the output transaction. 
Essentially, each testimony unit maintains a Merkle proof for corresponding input-tx that proves its inclusion within the longest chain. %A testimony resembles multiple units and the hash value of the associated output-tx. 
Specifically, the structure of a testimony unit is presented as follows:

\noindent \textbf{Input hash}. The hash of the corresponding input. 

\noindent \textbf{Originate block hash}. The hash of the originate block where the input-tx is packaged.

\noindent \textbf{Transaction Merkle proof}. A merkle proof which proves the inclusion of the input-tx in the originate block.

\noindent \textbf{Vote message}. An additional bit indicates an acceptance or rejection.

% \begin{itemize}
%     \item \textbf{Input hash}. The hash of the corresponding input. 
%     \item \textbf{Originate block hash}. The hash of the originate block where the input-tx is packaged.
%     \item \textbf{Transaction Merkle proof}. A merkle proof which proves the inclusion of the input-tx in the originate block.
%     \item \textbf{Vote message}. An additional bit indicates an acceptance or rejection.
% \end{itemize}

When a miner generates a new block, for each input-tx, it generates the corresponding testimony and sends it to all involved output shards. As the inputs of a cross-tx are distributed among different shards, each shard can create a partial testimony encompassing units for a portion of inputs. After receiving all partial testimonies, miners in the output shards combine them into one full testimony, composed of corresponding units for all inputs. The full protocol is illustrated as follows:
\begin{enumerate}
    \item \textbf{Initialization}. A user creates a cross-tx and gossips it to all involved input shards and output shards.
    \item \textbf{Expenditure}. When a miner in the input shards receive an input-tx, it proceeds as follows. First, it decides whether the coins can be spent based on the ledger information. If the input-tx is valid, the miner {\rrrevise labels} this input-tx as accepted and includes it in the mining pool. In the case of an invalid input-tx, it is labeled as rejected but is still placed in the mining pool. Upon successfully mining a valid block, for each involved accepted or rejected input-tx, the miner generates a testimony containing a vote message of corresponding acceptance or rejection and broadcasts it to all output shards. {\rrrevise Aligning with its label, an input-tx is confirmed as either accepted or rejected upon reaching a $\kappa$-confirmation (the block containing the input-tx is followed by $\kappa$ consensus blocks).}
    \item \textbf{Receipt}. When a miner in the output shards receives an output-tx along with all associated testimonies for all inputs, it is included in the mining pool. Output-txs are subsequently packaged into blocks and appear in the longest chain. For each output-tx in the longest chain, if all associated inputs are confirmed in other shards, miners validate their validity by verifying the corresponding testimonies. If all input-txs {\rrrevise are confirmed as} accepted, the output-tx is {\rrrevise labeled as} accepted. Conversely, if any input-txs {\rrrevise are confirmed as} rejected or testimonies are found to be invalid, the output-tx is {\rrrevise labeled as} rejected, triggering the subsequent refunding process. {\rrrevise When an output-tx reaches a $\kappa$-confirmation, it is confirmed as either accepted or rejected in accordance with its label.}
    %\begin{itemize}
        %\item 
        
    \textbf{Refund}. Upon the {\revise $\kappa$-confirmation} of a rejected output-tx and all corresponding input-txs, miners generate a refund transaction (refund-tx) along with an \textit{proof-of-rejection}. This proof-of-rejection contains two testimonies $\mathtt{tes}_{I}$ and $\mathtt{tes}_{O}$. $\mathtt{tes}_O$ proves the rejected output-tx is included in the corresponding output shard's longest chain. $\mathtt{tes}_I$ testify either the inclusion of a rejected input-tx in the longest chain of the corresponding input shard, or the inclusion of an input-tx in a deconfirmed block (where another block at the same height is confirmed) of the input shard. Subsequently, both the refund-tx and proof-of-rejection are broadcast to all input shards. Upon receiving a refund-tx and the proof-of-rejection from all output shards, miners in the input shards can verify the validity of the refund-tx through the proof-of-rejections. If all {\rrrevise output-txs are confirmed as rejected}, the refund-tx is {\rrrevise labeled as} accepted, miners include it in the blockchain to refund the coins to the input shards. {\rrrevise Similarly, a refund-tx is confirmed as neither accepted or rejected based on its label upon a $\kappa$-confirmation.}
    %\end{itemize}
\end{enumerate}

{\revise
This protocol ensures cross-tx atomicity across shards, and the detailed proof is provided in Appendix~\ref{cross_tx_atomicity}. Intuitively, honest miners confirm an output-tx only if all input-txs are confirmed. If an adversary attempts to disrupt the atomicity of a cross-tx, they must target one of the input-txs and include a malicious input-tx with a different decision in a block. They must then ensure that this block is confirmed at the same level as the block containing the original input-tx. However, our sharing mining protocol ensures that, with overwhelming probability, no two different blocks at the same level are confirmed by honest miners, thereby preventing the adversary from breaking atomicity.

We point out that cross-tx latency depends on the slowest shard, as output-txs cannot be confirmed before all input-txs are confirmed. This bottleneck, shared by other sharding protocols, arises from the need for atomicity: decisions require votes from all participants. However, when cross-txs are not concentrated in the slowest shard, our protocol achieves lower average latency than baseline sharding protocols, as shown in Section~\ref{tps_real_testbed}.
}

\section{Security and Scalability Analysis}\label{main_security_analysis}

\subsection{Honest Presence}\label{honest_presence}
% {\old The security of Manifoldchain relies on the presence of at least one honest miner in each shard, which is shown to hold in the following lemma:}
{\revise The security of Manifoldchain relies on the presence of at least one honest miner in each shard. Recall that $\alpha_i$ represents the fraction of new honest miners joining during the $i$-th shard formation phase, the following lemma demonstrates that Manifoldchain holds an honest presence within each shard.}

{\revise
\begin{lemma}\label{honest_presence_proof}
Let $\underline{\alpha} = \min_i \alpha_i$, i.e., the minimum fraction of newly joined miners among all miners in a shard formation phase. Given parameters $1 \leq S_Y \leq \max\{\underline{\alpha}\rho N, \lceil \frac{1}{1-\rho}\rceil - 1\}$ and $S_X \geq 1$, our shard formation mechanism ensures that each shard contains at least one honest miner, except with a negligible error probability bounded by $\varepsilon$, where
\begin{equation}\label{honest_presence_probability}
    \begin{split}
        \varepsilon = 
        \begin{cases}
        (S_XS_Y-1)(\frac{S_XS_Y-1}{S_XS_Y})^{\underline{\alpha} \rho N-1} & S_Y \geq \frac{1}{1-\rho}\\
        (S_XS_Y-1)(\frac{S_XS_Y-1}{S_XS_Y})^{\rho N-1} & 0 < S_Y < \frac{1}{1-\rho}
        \end{cases}
    \end{split}
\end{equation}
% \begin{small}
% \begin{empheq}[left={\varepsilon =\empheqlbrace}]{align}
%     (S_XS_Y-1)(\frac{S_XS_Y-1}{S_XS_Y})^{\underline{\alpha} \rho N-1}\quad & S_Y \geq \frac{1}{1-\rho}\label{honest_presence_eq_a}\\
%     (S_XS_Y-1)(\frac{S_XS_Y-1}{S_XS_Y})^{\rho N-1}\quad & 0 < S_Y < \frac{1}{1-\rho}\label{honest_presence_eq_b}
%   \end{empheq}
% \end{small}
% \begin{numcases}{\varepsilon =}
%     $(S_XS_Y-1)(\frac{S_XS_Y-1}{S_XS_Y})^{\underline{\alpha}\rho N-1}$ & $S_Y\geq \frac{1}{1-\rho}$\label{honest_presence_eq_a}
%     \\
%     $(S_XS_Y-1)(\frac{S_XS_Y-1}{S_XS_Y})^{\rho N-1}$ & $0 <S_Y<\frac{1}{1-\rho}$\label{honest_presence_eq_b}
% \end{numcases}

\end{lemma}
}

We prove this lemma in Appendix~\ref{honest_presence_section}, where we also show that even under an imperfect bandwidth estimation with small deviation, honest presence still holds in each shard except with a negligible error probability for overwhelmingly large $N$. 

% {\revise 
% In Lemma~\ref{honest_presence_proof_1}, the number of shards $m$ is constrained by the minimum number of new honest miners per shard formation phase, $\underline{\alpha} \rho N$. We assume $\underline{\alpha} \rho N$ is generally sufficient to support a large $m$. This assumption is not present in conventional sharding protocols, which use USF to rotate all miners. Manifoldchain, using BCSF, rotates miners only within Y-regions. If a Y-region is small, the adversary can completely corrupt it. Therefore, we require a certain amount of new honest miners to join and be distributed across shards to ensure honest presence in each shard. However, we argue that this assumption is not overly stringent. If this assumption does not hold, we can set $S_Y=1$ and uniformly distribute all miners across $S_X$ shards, effectively causing BCSF to degenerate into USF. Even in the least favorable scenario, Manifoldchain can support a higher $m$ than conventional sharding protocols, which have to limit $m$ to ensure an honest majority within each shard. Manifoldchain requires only one honest miner per shard, alleviating this constraint.
% }

{\rrrevise
We meticulously follow the selection strategy presented in Section~\ref{formation} to choose $S_X$ and $S_Y$, ensuring that Lemma~\ref{honest_presence_proof} holds--each shard contains at least one honest miner with high probability. Specifically, in a scenario with a sufficient number of new honest miners, i.e., $\underline{\alpha} \rho N > \lceil \frac{1}{1-\rho} \rceil - 1$, our shard partition mechanism ensures that each Y-region has approximately $\frac{\underline{\alpha} \rho N}{S_Y}$ new honest miners, which are uniformly distributed across $S_X$ shards. Therefore, each shard averages $\frac{\underline{\alpha} \rho N}{S_X S_Y}$ new honest miners, and Lemma~\ref{honest_presence_proof} holds with high probability, by Chebyshev's inequality. Conversely, in a scenario with few new honest miners, i.e., $\underline{\alpha} \rho N \leq \lceil \frac{1}{1-\rho} \rceil - 1$, we set $S_Y = \lceil \frac{1}{1-\rho} \rceil - 1$ and uniformly distribute all honest miners in each Y-region across $S_X$ shards. Similarly, each shard averages $\frac{\rho N}{S_X S_Y}$ honest miners, and Lemma~\ref{honest_presence_proof} holds with high probability.
}

\subsection{Consistency and Liveness}

% When Condition \ref{honest_condition} holds, $S_Y \leq \rho N$, $G \leq \frac{\rho N}{S_Y}$, $S_X \leq G$, thus $m=S_X\times S_Y \leq \rho N$. 
Next, we demonstrate that given honest presence in each shard, Manifoldchain attains three basic security properties: Common Prefix (CP), Chain Quality (CQ), and Chain Growth (CG), formally defined in Appendix~\ref{security_property}. The CP property implies consistency, while the combined CQ and CG properties imply liveness. Following the tradition set by Nakamoto~\cite{nakamoto2008bitcoin}, we adopt a continuous-time model and represent PoW mining as a homogeneous Poisson point process. This model was also used in several subsequent influential works~\cite{sompolinsky2015secure,ren2019analysis,li2020continuous,dembo2020everything}. Let $\lambda$ be the total mining rate of the network, specifically denoting the number of blocks generated within one second. Honest block arrivals and adversarial block arrivals as two independent Poisson processes with rates $\rho \lambda$ and $(1-\rho)\lambda$, respectively. %For a more detailed subdivision, 
We use $\lambda_i$ to denote the mining rate of exclusive blocks within shard $i$. As inclusive blocks extend on the longest chains in all shards, there exists only one mining rate for inclusive blocks spanning all shards, denoted as $\lambda_s$. We denote $\rho_i$ as the honest participation ratio within shard $i$, and $\Delta_i$ as the maximum network delay observed within shard $i$. {\new Below we present the main security result of Manifoldchain.}

\begin{theorem}
\label{informal_cp_property}
{\new If there is at least one honest miner in each shard}, CP, CG, and CQ hold for Manifoldchain regardless of the adversary's attack strategy within each shard $i$, as long as 
\begin{equation}
\begin{split}\label{main_eq_pi}
p_i = \frac{\lambda_i \rho_i + m \lambda_s \rho}{\lambda_i + m\lambda_s} e^{-(\lambda_i + \lambda_s)\Delta_i} > \frac{1}{2},
\end{split}
\end{equation}

The CP property is hold except with a negligible error probability $\varepsilon(\kappa)$ bounded by
\begin{equation}
\begin{split}\label{main_error_probability}
\varepsilon(\kappa) \leq (2 + 2\sqrt{\frac{p_i}{1-p_i}})(4p_i(1-p_i))^{\kappa}.
\end{split}
\end{equation}

For any $\delta_i >0$, $\delta_i' > 0$, the chain growth rate $\mathtt{g}_i$ and chain quality rate $\mathtt{q}_i$ satisfy
\begin{equation}
\begin{split}
\mathtt{g}_i \geq& (1 - \delta_i)\frac{p_i(\lambda_s + \lambda_i)}{1 + \Delta_i(\lambda_s + \lambda_i)},\\
\mathtt{q_i} \geq& 1 - (1+\delta_i')\frac{1 + \Delta_i p_i\rho_i(\lambda_s + \lambda_i)}{p_i},
\end{split}
\end{equation}
with respective error probabilities of $\varepsilon(\delta_i)$ and $\varepsilon(\delta_i')$.
\end{theorem}

% \begin{figure}
%   \centering
%   \includegraphics[width=8cm]{./img/theory.png}
%   \vspace{-2mm}
%   \caption{Shards with low delays achieve higher mining rates. In this setting, $D=10s$, $\rho=\frac{3}{4}$, $\kappa=50$, and $m \gg 5$.}
%   \label{Fig:theory}
%   \vspace{-5mm}
% \end{figure}

{\revise We provide the detailed proof in Appendix \ref{proof_of_security_properties} and present the roadmap of our proof in Fig.~\ref{fig:roadmap}.} Specifically, we shift from a standard to a \textit{severe execution environment}, as defined in Definition~\ref{severe_execution_environment}, where any delivered full block experiences the maximal network delay, and during the block delivery any mining operation is considered adversarial. Achieving security in severe environment is more challenging than in standard environment. We then establish by Theorem \ref{private_mining_worse} in Appendix \ref{severe_execution_environment_section}, that the private-mining attack is the optimal strategy {\revise in a severe execution environment}. Additionally, we calculate the upper bound on mining rates necessary to maintain CP property against the private-mining attack, serving as a constraint against all potential attacks {\revise in all environments}. Towards CG and CQ properties, we obtain the upper bounds of chain growth speed and honest block fraction based on these bounded mining rates. {\revise By leveraging Lemma~\ref{honest_presence_proof} in conjunction with the CP, CG, and CQ properties, we substantiate the proof of Theorem~\ref{informal_cp_property}.}
%Regarding the CG property, we initially demonstrate that the chain growth speed in a standard execution environment is not lower than that in a severe execution environment. Subsequently, we compute the lower bound of Manifoldchain's chain growth speed in the severe execution environment, which also represents the lower bound in the standard execution environment. As for the CQ property, we compute the maximum fraction of adversarial blocks within a block sequence flanked by two honest blocks. Subtracting this fraction from 1 yields the lower bound of the honest block fraction.

\begin{figure}
  \centering
  \includegraphics[width=7cm]{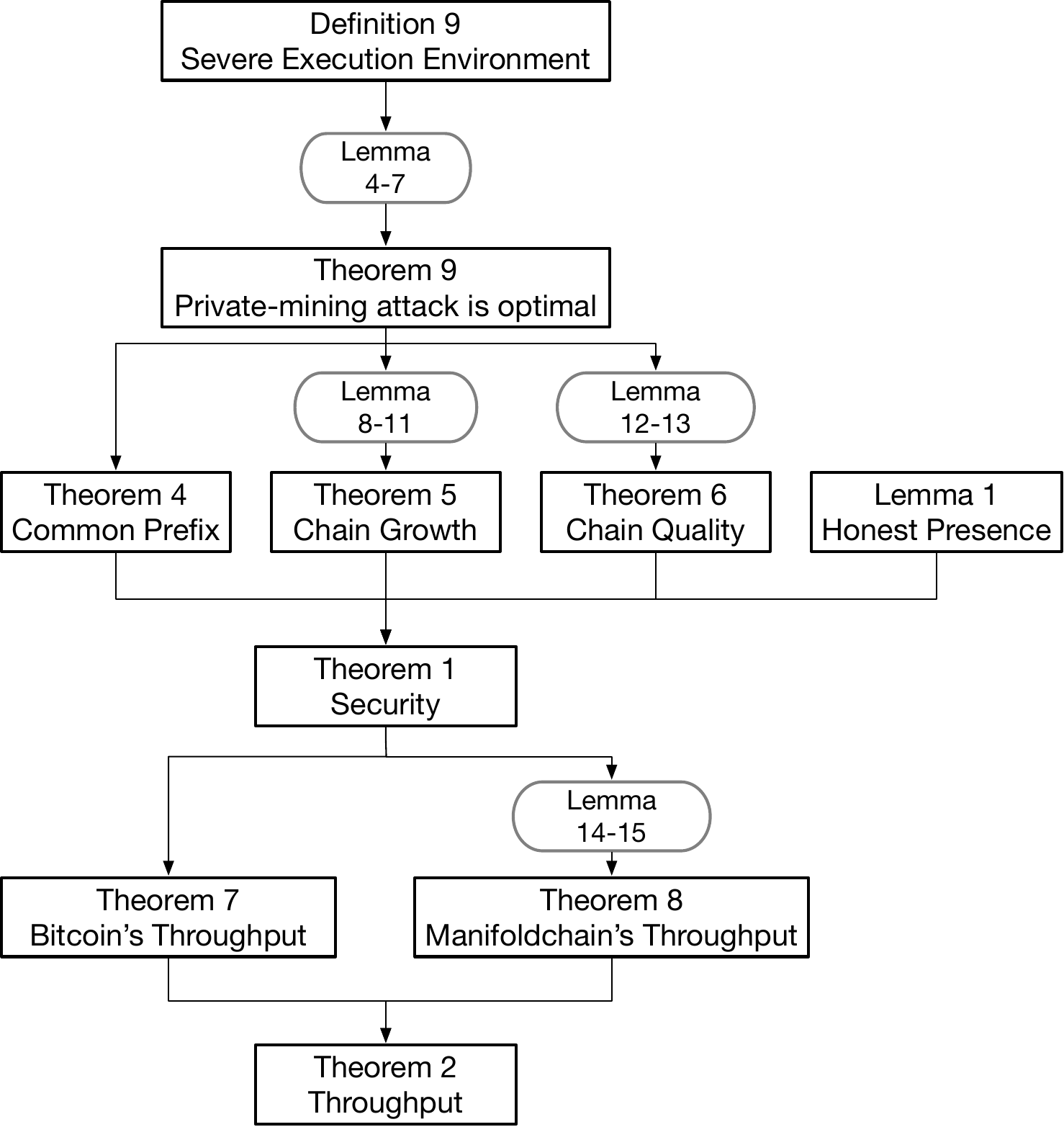}
  \vspace{-2mm}
  \caption{\revise The proof roadmap of Theorem~\ref{informal_cp_property} and~\ref{informal_tps_theorem}.}
  \label{fig:roadmap}
  \vspace{-5mm}
\end{figure}

This theorem suggests that shards experiencing lower delays can achieve faster mining rates compared to those with higher delays, while still maintaining the same error probability. Assuming all shards share the same $\kappa$, according to Eq.~\ref{main_error_probability}, there exists a one-to-one mapping between $\varepsilon(\kappa)$ and $p_i$. An increase in $\Delta_i$ also results in a decrease in $p_i$ and $\varepsilon(\kappa)$. Referring to Eq.~\ref{main_eq_pi}, the maximal values of $\lambda_s$ and $\lambda_i$ are constrained by $\varepsilon(\kappa)$. Therefore, shards with lower $\Delta_i$ can be configured with higher $\lambda_s + \lambda_i$.
% an increase in $\lambda_s$ and $\lambda_i$ leads to a decrease in $p_i$. Consequently, the maximal values of $\lambda_s$ and $\lambda_i$ are constrained by $\varepsilon(\kappa)$. From another aspect, , which means that, under the condition that all shards maintain a fixed error probability $\varepsilon(\kappa)$, shards with lower $\Delta_i$ can be configured with higher $\lambda_s + \lambda_i$. \xw{repeated sentence?} 
With respect to the liveness, the longest chain in shard $i$ achieves an average chain growth rate of $\frac{p_i(\lambda_s + \lambda_i)}{1 + \Delta_i(\lambda_s + \lambda_i)}$ and an average adversarial block ratio of $\frac{1 + \Delta_i p_i\rho_i(\lambda_s + \lambda_i)}{p_i}$. Consequently, it can be deduced from this theorem that the longest chains in faster shards possesses higher chain growth rates and chain quality ratios over slower shards.

\subsection{Optimal Throughput}\label{optimal_throughput_section}

\noindent {\bf Throughput-security trade-off on bandwidth estimation.} 
% {\old Given a fixed $m$ such that the error probability in Lemma~\ref{honest_presence_proof} is negligible, $S_Y$ can range from $1$ to $m$, and $S_X=m/S_Y$.} 
A larger $S_Y$ helps to better distinguish between fast and slow nodes, thereby enhancing throughput. However, increasing $S_Y$ relies on a more accurate bandwidth estimation. A large estimation error derived from the gap between realistic bandwidth distribution and the estimated one can lead to unevenly distributed honest miners across Y-regions, with some Y-regions having a lower number of honest miners than the average value. This uneven distribution may lead to a larger error probability in Lemma~\ref{honest_presence_proof}, as shown in Appendix \ref{honest_presence_section}.

\noindent {\bf Optimal mining rate configuration.} A key question arises on how to optimally configure $\lambda_s$ and $\lambda_i$ to maximize throughput while maintaining security established by Theorem~\ref{informal_cp_property}. Assuming $\varepsilon(\kappa)$ is regulated to less than $\varepsilon$ (We do not consider $\varepsilon(\delta_i)$ and $\varepsilon(\delta_i')$ here because they depend on the Chernoff bound parameters $\delta_i$ and $\delta_i'$, which are unrelated to the protocol) and $\kappa$ is fixed to $\kappa'$. The network delay within each shard is computed based on the lowest bandwidth within that shard. Specifically, $\Delta_i = \frac{\mathtt{B}}{\underline{C}_i} + \Delta_p$, where $\mathtt{B}$ denotes the block size, $\underline{C_i}$ represents the lowest bandwidth in shard $i$, and $\Delta_p$ is a constant denoting the propagation delay. As the exact honest ratio in a specific shard is generally unknown, we consider the worst case where $\rho_i = 0$, $\forall i$.
%We are unable to know the exact honest ratio in a specific shard so we assume $\rho_i = 0, \forall i$. 
% This question can addressed by formulating it as an optimization problem.:
Consequently, we construct the following optimization problem.
\begin{equation}\label{optimization}
    \begin{split}
        &\max_{\{\lambda_i\}_m, \lambda_s} \quad \sum_i^m(\lambda_i + \lambda_s)\\
    &\begin{array}{r@{\quad}r@{}l@{\quad}l}
    s.t. &p_i = \frac{ m \lambda_s \rho}{\lambda_i + m\lambda_s} e^{-(\lambda_i + \lambda_s)\Delta_i} > \frac{1}{2}, \\
        &(2 + 2\sqrt{\frac{p_i}{1-p_i}})(4p_i(1-p_i))^{\kappa'} \leq \varepsilon. \\
    \end{array}
    \end{split}
\end{equation}

The key challenge to solve this optimization problem lies in how to strike a balance between $\lambda_i$'s and $\lambda_s$, given that both contribute to the overall throughput. %\sz{Do we show this algorithm gives the optimal solution to (4)? Or this is only one way to choose $\lambda_i$ and $\lambda_s$? Need to state it precisely.}
We present the algorithm to find the optimal solution for this problem and the key insight here, deferring the detailed proof to Appendix \ref{mfd_tps_proof}.
\begin{itemize}
    \item Miners in the slowest shard only mine inclusive blocks (i.e., $\lambda_j=0$). According to Theorem \ref{informal_cp_property}, there exist trade-offs between $\lambda_i$/$\lambda_s$ and the error probability $\varepsilon(\kappa)$ within each shard. Crucially, the trade-off between $\lambda_s$ and $\varepsilon(\kappa)$ is more favorable than that between $\lambda_i$ and $\varepsilon(\kappa)$. In other words, for the same increment in $\lambda_s$ or $\lambda_i$, the latter results in a more pronounced increase in $\varepsilon(\kappa)$. Therefore, with a bounded $\varepsilon$, $\lambda_s$ can be set higher than $\lambda_i$. In shard $j$, $\Delta_j = \Delta$. Given all other known parameters, the maximum value of $\lambda_s$, denoted as $\lambda_s'$, is easily obtained through binary search.
    \item In other shard $i$ where $\Delta_i < \Delta$, the first constraint of (\ref{optimization}) is relaxed, allowing us to set a non-zero $\lambda_i'$ to achieve an improved throughput.
\end{itemize}

It is noteworthy that the slowest shard attains the same throughput as Bitcoin (proved in Appendix \ref{mfd_tps_proof}). Denote the throughput of Bitcoin as $T$,  we present the total throughput of Manifoldchain as follows: 

\begin{theorem}[Informal]\label{informal_tps_theorem}
    Assuming $\forall i, m \gg \frac{\lambda_i}{\lambda_s}$, in a scenario where Bitcoin achieves a throughput of $T$, Manifoldchain attains a throughput of $\frac{T\Delta}{\rho}\sum_{i}^m\frac{\rho_i}{\Delta_i}$, while maintaining the same level of security as Bitcoin.
\end{theorem}

{\revise 
We defer the detailed proof to Appendix \ref{scalability_analysis} and present the proof roadmap in Figure~\ref{fig:roadmap}. Specifically, we calculate the maximum throughput of Bitcoin and Manifoldchain while maintaining identical error probabilities as stated in Theorem~\ref{informal_cp_property}. By benchmarking Manifoldchain against Bitcoin, we substantiate the proof of Theorem~\ref{informal_tps_theorem}.
}.

Given $m \gg \frac{\lambda_i}{\lambda_s}$, in a general scenario where $\forall \, \rho_i=\rho$, Manifoldchain can achieve a throughput of nearly $\frac{\Delta}{\Delta_i}T$ within shard $i$. In an extreme scenario where the adversary concentrates its hashing power within some shards, and in other shards $\rho_i = 1$, Manifoldchain can achieve a throughput of $\frac{\Delta}{\rho\Delta_i}T$.

% \sz{There is a problem with this result. The analysis on the maximum m is missing. As we discussed, there should be a maximum m for sharding to be secure.}

\section{Experiments}
% \begin{figure*}
%     \centering\hspace{-3mm}
%     \subfigure[Throughput in Alibaba Cloud's ECS]{\includegraphics[width=0.27\textwidth]{./img/exper_1.png}}\label{fig:exper_1}\hspace{-5mm}
%     \subfigure[BCSF vs USF]{\includegraphics[width=0.27\textwidth]{./img/exper_2.png}}\label{fig:exper_2}\hspace{-5mm}
%     \subfigure[Vertical Scalability]{\includegraphics[width=0.27\textwidth]{./img/exper_3.png}}\label{fig:exper_3}\hspace{-5mm}
%     \subfigure[Horizontal Scalability]{\includegraphics[width=0.27\textwidth]{./img/exper_4.png}}\label{fig:exper_4}
%     \hspace{-5mm}
%     \caption{These four experiments demonstrate Manifoldchain's throughput on a realistic testbed, improvement achieved through BCSF mechanism, vertical scalability and horizontal scalability respectively.}
%     \vspace{-5mm}
% \end{figure*}

% \begin{figure}
%   \centering
%   \includegraphics[width=7cm]{./img/implementation.pdf}
%   \caption{The architecture of Manifoldchain implementation \xw{Network --> Networker} \xw{Validator is not a data structure, just a part of miner and networker. Simply remove it from the figure.}}
%   \label{fig:implementation}
% \end{figure}
% {\new
% In our evaluation, we answer the following questions.
% \begin{itemize}
%     \item What is the throughput and the latency of Manifoldchain in a realistic deployment?
%     \item Is Manifoldchain able to achieve better throughput compared with the state-of-the-art PoW-based sharding protocol?
%     \item How does the system scale to more miners and higher bandwidths?
% \end{itemize}
% }
% \subsection{Implementation}
We implemented a prototype instantiation of a Manifoldchain client in about 15,000 lines of Rust code~\cite{codeAnon}. In our implementation, a Manifoldchain client employs two parallel threads—one for mining (called \textit{Miner}) and another for delivering incoming blocks from the network (called \textit{Networker}). Both threads have access to two core data structures: \textit{Mempool} and \textit{Multichain}. Below, we outline the specific functionalities of each module:

\begin{itemize}
    \item \textbf{Mempool} receives verified transactions and testimonies from the network module and pairs testimonies with their associated transactions. Upon requested by the Miner module, returns matched transactions and testimonies.
    % \item \textbf{Validator} inherits all the validation functionalities, including format validation and double-spending validation from transactions and blocks. It requests necessary information from Multichain for validation. 
    % \item \textbf{Blockchain} maintains a ledger for one shard. It stores a tree for recording all historical forks, and updates states (available UTXO set) upon each insertion. Blockchain implements sufficient interfaces for block insertion, fork pruning, and obtaining chain information.
    \item \textbf{Multichain} provides the interface for reading and writing blockchains across all shards, with each individual blockchain maintaining a ledger for a single shard. 
    \item \textbf{Miner} retrieves transactions and testimonies from the Mempool, conducts double-spending check, and solves PoW to generate a valid block. Once a valid block is created, it is inserted into Multichain, and also immediately relayed to the Networker module for broadcasting.
    \item \textbf{Networker} manages inter-node communications, employing the Gossip protocol to broadcast and deliver messages. It handles three types of messages: transactions, testimonies, and blocks. Upon receiving messages, it validates their validity before further processing.
\end{itemize}

With this prototype, we conducted experimental evaluations on the following two different testbeds:

\begin{itemize}
    \item \textit{Real-world testbed}: We deployed Manifoldchain on Amazon Elastic Compute Cloud (EC2) to evaluate its performance in a genuine geo-distributed environment. Each miner operates within an EC2 $\mathtt{t3.medium}$ instance outfitted with 2 CPU cores, 4GB of RAM, and a 20GB NVMe SSD. {\new Importantly, these EC2 instances are positioned in 10 major cities geographically distributed worldwide, spanning Sydney, London, Virginia, California, Canada, Ireland, Mumbai, Sao Paulo, Stockholm, and Tokyo.} We deploy Manifoldchain in this real-world testbed to evaluate its actual throughput and latency.
    \item \textit{Simulated testbed}: Due to the complexity of the real-world testbed environment, we additionally deploy Manifoldchain in a simulated testbed to obtain an accurate and reproducible evaluation of its throughput under various bandwidth settings, while maintaining other configurations unchanged. {\new Specifically, we operate a Manifoldchain system on a single machine with 4 CPU cores, 16GB of RAM, and a 100GB NVMe SSD. We simulate the miners with different processes, adjusting the network bandwidths and delays using $\mathsf{tc}$ commands.}    
    % Without sacrificing authenticity, we deploy Manifoldchain within a single server, manipulating bandwidth and delay using specialized Linux commands. Specifically, we established distinct network conditions for each miner, setting unique routes with the $\mathsf{ip}$ command and defining bandwidth and delay parameters via the $\mathsf{tc}$ command. 
   With this fully controllable tested, we could easily and accurately evaluate Manifoldchain's performance under various network scenarios. 
\end{itemize}

% {\old
% The experimental results demonstrate that Manifoldchain attains a maximum throughput of 2.84 Kb/s in the realistic testbed. Furthermore, a comparative experiment indicates that when Manifoldchain employs the BCSF mechanism, it demonstrates a substantial performance improvement of nearly 30.97\%, compared to the employment of USF mechanism. What's more, we conduct evaluation on Manifoldchain's vertical scalability and horizontal scalability, and results exhibits a significant enhancement in throughput with a reduction in network delay and an increase in the number of miners.
% }

We perform experiments on these testbeds to answer the following questions.
\begin{itemize}
    \item What is the throughput and the latency of Manifoldchain in a realistic deployment?
    \item Is Manifoldchain able to achieve better throughput compared with the SOTA full sharding protocol?
    \item How does the system scale to more miners and higher bandwidths?
\end{itemize}

{\new To our best knowledge, Monoxide is the only full sharding protocol offering a comparable security level to ours, i.e., 1/2 fault tolerance. Other permissionless sharding protocols, like Elastico, Omniledger, and RapidChain, have fault tolerance less than 1/3 or even 1/4, rendering direct experimental comparisons unfair. We compare with Monoxide to highlight our primary motivation: the throughput can be boost by clustering miners with similar bandwidths.} As Monoxide is not open source, we implement it in Rust, and choose the UTXO model instead of its original account/balance model as the transaction model. In this case, Manifoldchain and Monoxide utilizes the same verification scheme to verify both the domestic-txs and cross-txs. Moreover, Manifoldchain and Monoxide have different shard formation mechanisms. Manifoldchain employs the innovative BCSF mechanism proposed in this paper, while Monoxide employs the standard USF mechanism. We evaluate performance of Monoxide on the same testbed as Manifoldchain. In all experiments, all miners share a common block size of approximately 547.14KB (with each block containing 2048 transactions) and maintain a confirmation depth of $\kappa=6$. The mining rates across shards are configured by setting respective mining difficulties.

\subsection{Results on Amazon EC2}\label{tps_real_testbed}

{\bf Throughput.} Our first experiment measures the throughput of Manifoldchain on the realistic testbed. We run Manifoldchain clients on 50 EC2 instances, each hosting a single miner. These instances have been configured with various bandwidth capacities. The available bandwidths are ${5, 10, 20, 40, 60}$ Mbps, with each bandwidth being adopted by 10 miners. We distribute these 50 miners into 5 shards, with 10 miners in each shard. As Manifoldchain and Monoxide adopts different shard formation mechanisms, they exhibit different bandwidth distributions across shards. Specifically, 

\begin{itemize}
    \item Monoxide adopts USF mechanism:
    \begin{itemize}
        \item Shard 0: $\{5, 5, 10, 10, 20, 20, 40, 40, 60, 60\}$.
        \item Shard 1: $\{5, 5, 10, 10, 20, 20, 40, 40, 60, 60\}$.
        \item Shard 2: $\{5, 5, 10, 10, 20, 20, 40, 40, 60, 60\}$.
        \item Shard 3: $\{5, 5, 10, 10, 20, 20, 40, 40, 60, 60\}$.
        \item Shard 4: $\{5, 5, 10, 10, 20, 20, 40, 40, 60, 60\}$.
    \end{itemize}
    \item Manifoldchain adopts BCSF mechanism:
    \begin{itemize}
        \item Shard 0: $\{5, 5, 5, 5, 5, 5, 5, 5, 5, 5\}$.
        \item Shard 1: $\{10, 10, 10, 10, 10, 10, 10, 10, 10, 10\}$.
        \item Shard 2: $\{20, 20, 20, 20, 20, 20, 20, 20, 20, 20\}$.
        \item Shard 3: $\{40, 40, 40, 40, 40, 40, 40, 40, 40, 40\}$.
        \item Shard 4: $\{60, 60, 60, 60, 60, 60, 60, 60, 60, 60\}$.
    \end{itemize}
\end{itemize}

\begin{figure}
  \centering
  \includegraphics[width=8cm]{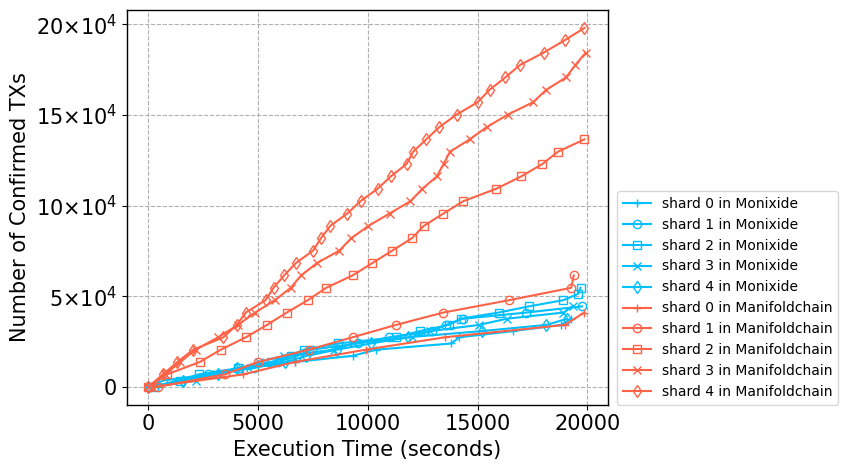}
  \vspace{-4mm}
  \caption{Throughputs of Manifoldchain and Monoxide in distinct shards. Comprising stragglers, all shards in Monoxide and the slowest shard in Manifoldchain achieves nearly the same throughput. Other shards in Manifoldchain, free from stragglers, achieve faster throughput.}
  \label{Fig:exper1}
  \vspace{-3mm}
\end{figure}

We set the mining rates $\lambda_s$ and $\lambda_i$ within each shard by solving Optimization~\ref{optimization}, with the same target error probability for Manifoldchain and Monoxide. The comparable security levels of these two sharding protocols are evidenced by their similar forking rates, both measured at less than 3\% in the experiment. As shown in Figure~\ref{Fig:exper1}, across all shards, Manifoldchain achieves an average throughput of 30.41 tx/sec, while Monoxide only achieves an average throughput of 10.62 tx/sec. The gain in throughput of Manifoldchain is attributed to improved performance in the fast shards where miners are not compromised by stragglers: the throughput in the fastest shard is 5.11 times that of the slowest shard.

%Fig. \ref{Fig:exper1} illustrate the performance of Manifoldchain and Monoxide in the real-world testbed. 
%We compare their performances in the following two aspects:
% \begin{itemize}
%     \item \textbf{Security}. Manifoldchain and Monoxide achieve the same error probability within each shard. Reflecting on the trunk rate, Manifoldchain achieves an average trunk rate of 97.6\%, while Monoxide achieves an average trunk rate of 99.3\%. The gap is negligible due to the complexity of the network in the realistic scenario.
%     \item \textbf{Throughput}. Across all shards, Manifoldchain achieves an average TPS of 30.41 tx/sec, while Monoxide only achieves an average TPS of 10.62 tx/sec. The gain in throughput of Manifoldchain is attributed to improved performance in the fast shards where miners are not compromised by stragglers: the throughput in the fastest shard is 5.11 times that of the slowest shard.
% \end{itemize}

% \subsection{Latency on the Realistic Testbed}

\begin{figure}
  \centering
  \includegraphics[width=6cm]{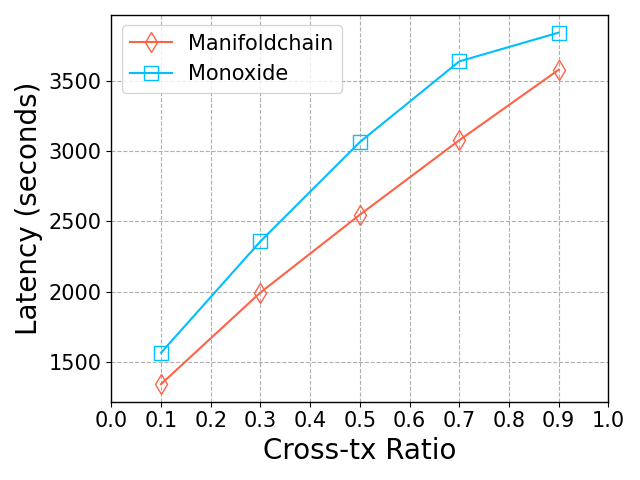}
  \vspace{-5mm}
  \caption{Average transaction latency for different cross-tx ratios. Manifoldchain consistently exhibits lower latency than Monoxide. As the cross-tx ratio increases, latency rises in both Manifoldchain and Monoxide.}
  \label{Fig:exper2}
  \vspace{-5mm}
\end{figure}

\noindent {\bf Latency.} Under the same parameter setting, we evaluate the confirmation latency of Manifoldchain and compare it with Monoxide. We manually generate domestic-tx as well as cross-txs and broadcast them within the corresponding shards. A domestic-tx is randomly assigned to one shard, while a cross-tx is randomly assigned to four shards. Setting $\kappa = 6$ across all shards, we evaluate the average latency of Manifoldchain and Monoxide over all shards under different cross-tx ratios ranging from $0.1$ to $0.9$. Figure \ref{Fig:exper2} shows that Manifoldchain achieves an overall lower latency than Monoxide. It is because that Manifoldchain generally has faster mining rates, and reaches the target confirmation depth with less time. As cross-txs rely on confirmations of all input-txs and output-txs which may locate in slow shards, their confirmation time is longer than that of domestic-txs. Consequently, as the cross-tx ratio increases, the average latency also rises.

\subsection{Results on Simulated Testbed}
The following experiments are conducted on our simulated testbed, to demonstrate the horizontal and vertical scalability of Manifoldchain, and compare them with Monoxide. 

\noindent {\bf Horizontal scalability.} %We conduct our third experiment on the simulated testbed to demonstrate the horizontal scalability of Manifoldchain and compare it to Monoxide. 
We set the size of each shard to 5 and progressively increase the number of shards $m$. To simulate a scenario with varying bandwidths, we have 5 miners configured with a bandwidth configuration of $\{5, 10, 20, 40, 60\}$ Mbps. For each increment in $m$, we add 5 miners with this bandwidth configuration. Subsequently, we measure the maximum throughput achieved by Manifoldchain and Monoxide under varying values of $m$. %What is consistently maintained is that the error probabilities across shards remain constant under different values of $m$.

\begin{figure}
  \centering
  \includegraphics[width=6cm]{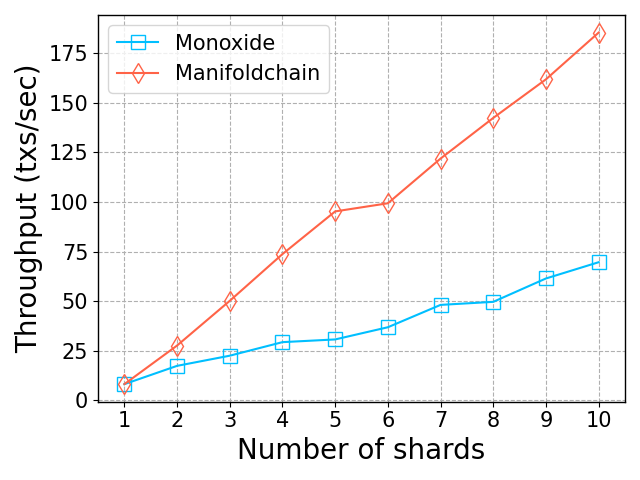}
  \vspace{-4mm}
  \caption{Throughput scaling with the growing number of shards. Manifoldchain demonstrates a more rapid increase in throughput compared with Monoxide.}
  \label{Fig:exper3}
  \vspace{-3mm}
\end{figure}

Fig. \ref{Fig:exper3} demonstrates the horizontal scalability of Manifoldchain and Monoxide. Manifoldchain achieves an average increase in throughput of 19.68 tx/sec for each increment in $m$, while Monoxide can only achieve an average increase in throughput of 6.83 tx/sec under the same increment. The reason why Manifoldchain achieves better horizontal scalability is that it fully utilizes the fast miners' bandwidths through BCSF mechanism for each shard increment. In Monoxide, the newly joined miners contain both stragglers and fast miners. The stragglers are uniformly distributed across shards. Therefore, regardless of the shard to which the fast miners are allocated, they are compromised by the presence of stragglers. %Subsequently, the throughput increase caused by each shard increment is limited by stragglers. 
In sharp contrast, in Manifoldchain, the stragglers among the new miners are assigned to the slow shards, while the fast miners are allocated to the fast shards, admitting higher mining rates. Therefore, the throughput increase results from each shard increment depends on the overall bandwidths configured by all the newly joined miners, which is larger than that in Monoxide.

% \subsection{Vertical Scalability}
\noindent {\bf Vertical scalability.} %Our fourth experiment is conducted on the simulated testbed to showcase the vertical scalability of Manifoldchain and compare it to Monoxide. 
The vertical scalability reflects in how the throughput scales with the increase in the bandwidths of normal miners while considering the presence of stragglers. We simulate this scenario by sampling bandwidths of stragglers and normal miners from two distinct normal distributions. Specifically, there are a total of 25 miners, with 20 being normal miners and the remaining 5 miners being stragglers. The stragglers follow a normal distribution $\mathcal{N}_1(10, 2)$, and the normal miners follow a normal distribution $\mathcal{N}_2(\mu, 7)$, where $\mu$ varies from $35$ to $54$. These 25 miners are distributed across 5 shards, with 5 miners in each shard. In this scenario, we evaluate the maximum throughput achievable by Manifoldchain and Monoxide under varying values of $\mu$. %Similarly, we consistently maintain that the error probabilities across shards remain constant under different values of $\mu$.

\begin{figure}
  \centering
  \includegraphics[width=6cm]{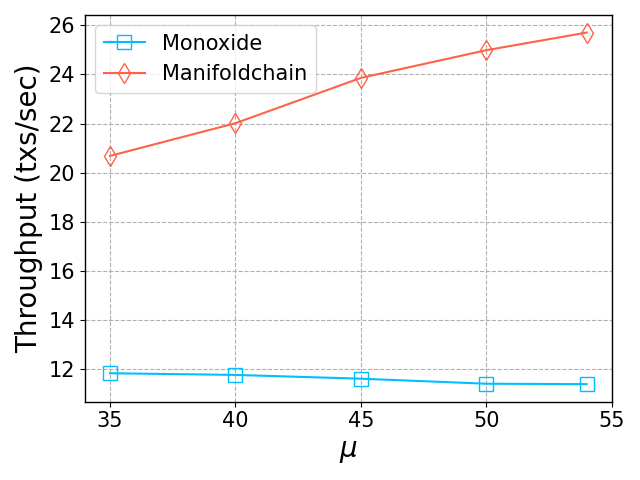}
  \vspace{-5mm}
  \caption{Throughput scaling with the average bandwidth of normal miners ($\mu$). Manifoldchain achieves higher throughput as $\mu$ increases, whereas Monoxide maintains nearly constant throughput across varying $\mu$.}
  \label{Fig:exper4}
  \vspace{-5mm}
\end{figure}

As shown in Figure~\ref{Fig:exper4}, %illustrates the horizontal scalability of Manifoldchain and Monoxide. 
Manifoldchain achieves an average increase in throughput of 1.25 tx/sec for each increment in $\mu$, while Monoxide merely achieves a constant throughput under different $\mu$. The absence of vertical scalability in Monoxide can be attributed to the presence of stragglers within each shard. Since the stragglers are uniformly distributed across shards, they compromise nearly all the normal miners. Despite the increase in the bandwidths of normal miners, the mining rates cannot be correspondingly increased. 
%resulting in a constant throughput. 
In contrast, normal miners in Manifoldchain are not compromised by the presence of stragglers. If the bandwidths of normal miners increase, they can be configured with a higher mining rate to achieve higher throughput while keeping the error probability unchanged. 

\section{Conclusion and Discussions} \label{discussion}
% {\old \noindent {\bf Relaxing assumption on arrival of honest miners}. By slightly altering the shard formation mechanism, we can relax Assumption~\ref{new_miners_assumption}, and improve the protocol's resilience against corruption attacks. Specifically, the bandwidth space is initially divided into $S_X \times (S_Y+V)$ potential shards. Miners within the top $V$ Y-regions, signifying those with the highest bandwidths, are then evenly redistributed across the lower Y-regions.
% Under this mechanism, for an adversary to completely control any single Y-region, it must corrupt not only the miners within that specific region but also those within the top $V$ Y-regions, essentially requiring control over $1+V$ Y-regions in total. This redistribution, while bolstering security, comes with a minor compromise in Manifoldchain's overall throughput, attributed to the dilution of faster miners among slower ones.}
{\new We propose Manifoldchain, a permissionless full sharding protocol which boost vertical scalability within each shard via clustering miners with similar bandwidths. This uneven shard partition may introduce a new vulnerability: the adversary could concentrate hashing power in one shard to establish an adversarial majority. To counter this attack, we offer sharing mining to diffuse the honest hashing power to the entire network, thereby ensuring each shard attains the same level of security as an unsharded blockchain. This design positions Manifoldchain as the most robust sharding protocol, achieving optimal fault tolerance within each shard: security is upheld as long as there is at least one miner in each shard. 
Besides, we provide a tight security analysis that guides us to set the optimal mining rates across shards to maximize throughput. Furthermore, we implement Manifoldchain, 
%establish a solid implementation of Manifoldchain 
and evaluate its performance on both realistic and simulated testbeds. The experimental results not only validate our theoretical assertions but also illustrate that Manifoldchain delivers a notable improvement in throughput over SOTA full sharding protocols.}

\noindent {\bf Adapting to non-PoW consensus}. Manifoldchain's core design principle is versatile, enabling straightforward extension to non-PoW sharding protocols, notably Proof-of-Stake (PoS) and permissioned settings. In PoS systems, uneven stake distribution poses a risk of adversarial dominance within shards. Sharing mining can mitigate this by allowing miners to use a single coin to mine across shards, thus dispersing their power to secure potentially vulnerable shards. On the other hand, incorporating sharing mining into permissioned sharding protocols resemble the design of \textit{shared security} in the Cosmos ecosystem~\cite{InterchainSecurity}, where different Cosmos chains share a common set of validators. We defer a detailed comparison and analysis to future work.

\noindent {\bf Achieving optimal vertical scalability}. Manifoldchain improves horizontal and vertical scalability by clustering nodes with similar bandwidth within the same shard. This idea could be synergistically integrated with protocols optimizing vertical scalability, such as Prism~\cite{DBLP:conf/ccs/BagariaKTFV19} and recent DAG-based BFT protocols~\cite{danezis2022narwhal,keidar2021all,spiegelman2022bullshark,spiegelman2023shoal}. These protocols achieve high throughput, approaching the network capacity, by decoupling block proposal and transaction validation. Incorporating these protocols as the intra-shard consensus, Manifoldchain stands to optimize both vertical and horizontal scalability. Specifically, integrating Prism into Manifoldchain is straightforward: the proposer blocks in Prism can be substituted with exclusive/inclusive blocks, following the same mining and verification processes as Manifoldchain. A detailed analysis of this integration is earmarked for future exploration.

\section{Acknowledgement}

This work is supported in part by the National Nature Science Foundation of China (NSFC) Grant 62106057, the Fundamental Research Funds for the Central Universities (Grant No. 2242024k30059), a gift from Stellar Development Foundation, and the Guangzhou-HKUST(GZ) Joint Funding Program (No. 2024A03J0630). We are also grateful to Lei Yang (MIT) for his valuable suggestions on our evaluation. %We also thank our shepherd and all anonymous reviewers for their helpful feedback.

% trigger a \newpage just before the given reference
% number - used to balance the columns on the last page
% adjust value as needed - may need to be readjusted if
% the document is modified later
%\IEEEtriggeratref{8}
% The "triggered" command can be changed if desired:
%\IEEEtriggercmd{\enlargethispage{-5in}}

% references section

% can use a bibliography generated by BibTeX as a .bbl file
% BibTeX documentation can be easily obtained at:
% http://www.ctan.org/tex-archive/biblio/bibtex/contrib/doc/
% The IEEEtran BibTeX style support page is at:
% http://www.michaelshell.org/tex/ieeetran/bibtex/
%\bibliographystyle{IEEEtranS}
% argument is your BibTeX string definitions and bibliography database(s)
%\bibliography{IEEEabrv,../bib/paper}
%
% <OR> manually copy in the resultant .bbl file
% set second argument of \begin to the number of references
% (used to reserve space for the reference number labels box)
% \begin{thebibliography}{1}

% \bibitem{IEEEhowto:kopka}
% H.~Kopka and P.~W. Daly, \emph{A Guide to \LaTeX}, 3rd~ed.\hskip 1em plus
%   0.5em minus 0.4em\relax Harlow, England: Addison-Wesley, 1999.

% \end{thebibliography}

\bibliographystyle{unsrt}
\bibliography{ref}

\begin{thebibliography}{10}

\bibitem{nakamoto2008bitcoin}
Satoshi Nakamoto.
\newblock Bitcoin: A peer-to-peer electronic cash system.
\newblock {\em Decentralized business review}, 2008.

\bibitem{garay2017bitcoin}
Juan Garay, Aggelos Kiayias, and Nikos Leonardos.
\newblock The bitcoin backbone protocol with chains of variable difficulty.
\newblock In {\em Annual International Cryptology Conference}, pages 291--323. Springer, 2017.

\bibitem{garay2020does}
Juan~A Garay, Aggelos Kiayias, and Nikos Leonardos.
\newblock How does nakamoto set his clock? full analysis of nakamoto consensus in bounded-delay networks.
\newblock {\em Cryptology ePrint Archive}, 2020.

\bibitem{charts}
Charts.
\newblock Bitcoin network hash rate.
\newblock https://ycharts.com/indicators/bitcoin-network-hash-rate.

\bibitem{luu2016secure}
Loi Luu, Viswesh Narayanan, Chaodong Zheng, Kunal Baweja, Seth Gilbert, and Prateek Saxena.
\newblock A secure sharding protocol for open blockchains.
\newblock In {\em Proceedings of the 2016 ACM SIGSAC conference on computer and communications security}, pages 17--30, 2016.

\bibitem{kokoris2018omniledger}
Eleftherios Kokoris-Kogias, Philipp Jovanovic, Linus Gasser, Nicolas Gailly, Ewa Syta, and Bryan Ford.
\newblock Omniledger: A secure, scale-out, decentralized ledger via sharding.
\newblock In {\em 2018 IEEE Symposium on Security and Privacy (SP)}, pages 583--598. IEEE, 2018.

\bibitem{wang2019monoxide}
Jiaping Wang and Hao Wang.
\newblock Monoxide: Scale out blockchains with asynchronous consensus zones.
\newblock In {\em 16th USENIX symposium on networked systems design and implementation (NSDI 19)}, pages 95--112, 2019.

\bibitem{zamani2018rapidchain}
Mahdi Zamani, Mahnush Movahedi, and Mariana Raykova.
\newblock Rapidchain: Scaling blockchain via full sharding.
\newblock In {\em Proceedings of the 2018 ACM SIGSAC conference on computer and communications security}, pages 931--948, 2018.

\bibitem{rana2022free2shard}
Ranvir Rana, Sreeram Kannan, David Tse, and Pramod Viswanath.
\newblock Free2shard: Adversary-resistant distributed resource allocation for blockchains.
\newblock {\em Proceedings of the ACM on Measurement and Analysis of Computing Systems}, 6(1):1--38, 2022.

\bibitem{gencer2016service}
Adem~Efe Gencer, Robbert van Renesse, and Emin~G{\"u}n Sirer.
\newblock Service-oriented sharding with aspen.
\newblock {\em arXiv preprint arXiv:1611.06816}, 2016.

\bibitem{tor_wiki}
Roger Dingledine, Nick Mathewson, Paul~F Syverson, et~al.
\newblock Tor: The second-generation onion router.
\newblock In {\em USENIX security symposium}, volume~4, pages 303--320, 2004.

\bibitem{loesing2009measuring}
Karsten Loesing.
\newblock Measuring the tor network from public directory information.
\newblock {\em Proc. HotPETS, Seattle, WA}, 2009.

\bibitem{10.1145/2946802}
Mashael Alsabah and Ian Goldberg.
\newblock Performance and security improvements for tor: A survey.
\newblock {\em ACM Comput. Surv.}, 49(2), sep 2016.

\bibitem{6407715}
Andriy Panchenko, Fabian Lanze, and Thomas Engel.
\newblock Improving performance and anonymity in the tor network.
\newblock In {\em 2012 IEEE 31st International Performance Computing and Communications Conference (IPCCC)}, pages 1--10, 2012.

\bibitem{gearbox}
Bernardo David, Bernardo Magri, Christian Matt, Jesper~Buus Nielsen, and Daniel Tschudi.
\newblock Gearbox: Optimal-size shard committees by leveraging the safety-liveness dichotomy.
\newblock In {\em Proceedings of the 2022 ACM SIGSAC Conference on Computer and Communications Security}, CCS '22, page 683–696, New York, NY, USA, 2022. Association for Computing Machinery.

\bibitem{reticulum}
Yibin Xu, Jingyi Zheng, Boris D{\"{u}}dder, Tijs Slaats, and Yongluan Zhou.
\newblock A two-layer blockchain sharding protocol leveraging safety and liveness for enhanced performance.
\newblock {\em {IACR} Cryptol. ePrint Arch.}, page 304, 2024.

\bibitem{zhang2023frontrunning}
Jianting Zhang, Wuhui Chen, Sifu Luo, Tiantian Gong, Zicong Hong, and Aniket Kate.
\newblock Front-running attack in sharded blockchains and fair cross-shard consensus, 2023.

\bibitem{DBLP:conf/fc/YuSLAKV20}
Mingchao Yu, Saeid Sahraei, Songze Li, Salman Avestimehr, Sreeram Kannan, and Pramod Viswanath.
\newblock Coded merkle tree: Solving data availability attacks in blockchains.
\newblock In Joseph Bonneau and Nadia Heninger, editors, {\em Financial Cryptography and Data Security - 24th International Conference, {FC} 2020, Kota Kinabalu, Malaysia, February 10-14, 2020 Revised Selected Papers}, volume 12059 of {\em Lecture Notes in Computer Science}, pages 114--134. Springer, 2020.

\bibitem{al2018fraud}
Mustafa Al-Bassam, Alberto Sonnino, and Vitalik Buterin.
\newblock Fraud and data availability proofs: Maximising light client security and scaling blockchains with dishonest majorities.
\newblock {\em arXiv preprint arXiv:1809.09044}, 2018.

\bibitem{gray2005notes}
James~N Gray.
\newblock Notes on data base operating systems.
\newblock {\em Operating systems: An advanced course}, pages 393--481, 2005.

\bibitem{garay2015bitcoin}
Juan Garay, Aggelos Kiayias, and Nikos Leonardos.
\newblock The bitcoin backbone protocol: Analysis and applications.
\newblock In {\em Annual international conference on the theory and applications of cryptographic techniques}, pages 281--310. Springer, 2015.

\bibitem{dembo2020everything}
Amir Dembo, Sreeram Kannan, Ertem~Nusret Tas, David Tse, Pramod Viswanath, Xuechao Wang, and Ofer Zeitouni.
\newblock Everything is a race and nakamoto always wins.
\newblock In {\em Proceedings of the 2020 ACM SIGSAC Conference on Computer and Communications Security}, pages 859--878, 2020.

\bibitem{codeAnon}
Manifoldchain.
\newblock Rust implementation of the manifoldchain sharding protocol.
\newblock https://github.com/Hide-on-bush2/Manifoldchain.

\bibitem{DBLP:conf/nsdi/EyalGSR16}
Ittay Eyal, Adem~Efe Gencer, Emin~G{\"{u}}n Sirer, and Robbert van Renesse.
\newblock Bitcoin-ng: {A} scalable blockchain protocol.
\newblock In Katerina~J. Argyraki and Rebecca Isaacs, editors, {\em 13th {USENIX} Symposium on Networked Systems Design and Implementation, {NSDI} 2016, Santa Clara, CA, USA, March 16-18, 2016}, pages 45--59. {USENIX} Association, 2016.

\bibitem{DBLP:conf/ccs/BagariaKTFV19}
Vivek~Kumar Bagaria, Sreeram Kannan, David Tse, Giulia Fanti, and Pramod Viswanath.
\newblock Prism: Deconstructing the blockchain to approach physical limits.
\newblock In Lorenzo Cavallaro, Johannes Kinder, XiaoFeng Wang, and Jonathan Katz, editors, {\em Proceedings of the 2019 {ACM} {SIGSAC} Conference on Computer and Communications Security, {CCS} 2019, London, UK, November 11-15, 2019}, pages 585--602. {ACM}, 2019.

\bibitem{castro1999practical}
Miguel Castro, Barbara Liskov, et~al.
\newblock Practical byzantine fault tolerance.
\newblock In {\em OsDI}, volume~99, pages 173--186, 1999.

\bibitem{DBLP:journals/vldb/HellingsS23}
Jelle Hellings and Mohammad Sadoghi.
\newblock Byshard: sharding in a byzantine environment.
\newblock {\em {VLDB} J.}, 32(6):1343--1367, 2023.

\bibitem{DBLP:journals/corr/abs-2203-06357}
Dongning Guo and Ling Ren.
\newblock Bitcoin's latency-security analysis made simple.
\newblock {\em CoRR}, abs/2203.06357, 2022.

\bibitem{DBLP:conf/ccs/Gazi0R22}
Peter Gazi, Ling Ren, and Alexander Russell.
\newblock Practical settlement bounds for proof-of-work blockchains.
\newblock In Heng Yin, Angelos Stavrou, Cas Cremers, and Elaine Shi, editors, {\em Proceedings of the 2022 {ACM} {SIGSAC} Conference on Computer and Communications Security, {CCS} 2022, Los Angeles, CA, USA, November 7-11, 2022}, pages 1217--1230. {ACM}, 2022.

\bibitem{DBLP:conf/aft/LiG021}
Jing Li, Dongning Guo, and Ling Ren.
\newblock Close latency-security trade-off for the nakamoto consensus.
\newblock In Foteini Baldimtsi and Tim Roughgarden, editors, {\em {AFT} '21: 3rd {ACM} Conference on Advances in Financial Technologies, Arlington, Virginia, USA, September 26 - 28, 2021}, pages 100--113. {ACM}, 2021.

\bibitem{DBLP:conf/eurocrypt/PassSS17}
Rafael Pass, Lior Seeman, and Abhi Shelat.
\newblock Analysis of the blockchain protocol in asynchronous networks.
\newblock In Jean{-}S{\'{e}}bastien Coron and Jesper~Buus Nielsen, editors, {\em Advances in Cryptology - {EUROCRYPT} 2017 - 36th Annual International Conference on the Theory and Applications of Cryptographic Techniques, Paris, France, April 30 - May 4, 2017, Proceedings, Part {II}}, volume 10211 of {\em Lecture Notes in Computer Science}, pages 643--673, 2017.

\bibitem{DBLP:conf/nsdi/YangPAKT22}
Lei Yang, Seo~Jin Park, Mohammad Alizadeh, Sreeram Kannan, and David Tse.
\newblock Dispersedledger: High-throughput byzantine consensus on variable bandwidth networks.
\newblock In Amar Phanishayee and Vyas Sekar, editors, {\em 19th {USENIX} Symposium on Networked Systems Design and Implementation, {NSDI} 2022, Renton, WA, USA, April 4-6, 2022}, pages 493--512. {USENIX} Association, 2022.

\bibitem{pass2017analysis}
Rafael Pass, Lior Seeman, and Abhi Shelat.
\newblock Analysis of the blockchain protocol in asynchronous networks.
\newblock In {\em Annual International Conference on the Theory and Applications of Cryptographic Techniques}, pages 643--673. Springer, 2017.

\bibitem{DBLP:conf/fc/SompolinskyZ15}
Yonatan Sompolinsky and Aviv Zohar.
\newblock Secure high-rate transaction processing in bitcoin.
\newblock In Rainer B{\"{o}}hme and Tatsuaki Okamoto, editors, {\em Financial Cryptography and Data Security - 19th International Conference, {FC} 2015, San Juan, Puerto Rico, January 26-30, 2015, Revised Selected Papers}, volume 8975 of {\em Lecture Notes in Computer Science}, pages 507--527. Springer, 2015.

\bibitem{DBLP:conf/p2p/DeckerW13}
Christian Decker and Roger Wattenhofer.
\newblock Information propagation in the bitcoin network.
\newblock In {\em 13th {IEEE} International Conference on Peer-to-Peer Computing, {IEEE} {P2P} 2013, Trento, Italy, September 9-11, 2013, Proceedings}, pages 1--10. {IEEE}, 2013.

\bibitem{9479723}
Yuan Wang, Hideaki Ishii, François Bonnet, and Xavier Défago.
\newblock Resilient real-valued consensus in spite of mobile malicious agents on directed graphs.
\newblock {\em IEEE Transactions on Parallel and Distributed Systems}, 33(3):586--603, 2022.

\bibitem{DBLP:conf/tcc/WanXDS20}
Jun Wan, Hanshen Xiao, Srinivas Devadas, and Elaine Shi.
\newblock Round-efficient byzantine broadcast under strongly adaptive and majority corruptions.
\newblock In Rafael Pass and Krzysztof Pietrzak, editors, {\em Theory of Cryptography - 18th International Conference, {TCC} 2020, Durham, NC, USA, November 16-19, 2020, Proceedings, Part {I}}, volume 12550 of {\em Lecture Notes in Computer Science}, pages 412--456. Springer, 2020.

\bibitem{DBLP:journals/dc/KlonowskiKM19}
Marek Klonowski, Dariusz~R. Kowalski, and Jaroslaw Mirek.
\newblock Ordered and delayed adversaries and how to work against them on a shared channel.
\newblock {\em Distributed Comput.}, 32(5):379--403, 2019.

\bibitem{DBLP:journals/dc/AbrahamCDNPRS23}
Ittai Abraham, T.{-}H.~Hubert Chan, Danny Dolev, Kartik Nayak, Rafael Pass, Ling Ren, and Elaine Shi.
\newblock Communication complexity of byzantine agreement, revisited.
\newblock {\em Distributed Comput.}, 36(1):3--28, 2023.

\bibitem{DBLP:conf/crypto/KiayiasRDO17}
Aggelos Kiayias, Alexander Russell, Bernardo David, and Roman Oliynykov.
\newblock Ouroboros: {A} provably secure proof-of-stake blockchain protocol.
\newblock In Jonathan Katz and Hovav Shacham, editors, {\em Advances in Cryptology - {CRYPTO} 2017 - 37th Annual International Cryptology Conference, Santa Barbara, CA, USA, August 20-24, 2017, Proceedings, Part {I}}, volume 10401 of {\em Lecture Notes in Computer Science}, pages 357--388. Springer, 2017.

\bibitem{10.1145/3012426.3022184}
Neal Cardwell, Yuchung Cheng, C.~Stephen Gunn, Soheil~Hassas Yeganeh, and Van Jacobson.
\newblock Bbr: Congestion-based congestion control: Measuring bottleneck bandwidth and round-trip propagation time.
\newblock {\em Queue}, 14(5):20–53, oct 2016.

\bibitem{10.1145/1064212.1064242}
Manish Jain and Constantinos Dovrolis.
\newblock End-to-end estimation of the available bandwidth variation range.
\newblock In {\em Proceedings of the 2005 ACM SIGMETRICS International Conference on Measurement and Modeling of Computer Systems}, SIGMETRICS '05, page 265–276, New York, NY, USA, 2005. Association for Computing Machinery.

\bibitem{DBLP:journals/corr/abs-2210-11546}
Peiyao Sheng, Nikita Yadav, Vishal Sevani, Arun Babu, S.~V.~R. Anand, Himanshu Tyagi, and Pramod Viswanath.
\newblock Proof of backhaul: Trustfree measurement of broadband bandwidth.
\newblock {\em CoRR}, abs/2210.11546, 2022.

\bibitem{DBLP:conf/pods/KeidarD95}
Idit Keidar and Danny Dolev.
\newblock Increasing the resilience of atomic commit at no additional cost.
\newblock In Mihalis Yannakakis and Serge Abiteboul, editors, {\em Proceedings of the Fourteenth {ACM} {SIGACT-SIGMOD-SIGART} Symposium on Principles of Database Systems, May 22-25, 1995, San Jose, California, {USA}}, pages 245--254. {ACM} Press, 1995.

\bibitem{DBLP:conf/itng/ParkLY13}
Sung{-}Hoon Park, Jea{-}Yep Lee, and Su{-}Chang Yu.
\newblock Non-blocking atomic commitment algorithm in asynchronous distributed systems with unreliable failure detectors.
\newblock In Shahram Latifi, editor, {\em Tenth International Conference on Information Technology: New Generations, {ITNG} 2013, 15-17 April, 2013, Las Vegas, Nevada, {USA}}, pages 33--38. {IEEE} Computer Society, 2013.

\bibitem{DBLP:conf/ndss/DanezisM16}
George Danezis and Sarah Meiklejohn.
\newblock Centrally banked cryptocurrencies.
\newblock In {\em 23rd Annual Network and Distributed System Security Symposium, {NDSS} 2016, San Diego, California, USA, February 21-24, 2016}. The Internet Society, 2016.

\bibitem{sompolinsky2015secure}
Yonatan Sompolinsky and Aviv Zohar.
\newblock Secure high-rate transaction processing in bitcoin.
\newblock In {\em Financial Cryptography and Data Security: 19th International Conference, FC 2015, San Juan, Puerto Rico, January 26-30, 2015, Revised Selected Papers 19}, pages 507--527. Springer, 2015.

\bibitem{ren2019analysis}
Ling Ren.
\newblock Analysis of nakamoto consensus.
\newblock {\em Cryptology ePrint Archive}, 2019.

\bibitem{li2020continuous}
Jing Li and Dongning Guo.
\newblock Continuous-time analysis of the bitcoin and prism backbone protocols.
\newblock {\em arXiv preprint arXiv:2001.05644}, 2020.

\bibitem{InterchainSecurity}
Christina Cosmos.
\newblock Interchain security begins a new era for cosmos.
\newblock https://blog.cosmos.network/interchain-security-begins-a-new-era-for-cosmos-a2dc3c0be63.

\bibitem{danezis2022narwhal}
George Danezis, Lefteris Kokoris-Kogias, Alberto Sonnino, and Alexander Spiegelman.
\newblock Narwhal and tusk: a dag-based mempool and efficient bft consensus.
\newblock In {\em Proceedings of the Seventeenth European Conference on Computer Systems}, pages 34--50, 2022.

\bibitem{keidar2021all}
Idit Keidar, Eleftherios Kokoris-Kogias, Oded Naor, and Alexander Spiegelman.
\newblock All you need is dag.
\newblock In {\em Proceedings of the 2021 ACM Symposium on Principles of Distributed Computing}, pages 165--175, 2021.

\bibitem{spiegelman2022bullshark}
Alexander Spiegelman, Neil Giridharan, Alberto Sonnino, and Lefteris Kokoris-Kogias.
\newblock Bullshark: Dag bft protocols made practical.
\newblock In {\em Proceedings of the 2022 ACM SIGSAC Conference on Computer and Communications Security}, pages 2705--2718, 2022.

\bibitem{spiegelman2023shoal}
Alexander Spiegelman, Balaji Aurn, Rati Gelashvili, and Zekun Li.
\newblock Shoal: Improving dag-bft latency and robustness.
\newblock {\em arXiv preprint arXiv:2306.03058}, 2023.

\bibitem{multinomial_distribution}
Janos Galambos.
\newblock {\em Advanced probability theory}, volume~10.
\newblock CRC Press, 1995.

\end{thebibliography}

\newpage

\appendices

\section{Pseudocode of Proof-of-Work}\label{pseudocode_pow}

\begin{algorithm}
\caption{Proof of Work}
\label{alg:pow}
\begin{algorithmic}[1]
% \Require $\sigma$
\Function{$\textbf{PoW}^{\sigma}$}{$\mathtt{hash}$, $\mathtt{info}$, $\mathtt{nonce}$}
    % \If{$\mathtt{view} = \varnothing$}
    %     \State $\mathtt{hash} = 0$
    % \Else
    %     \State $\mathtt{hash} = F(\mathtt{view})$
    % \EndIf
    \State $\mathcal{B} = \text{none}$
    \State $h = \mathcal{G}(\mathtt{hash}, \mathtt{info})$
    \If{$\mathcal{H}(h, \mathtt{nonce}) < \sigma$}
        \State $\mathcal{B}  = \left< \mathtt{hash}, \mathtt{info}, \mathtt{nonce}\right>$
    \EndIf
    \State \textbf{Return} $\mathcal{B}$
\EndFunction
\\
\Function{$\mathbf{PoW}^{\sigma}$.verify}{$\mathcal{B}_{old}, \mathcal{B}_{new}$}
    \State $\left<\mathtt{hash}, \mathtt{info}, \mathtt{nonce}\right> = \mathcal{B}_{new}$
    \State $\left<\mathtt{hash}', \mathtt{info}', \mathtt{nonce}'\right> = \mathcal{B}_{old}$
    \State $\mathtt{hash}'' = \mathcal{H}(\mathcal{G}(\mathtt{hash}', \mathtt{info}'), \mathtt{nonce}')$
    \If{$\mathtt{hash}'' \neq \mathtt{hash}$ or $\mathtt{hash}'' > \sigma$}
        \State \textbf{Return} $\mathtt{false}$
    \EndIf
    \State \textbf{Return} $\mathtt{true}$
\EndFunction
\end{algorithmic}
\end{algorithm}

{\revise
\section{Adaptive Adversary}\label{adaptive_adversary_appendix}

Adaptive adversaries are further categorized into three types: 
\begin{enumerate}
    \item \textit{Mildly} adaptive~\cite{kokoris2018omniledger}: When the adversary requests to corrupt a miner, the corruption occurs after a certain delay.
    \item \textit{Weakly} adaptive~\cite{DBLP:conf/tcc/WanXDS20, DBLP:journals/dc/KlonowskiKM19}: When the adversary requests to corrupt a miner, the corruption occurs after the miner sends all outgoing messages. Additionally, any messages already sent by the miner cannot be erased by the adversary.
    \item \textit{Strongly} adaptive~\cite{DBLP:conf/tcc/WanXDS20, DBLP:journals/dc/AbrahamCDNPRS23}: When the adversary requests to corrupt a miner, the corruption occurs immediately. Furthermore, any messages sent by the miner before the corruption that have not yet been arrived can be erased by the adversary.
\end{enumerate}
}

\section{Hashing Power Splitting Attack}\label{HPSA}

We define the hashing power splitting attack (HPSA) as follows:

\begin{definition}
(Hashing Power Splitting Attack.) Adversary mines an exclusive block with an invalid transaction block and send the exclusive block to other shards, splitting the hashing power of honest miners. 
% In the context of initial design of sharing mining which lacks the verification of the transaction blocks, this block is considered invalid in the source shard, but is considered valid in other shards. The hashing power of honest miners is split since in-shard honest miners and out-shard honest miners are working on different chains.
\end{definition}

As depicted in Fig \ref{fig:split_attack}, a corrupted miner within shard $1$ mines an exclusive block $3$, which contains an invalid transaction block. In the context of initial Sharing mining, only the block headers are broadcast among all shards. Consequently, a miner can only verify the validity and availability of transactions that belong to their respective shard. Within shard $1$, this exclusive block is considered invalid. Nevertheless, the adversary transmits this exclusive block to shard $2$, where it is considered valid, as miners in shard $2$ do not verify the validity of the corresponding transaction block. When a inclusive block $5$ is subsequently mined atop block $3$, miners in shard $1$ refuse to accept this chain, opting instead to continue mining an alternative chain consisting of blocks $2$ and $4$. Under this attack scenario, the hashing power of honest miners is split since honest miners within shard $1$ and honest miners outside of shard $1$ are working on different chains.

\begin{figure}
  \centering
  \includegraphics[width=8cm]{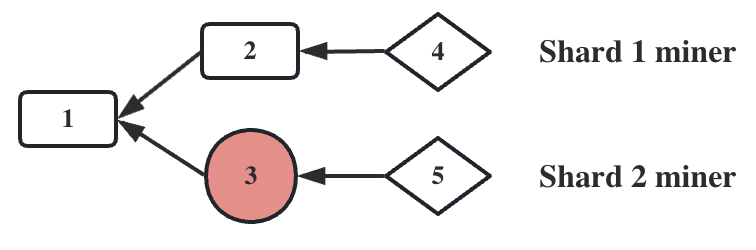}
  \caption{Hashing power splitting attack. The adversary mines an exclusive block 3 with an associated invalid transaction block in shard 1. The exclusive block is deemed invalid in shard 1 but regarded as valid in other shards.}
  \label{fig:split_attack}
\end{figure}

Under the worst-case scenario, HPSA could lead to a consensus failure in a shard with an adversarial majority. This attack entails an adversary with overwhelming hashing power mining a longer chain that consists of all invalid transaction blocks and making it public to the global view. In this scenario, honest miners would consider the adversarial chain as valid because they are unable to verify the availability and validity of transactions since only the block headers are accessible. Therefore, other miners across shards cannot access and collaborate on the susceptible honest chain in the source shard, resulting in their inability to contribute their hashing power to the honest miners in that shard. Consequently, without honest majority and validity/availability verification, the adversary could arbitrarily cause double-spending failures.

% Ensuring an honest majority in our BCSF mechanism poses a significant challenge. The adversary can corrupt miners with similar bandwidths, causing them to be assigned to a shard with an adversarial majority. Additionally, the verification mechanism, with a communication complexity of $\mathcal{O}(B)$, is not viable as a full miner cannot download the entire block belonging to other shards in a sharding protocol.

\section{Transaction  Validity \& Block Availability Verification}\label{validity_availability_verification}

\textbf{Transaction validity verification}. In the context of a target shard where adversary wants to split the hashing power, the miners within the target shard are referred as ``in-shard miner'', while miners outside the target shard are referred as ``out-shard miners''. Out-shard miners rely on at least one honest in-shard miner to verify the data validity of a transaction block. In particular, an honest in-shard miner sends a fraud proof to other out-shard miners, facilitating their verification of the transaction block for formal correctness and absence of double spending issues. The fraud proof, denoted as $\mathtt{proof\_of\_invalidity}$, is formatted as follows:

\begin{equation}
\begin{split}\label{pow_eq}
\mathtt{proof\_of\_invalidity} = &\\
\{(\mathtt{tx}_0, \mathtt{block\_hash}_0,& \mathtt{tx\_merkle\_proof}_0), \\
(\mathtt{tx}_1, \mathtt{block\_hash}_1,& \mathtt{tx\_merkle\_proof}_1), \\
\mathtt{fault\_type}&\}
\end{split}
\end{equation}

With respect to double spending fault, the fraud proof reveals the invalidity of a transaction block by proving the existence of two conflicting transactions that use the same UTXOs as inputs. $\mathtt{tx}_0$ represents a transaction within the invalid transaction block, with $\mathtt{block\_hash}_0$ denoting the hash value of the corresponding consensus block. The inclusion of $\mathtt{tx}_0$ in the transaction block can be verified through $\mathtt{tx\_merkle\_proof}_0$. Moreover, $\mathtt{tx}_0$ is characterized as a transaction conflicting with $\mathtt{tx}_1$, and correspondingly, $\mathtt{block\_hash}_0$ stands as the hash value associated with its corresponding consensus block. The presence of $\mathtt{tx}_1$ within the transaction block can be substantiated by employing $\mathtt{tx\_merkle\_proof}_1$. In the context of incorrect formatting fault, i.e., certain signatures are improperly generated, the second component of the fraud proof remains void, and only the proof of inclusion of $\mathtt{tx}_0$ inclusion is only requisite. Upon receiving a fraud proof, out-shard miners assess its legitimacy, and mark the designated block as verified if the fraud proof is correct, otherwise mark the block as invalid.

\textbf{Data availability verification}. We utilize Coded Merkle Tree (CMT) proposed by Yu et al. \cite{DBLP:conf/fc/YuSLAKV20} to facilitate the data availability verification across shards. This data structure enables any miner to verify the full availability of a transaction block through the mere download of a block hash commitment with a size of $\mathcal{O}(1)$ byte, combined with a random sampling of $\mathcal{O}(\log(B))$ bytes. Therefore, out-shard miners can validate the data availability of an in-shard transaction block by requesting just $\mathcal{O}(\log(B))$ samples. Specifically, different miners' behaviours are summarized as follows:

\begin{enumerate}
\item Block producer (in-shard miner): (a) Generates a consensus block (which may be either an exclusive or inclusive block) containing CMT and broadcasts this block to all miners, while broadcasting the associated transaction block exclusively to in-shard miners. (b) Responds to sample requests received from out-shard miners.
\item Out-shard miner: (a) Upon receiving a new consensus block from other shards, initiates separate, anonymous, and intermittent sampling requests with replacement directed at in-shard miners who claim the block is available. (b) Broadcasts received samples to all connected in-shard miners. (c) Assumes block availability if all requested samples are received. (d) Rejects the consensus block if any requested samples are not received within the bounded network delay $\Delta$. (e) Rejects the block if any error proofs, such as $\mathtt{incorrect\_coding\_proof}$ or $\mathtt{bad\_code\_proof}$, are received.
\item Other in-shard miner: (a) Upon receiving valid samples, attempts to recover the data block through both downloading the associated transaction block from other in-shard miners and collecting samples forwarded by out-shard miners. (b) Rebroadcasts received valid samples. (c) Sends corresponding proof and rejects the block if adversarial behavior, such as incorrect coding or bad code, is detected. (d) Declares the availability of the block to all other miners and responds to sample requests from out-shard miners upon successful receipt or full decoding and verification of a data block.
\end{enumerate}

{\revise
\section{Cross-shard Transaction Atomicity}\label{cross_tx_atomicity}

In this section,we prove the cross-shard transaction atomicity of Manifoldchain. To facilitate the proof, we present the following lemma:

\begin{lemma}\label{eventual_confirmation}
In Manifoldchain, all input-txs, output-txs, refund-txs created by honest users will eventually be confirmed.
\end{lemma}

\begin{proof}
Honest users always faithfully follow the protocol to broadcast input-txs, output-txs, and refund-txs in all involved shards. Due to the liveness of Manifoldchain, these transactions will eventually be confirmed.
\end{proof}

Formally, we have the following theorem:

\begin{theorem}[Cross-shard transaction atomicity]
In Manifoldchain, for a cross-tx created by honest users, (1) if all input-txs are confirmed as accepted, then all output-txs will eventually be confirmed as accepted; (2) if any input-tx is confirmed as rejected, then all output-txs will eventually be confirmed as rejected, and all refund-txs will eventually be confirmed as accepted.
\end{theorem}

\begin{proof}
We will prove statement (1) by contradiction. Suppose that all input-txs are confirmed as accepted and Lemma~\ref{eventual_confirmation} holds. For (1) to be false, there must exist an output-tx that is confirmed as rejected. However, according to the verification mechanism described in Section~\ref{cross-tx method}, a confirmed rejected input-tx must exist in at least one input shard for an output-tx to be confirmed as rejected. This directly contradicts the initial condition that all input-txs are confirmed as accepted, as an input-tx can only be confirmed as either accepted or rejected, not both. 

We will also prove statement (2) by contradiction. Suppose that there exists an input-tx that is confirmed as rejected and Lemma~\ref{eventual_confirmation} holds. For (2) to be false, at least one of the following events must occur: (i) at least one output-tx is confirmed as accepted, or (ii) at least one refund-tx is confirmed as rejected. Event (i) necessitates that all input-txs are confirmed as accepted, which directly contradicts the condition that one of the input-txs is confirmed as rejected, as an input-tx can only be confirmed as either accepted or rejected. Therefore, event (ii) must also be false, as it requires at least one output-tx to be confirmed as accepted, which has already been proven false.
\end{proof}
}

\section{Definition \& Theorem}\label{theoretical_anslysis}

In this section, we present (1) the formal definitions of security properties; (2) the main theorems regarding the security and throughput. Prior to that, we introduce some basic notations. Formally, the initial block of a chain $\mathcal{C}$ is referred to as its \textit{genesis block} and denoted by $gen(\mathcal{C})$. The final block of $\mathcal{C}$ is represented by $head(\mathcal{C})$. A sub-chain of $\mathcal{C}$, in which the last $\kappa$ blocks have been removed, is denoted by $\mathcal{C}^{\lceil k}$. The extension of a chain $\mathcal{C}$ by incorporating a block $\mathcal{B}$ is represented as $\mathcal{C}' = \mathcal{C} || \mathcal{B}$. For simplicity, we express $\mathcal{B} \in \mathcal{C}$ if and only if $\mathcal{C}$ contains the block $\mathcal{B}$. The number of blocks in $\mathcal{C}$ is denoted by $|\mathcal{C}|$. We denote by $\mathcal{C}_i^t$ the longest chain in miner $i$'s view at time $t$. We write $\mathcal{C}_1 \preceq \mathcal{C}_2$ if $\mathcal{C}_1$ is a prefix of $\mathcal{C}_2$. 

\noindent\textbf{Execution Environment.} We denote by $\mathcal{E}nv(\Delta, \rho) \in \mathbb{E}$ the execution environment with bounded network delay $\Delta$ and honest ratio $\rho$, where $\mathbb{E}$ represents the set of execution environments with all possible $\Delta$ and $\rho$. Besides, we define a shard formation mechanism as $f: \mathbb{E}\rightarrow\{\mathbb{E}\}^{m}$ which splits a global execution environment into $m$ sub-environments $\{\mathcal{E}nv(\Delta_i, \rho_i)\}^m$.

\noindent\textbf{Execution Model.} We employ a natural continuous-time model and represent PoW mining as a homogeneous Poisson point process. Let $\lambda$ be the total mining rate of the network, specifically denoting the number of blocks generated within one second. Honest block arrivals and adversarial block arrivals as two independent Poisson processes with rates $\rho \lambda$ and $(1-\rho)\lambda$, respectively. Owing to the Poisson merging and splitting properties, it is equivalent to regard all block arrivals as a single Poisson process with a rate of $\lambda$, where each block is honest with probability $\rho$ and adversarial with probability $1-\rho$.

% \begin{itemize}
%     \item \textit{Common Prefix Property (CP)}. A protocol satisfies the common prefix property with parameter $\kappa \in \mathbb{N}$ if, for any pair of honest miners $M_1$ and $M_2$ adopting the chains $\mathcal{C}_1$ and $\mathcal{C}_2$ at times $t_1 \leq t_2$ respectively, the condition $\mathcal{C}_1^{\lceil \kappa} \preceq \mathcal{C}_2$ holds true.
%     \item \textit{Chain Quality Property (CQ)}. A protocol satisfies the chain quality property with parameters $u \in \mathbb{R}$ and $l \in \mathbb{N}$ if, for any honest miner $M$ adopting a chain $\mathcal{C}$, the condition that the fraction of honest blocks is at least $u$ holds true for any $l$ consecutive blocks within $\mathcal{C}$.
%     \item \textit{Chain Growth Property (CG)}. A protocol satisfies the chain growth property with parameters $\tau \in \mathbb{R}$ and $s \in \mathbb{N}$ if, for any honest miner $M$ adopting a chain $\mathcal{C}$, the condition holds that after any $s$ consecutive times units, it adopts a chain that is at least $\tau \cdot s$ blocks longer than $\mathcal{C}$.
% \end{itemize}

% These three properties are commonly considered in longest-chain style protocols to comprehensively establish their security. The CP property implies consistency, while the combined CQ and CG properties imply liveness. 

\subsection{Security Property}\label{security_property}

% \noindent\textbf{Common Prefix Property (CP).} CP implies the consistency of a blockchain, which we denote as $\kappa$-consistency. This property stipulates that honest miners should agree on the current chain, with the exception of a small number, $\kappa$, of "unconfirmed" blocks at the end of the chain. If a blockchain protocol can show this property holds, honest miners are assured that for an sufficiently large choice of $\kappa$, "confirmed" blocks will never be lost from the longest chain except with an exponentially small probability $\varepsilon(\kappa)$ in $\kappa$. Specifically, it ensures that the following \textit{consistency violation} event $E_v$ will never happen except with probability $\varepsilon(\kappa)$:
In the context of a PoW-based protocol, it implements a robust transaction ledger if it satisfies the following three properties: Common Prefix(CP), Chain Quality(CQ), and Chan Growth(CG). These properties constitute the necessary requirement to achieve security, where the CP property implies consistency while the combined CQ and CG properties imply liveness.

\subsubsection{Common Prefix}
Prior to the definition of CP, we introduce the following bad event $E_v$:

\begin{definition}
    (\textbf{Consistency violation of a target transaction}). The consistency of a target transaction is violated when a block containing the target transaction is confirmed, while simultaneously, another distinct block (possibly containing or not containing the target transaction) is confirmed at the same height within the blockchain.
\end{definition}

We define the CP property $\mathcal{CP}(\kappa, \varepsilon(\kappa))$ held by protocol $\prod$ parameterized with confirmation-depth $\kappa$ and negligible function $\varepsilon(\kappa)$ as follows:

\begin{definition}
A protocol $\prod$ holds common prefix property $\mathcal{CP}(\kappa, \varepsilon(\kappa))$ in environment $\mathcal{E}nv(\Delta, \rho)$ if consistency violation events happen with an negligible probability $\varepsilon(\kappa)$ when employing $\kappa$-confirmation rule in case that any message is delayed for $\Delta$ seconds and the ratio of the adversarial hashing power is $\rho$.
\end{definition}

CP implies the consistency of a blockchain, which we denote as $\kappa$-consistency. This property stipulates that honest miners should agree on the current chain, with the exception of a small number, $\kappa$, of unconfirmed blocks at the end of the chain. Likewise, we introduce the concept of an \textit{anticipated growth} event initially and subsequently leverage it to formulate the definition of CG.

% \noindent\textbf{Chain Growth Property (CG).} CG as a primary element contributing to a blockchain's liveness, refers to the concept that the chain's length expands proportionally with its runtime. This property stipulates that in a honest miner's view, the length of the longest chain should grow at an average rate $g$ over a continuous time span of $\mathcal{T}$ seconds, except with an exponentially small probability $\varepsilon(\mathcal{T})$. Assuming $\mathcal{C}_i^t$ represents the longest chain in miner $i$'s view at time $t$. Similar to the formulation of CP, we define \textit{Anticipated Growth} event $E_g$ as follows:

\subsubsection{Chain Growth}
Prior to the definition of CG, we introduce the following bad event $E_g$:
\begin{definition}
We say the anticipated growth event $E_g(t, \Delta, \mathcal{T})$ occurs, denoted by $E_g(t, \Delta, \mathcal{T}) = 1$, if the following two events hold true:
\begin{itemize}
    \item (\textbf{Consistent Length.}) Given the current time $\mathtt{t}$, For any time $r \leq \mathtt{t} - \Delta$, $r + \Delta \leq r' \leq \mathtt{t} $, for every two miners $i, j$ such that $i$ is hones at $r$ and $j$ is honest at $r'$, $len(\mathcal{C}_j^{r'} \geq \mathcal{C}_i^{r})$ holds.
    \item (\textbf{Chain Growth.}) For any time $r\leq \mathtt{t} - t$, it holds that $\min\limits_{i,j}(len(\mathcal{C}_j^{r + t}) - len(\mathcal{C}_i^r)) \geq \mathcal{T}$. 
\end{itemize}
\end{definition}

We define the CG property held by protocol $\prod$ parameterized with the chain growth rate $\mathtt{g}$ and a negligible function $\varepsilon(\cdot)$ as follows:

\begin{definition}
A blockchain protocol $\prod$ holds chain growth property $\mathcal{CG}(\mathtt{g}, \varepsilon(\cdot))$ in environment $\mathcal{E}nv(\Delta, \rho)$ if anticipated growth event $E_g$ happens except with a negligible probability $\varepsilon(\mathcal{T})$. Formally, there exists a constant $c$ such that for ever $\mathcal{T} \geq c$ and $t \geq \frac{\mathcal{T}}{\mathtt{g}}$, the following holds:
\begin{equation}
\begin{split}
\mathtt{Pr}[E_g(t, \Delta, \mathcal{T}) = 1] \geq 1 - \varepsilon(\mathcal{T}).
\end{split}
\end{equation}
\end{definition}

CG as a primary element contributing to a blockchain's liveness, refers to the concept that the chain's length expands proportionally with its runtime. This property stipulates that in a honest miner's view, the length of the longest chain should grow at an average rate $g$ over a continuous time span of $\mathcal{T}$ seconds, except with an exponentially small probability $\varepsilon(\mathcal{T})$. In a similar vein, we establish the definition of the \textit{Anticipated Quality} event, denoted as $E_q$ herein, and subsequently construct the definition of the CQ based on this event.

% \noindent\textbf{Chain Quality Property (CQ).} The CQ property is another pivotal metric essential for a blockchain's liveness. It quantifies the number of honest blocks mined within a significantly long and continuous segment of the longest chain in an honest miner's view. In essence, the CQ property necessitates that any longest chain, as adopted by honest miners, must contain a certain proportion of honest blocks, otherwise the liveness of a blockchain is violated even in the presence of a positive chain growth rate. we categorize a block as honest if it is mined by an honest miner; in contrast, it is considered adversarial if it is produced otherwise. 

\subsubsection{Chain Quality}
Prior to the definition of CQ, we introduce the following bad event $E_q$:
\begin{definition}
We say the anticipated quality event $E_q(\mathtt{q}, \mathcal{T})$ occurs, denoted by $E_q(\mathtt{q}, \mathcal{T}) = 1$, provided that for each moment $t$ and every honest miner $i$, within any continuous sequence of $\mathcal{T}$ blocks in $\mathcal{C}_i^t$, the proportion of honest blocks is at a minimum of $\mathtt{q}$.
\end{definition}

We define the CQ property held by protocol $\prod$ parameterized with the chain quality rate $\mathtt{q}$ and a negligible function $\varepsilon(\mathcal{T})$ as follows:

\begin{definition}
A blockchain protocol $\prod$ holds chain quality property $\mathcal{CQ}(\mathtt{q}, \varepsilon(\cdot))$ in environment $\mathcal{E}nv(\Delta, \rho)$ if anticipated quality event $E_q$ occurs except with the negligible probability $\varepsilon(\mathcal{T})$. Formally, there exists a constant $c$ such that for every $\mathcal{T} \geq c$, the following holds:
\begin{equation}
\begin{split}
\mathtt{Pr}[E_q(\mathtt{q}, \mathcal{T}) = 1] \geq 1 - \varepsilon(\mathcal{T}).
\end{split}
\end{equation}
\end{definition}

The CQ property is another pivotal metric essential for liveness. It quantifies the number of honest blocks mined within a significantly long and continuous segment of the longest chain in an honest miner's view. In essence, the CQ property necessitates that any longest chain, as adopted by honest miners, must contain a certain proportion of honest blocks, otherwise the liveness of a blockchain is violated even in the presence of a positive chain growth rate.

\subsection{Security Guarantee}\label{security_analysis}

% Building on top of the foundation of the backbone protocol, we leverage previous security analyses \cite{DBLP:journals/corr/abs-2203-06357}\cite{DBLP:conf/eurocrypt/PassSS17}\cite{DBLP:conf/aft/LiG021}\cite{DBLP:conf/ccs/Gazi0R22} and extend them to address scenarios with multiple mining difficulties and various network delays. 

% We exhibit the security of Manifoldchain by showing that it achieves the three basic security properties: Common Prefix (CP), Chain Quality (CQ), and Chain Growth (CG).

We denote by $\prod_{mfd}(f, \lambda_s, \{\lambda_i\}^m)$ the Manifoldchain protocol, $\lambda_s$ and $\lambda_i$ represents the mining rates of inclusive block and exclusive block in shard $i$ respectively. We present the constraint that Manifoldchain must satisfy in order to uphold the CP as follows:

\begin{theorem}\label{cp_property}
The Manifoldchain protocol $\prod_{mfd}(f, \lambda_s, \{\lambda_i\}^m)$ holds the CP property $\mathcal{CP}(\kappa, \varepsilon(\kappa))$ in environment $\mathcal{E}nv(\Delta, \rho)$ as long as 

\begin{equation}
\begin{split}\label{eq_pi}
p_i = \frac{\lambda_i \rho_i + m \lambda_s \rho}{\lambda_i + m\lambda_s} e^{-(\lambda_i + \lambda_s)\Delta_i} > \frac{1}{2},
\end{split}
\end{equation}

where $\{\mathcal{E}nv(\Delta_i, \rho_i)\} = f(\mathcal{E}nv(\Delta, \rho))$ and $\varepsilon(\kappa)$ satisfies

\begin{equation}
\begin{split}
\varepsilon(\kappa) \leq (2 + 2\sqrt{\frac{p_i}{1-p_i}})(4p_i(1-p_i))^{\kappa},
\end{split}
\end{equation}

regardless of the adversary's attack strategy.
\end{theorem}

Notice that $\varepsilon(\kappa)$ represents the security level of CP. %We assume a blockchain protocol with static mining rates, where the mining rates remain unchanged during the execution. 
In a practical scenario with expected system fault probability $\rm \overline{\mathcal{P}}$, we say $\prod_{mfd}$ holding $\mathcal{CP}(\kappa, \varepsilon(\kappa))$ satisfies the security requirement as long as 

\begin{equation}
\begin{split}
\varepsilon(\kappa) \leq \overline{\mathcal{P}}.
\end{split}
\end{equation}

To prove Theorem \ref{cp_property}, we shift from a standard to a defined \textit{severe execution environment}, where any delivered full block experiences the maximal network delay and during the delivering any mining operation is considered invalid. Achieving security in this severe environment is more challenging than in standard environment. We then establish that the private-mining attack is the optimal strategy in this severe environment, simplifying our analysis. Finally, we conclude by calculating the upper bound on mining rates to prevent the private-mining attack, which also acts as a constraint against all possible attacks. The detailed proof is presented in Appendix \ref{CP_proof}.

As Manifoldchain implements independent growth of distinct chains with various mining rates across shards, different shards hold CG property with different chain growth rates. Consequently, we conduct our analysis by considering the individual CG property held by each shard respectively. We demonstrate that Manifoldchain holds CG property and present our main result as follows:
\begin{theorem}\label{cg_property}
In an environment $\mathcal{E}nv(\Delta, \rho)$, the Manifoldchain protocol $\prod_{mfd}(f, \lambda_s, \{\lambda_i\}^m)$ holds the CG property $\mathcal{CG}_i(\mathtt{g_i}, \varepsilon(\cdot))$ in shard $i$ respectively, and for any $\delta_i > 0$, the chain growth rate satisfies
\begin{equation}
\begin{split}
g_i =& (1 - \delta_i)\frac{p_i(\lambda_s + \lambda_i)}{1 + \Delta_i(\lambda_s + \lambda_i)},
\end{split}
\end{equation}
where
\begin{equation}
\begin{split}
\{\mathcal{E}nv(\Delta_i, \rho_i)\} = f(\mathcal{E}nv(\Delta, \rho)) \\
p_i = \frac{\lambda_i \rho_i + m \lambda_s \rho}{\lambda_i + m\lambda_s} e^{-(\lambda_i + \lambda_s)\Delta_i}.
\end{split}
\end{equation}
\end{theorem}

In our proof, we first establish the lower bound for Manifoldchain's chain growth rate in the severe environment, where the rate of chain growth is notably slower in comparison to the standard environment. We consider this lower bound as a baseline for the standard environment, subsequently enabling us to establish the validity of Theorem \ref{cg_property}. The proof of Theorem \ref{cg_property} can be found in Appendix \ref{cg_property_proof}, to which the reader is referred to further details. Next, we show that Manifoldchain holds CQ property in Theorem~\ref{cq_property}.

\begin{theorem}\label{cq_property}
In an environment $\mathcal{E}nv(\Delta, \rho)$, the Manifoldchain protocol $\prod_{mfd}(f, \lambda_s, \{\lambda_i\}^m)$ holds the CQ property $\mathcal{CQ}(\mathtt{q}, \varepsilon(\cdot))$ in shard $i$ respectively, and for any $\delta_i > 0$, the chain quality rate $\mathtt{q}_i$ satisfies
\begin{equation}
\begin{split}
\mathtt{q_i} =& 1 - (1+\delta_i)\frac{1 + \Delta_i p_i\rho_i(\lambda_s + \lambda_i)}{p_i},
\end{split}
\end{equation}
where
\begin{equation}
\begin{split}
\{\mathcal{E}nv(\Delta_i, \rho_i)\} = f(\mathcal{E}nv(\Delta, \rho)) \\
p_i = \frac{\lambda_i \rho_i + m \lambda_s \rho}{\lambda_i + m\lambda_s} e^{-(\lambda_i + \lambda_s)\Delta_i}.
\end{split}
\end{equation}
\end{theorem}

To establish Theorem \ref{cq_property}, we compute, within a fixed time window, the minimum number of blocks generated by honest miners using Theorem \ref{cg_property}, and the maximum number of blocks generated by all miners. This calculation allows us to derive the lower bound for the chain quality rate. The proof of Theorem \ref{cq_property} is provided in Appendix \ref{cq_property_proof}.

\subsection{Throughput Guarantee}\label{scalability_analysis}

It is universally recognized that there is a trade-off between security and throughput in Blockchain. In this section we explicitly characterized the trade-off between security and throughput for Manifoldchain and show how it achieves a better trade-off compared to the typical sharding protocol employing uniform distribution mechanism by well-designed parameters. Due to this optimized trade-off, Manifoldchain enables fast shards to significantly improve throughput and the slow shards to preserve the same throughput as Bitcoin. Formally, we define throughput, denoted by $\mathcal{T}PS$, as the number of transactions confirmed per second. The throughput achieved by Bitcoin serves as our benchmark for the theoretical comparison. 

Notice that any protocol with zero mining rate potentially holds CP property but dose not possess CG property. As discussing such a protocol without liveness is meaningless and unrealistic,  we consider an ``efficient common prefix property'' which requires holder's mining rate to be greater than a threshold:

\begin{definition}
we say that a protocol $\prod(f, \lambda)$ holds efficient common prefix (ECP) property $\mathcal{ECP}(\kappa, \varepsilon(\kappa), \underline{\lambda})$ in $\mathcal{E}nv(\Delta, \rho)$ as long as 
\begin{itemize}
    \item $\prod$ holds $\mathcal{CP}(\kappa, \varepsilon(\kappa))$;
    \item $\lambda \geq \underline{\lambda}$;
\end{itemize}
\end{definition}

We initially present Bitcoin's throughput as follows:

\begin{theorem}\label{btc_tps}
Executing in environment $\mathcal{E}nv(\Delta, \rho)$, holding the ECP property $\mathcal{ECP}(\kappa, \varepsilon(\kappa), \underline{\lambda})$, satisfying the system security requirement $\overline{\mathcal{P}}$, the Bitcoin $\prod_{btc}(\lambda)$ achieves a $\mathcal{T}PS$ bounded by

\begin{equation}
\begin{split}
\mathcal{T}PS(\mathcal{E}nv_{btc}(\Delta, \rho), \mathcal{E}nv, \mathcal{ECP}, \overline{\mathcal{P}}) \leq \frac{\rho}{\Delta}\log \frac{\rho}{\delta(\overline{\mathcal{P}}, \kappa)},
\end{split} 
\end{equation}

where

\begin{equation}
\begin{split}
&\delta(\overline{\mathcal{P}}, \kappa) = \frac{1}{2}(1 + \sqrt{1 - \sqrt[\kappa]{\frac{\overline{\mathcal{P}}}{2 + 2\sqrt{\frac{{\rm exp}(-\underline{\lambda}\Delta)}{1-{\rm exp}(-\underline{\lambda}\Delta)}}}}}) \in (\frac{1}{2}, 1).
\end{split}
\end{equation}
\end{theorem}

Then we present Manifoldchain's throughput as follows:

\begin{theorem}\label{mfd_tps}
Executing in environment $\mathcal{E}nv(\Delta, \rho)$, holding the ECP property $\mathcal{ECP}(\kappa, \varepsilon(\kappa), \underline{\lambda})$, satisfying the system security requirement $\overline{\mathcal{P}}$, the Manifoldchain $\prod_{mfd}(f, \lambda_s, \{\lambda_i\}^m)$ achieves a $\mathcal{T}PS$ bounded by 

\begin{equation}\label{tps_manifold}
\begin{split}
\mathcal{T}PS(\prod_{mfd}, \mathcal{E}nv, \mathcal{ECP}) \leq \frac{1}{\Delta}\log \frac{\rho}{\delta(\overline{\mathcal{P}}, \kappa)}(m\rho + \sum_{i}^m \gamma_i\rho_i),
\end{split}
\end{equation}

as long as 

\begin{equation}\label{constrain}
\begin{split}
&\frac{1}{\gamma_i + 1} \frac{\log \frac{m\rho}{\delta(\overline{\mathcal{P}}, \kappa)(\gamma_i + m)}}{\log \frac{\rho}{\delta(\overline{\mathcal{P}}, \kappa)}} \geq \frac{\Delta_i}{\Delta},\\
&\underline{\lambda} \geq \frac{1}{(\gamma_i + 1)\Delta_i}\log(\frac{\gamma_i}{\gamma_i + m}\cdot\frac{\rho_i}{\rho} + 1),\\
&\gamma_i < \frac{\rho(\frac{\rho}{\delta(\overline{\mathcal{P}}, \kappa)}-1)m}{\rho + 1}.
\end{split}
\end{equation}

where

\begin{equation}
\begin{split}
& \gamma_i = \frac{\lambda_i}{\lambda_s},\\
&\{\mathcal{E}nv(\Delta_i, \rho_i)\}^m = f(\mathcal{E}nv(\Delta, \rho)),\\
&\underline{\Delta} = \min\{\Delta_i\}^m,\\ 
&\delta(\overline{\mathcal{P}}, \kappa) = \frac{1}{2}(1 + \sqrt{1 - \sqrt[\kappa]{\frac{\overline{\mathcal{P}}}{2 + 2\sqrt{\frac{{\rm exp}(-\underline{\lambda}\underline{\Delta})}{1-{\rm exp}(-\underline{\lambda}\underline{\Delta})}}}}}) \in (\frac{1}{2}, 1).
\end{split}    
\end{equation}

Ineq. \ref{tps_manifold} holds an equality when 

\begin{equation}
\begin{split}
\lambda_s = \frac{1}{\overline{\Delta}}\log \frac{\rho}{\delta(\overline{\mathcal{P}}, \kappa)}.
\end{split}
\end{equation}

\end{theorem}

It should be noted that Manifoldchain attains the throughput equivalent to that of Bitcoin, as presented in Theorem \ref{btc_tps}, for each individual shard when $\gamma_i$ is set to zero. The proofs for Theorems \ref{btc_tps} and \ref{mfd_tps} are deferred to Appendix \ref{mfd_tps_proof}.

We denote by $f_{usf}$ the USF mechanism and by $\prod_{mfd}(f_{usf}, \cdot, \cdot)$ the Manifoldchain protocol which deploys USF mechanism. Under this mechanism, each shard serves as a microcosm of the global network, equipped with the global honesty ratio $\rho$ and the maximum network delay $\Delta$, as stragglers are uniformly distributed among the shards. Formally, $\Delta_i = \Delta$ and $\rho_i = \rho$ are held for all $i$. The constraint, represented by Ineq. \ref{constrain}, necessitates that $\gamma_i$ equals zero to uphold the ECP property $\mathcal{ECP}(\kappa, \varepsilon(\kappa), \underline{\lambda})$ and the system security requirement $\overline{\mathcal{P}}$. In this scenario, $\lambda_i = \gamma_i \lambda_s = 0$ implies that only inclusive blocks are generated within any given shard. Furthermore, given that each inclusive block is broadcast across all shards and contributes to consensus, each shard eventually maintains chains with the same common prefix, causing Manifoldchain to regress to a state analogous to Bitcoin. Distinct from the original Bitcoin protocol where miners process all transactions, miners in Manifoldchain only process transactions within their respective shards. In the instance of $\gamma_i = 0$, Manifoldchain achieves a bounded throughput of $\frac{m\rho}{\Delta}\log \frac{\rho}{\delta(\overline{\mathcal{P}},\kappa)}$ according to Ineq.\ref{tps_manifold}, $\frac{\rho}{\Delta}\log \frac{\rho}{\delta(\overline{\mathcal{P}},\kappa)}$ for each shard. It implies that Manifoldchain equipped with the USF mechanism enables each shard to achieve a throughput equivalent to Bitcoin's, and to linearly scale the total throughput with the number of shards $m$.

We next consider the Manifoldchain protocol, $\prod_{mfd}(f_{bcsf}, \cdot, \cdot)$, which employs the BCSF protocol $f_{bcsf}$. This protocol allows for stragglers to be gathered into a single shard, thereby satisfying $\Delta_i < \Delta$ for some $i$. For a target shard $j$ where $\Delta_j < \Delta$, we find that $\gamma_j > 0$, as per Ineq.\ref{constrain}. Consequently, $\prod_{mfd}(f_{bcsf}, \cdot, \cdot)$ achieves a bounded throughput of $\frac{\rho}{\Delta}\log\frac{\rho}{\delta(\overline{\mathcal{P}}, \kappa)} + \frac{\gamma_j\rho_j}{\Delta}\log\frac{\rho}{\delta(\overline{\mathcal{P}}, \kappa)}$ in shard $j$. This represents an additional throughput contribution of $\frac{\gamma_j\rho_j}{\Delta}\log\frac{\rho}{\delta(\overline{\mathcal{P}}, \kappa)}$ compared with $\prod_{mfd}(f_{usf}, \cdot, \cdot)$. To quantify the magnitude of this throughput improvement in the Manifoldchain protocol within shard $j$, we define a metric $\mathcal{I}np(\prod_{mfd}, j)$:

\begin{equation}
\begin{split}
\mathcal{I}np(\prod_{mfd}, j) =& \frac{\gamma_j\rho_j}{\Delta}\log\frac{\rho}{\delta(\overline{\mathcal{P}}, \kappa)} / \frac{\rho}{\Delta}\log\frac{\rho}{\delta(\overline{\mathcal{P}}, \kappa)} \\
=& \frac{\gamma_j\rho_j}{\rho} > 0, \\
where&\quad \sum_j^m{\rho_j} = m\rho,\\
&\frac{1}{\gamma_j + 1} \frac{\log \frac{m\rho}{\delta(\overline{\mathcal{P}}, \kappa)(\gamma_j + m)}}{\log \frac{\rho}{\delta(\overline{\mathcal{P}}, \kappa)}} \geq \frac{\Delta_j}{\Delta}
\end{split}
\end{equation}

Hence, we denote by $\overline{\mathcal{I}np}(\prod_{mfd}) = \sum_{i}^m \mathcal{I}np(\prod_{mfd}, i)$ the aggregated improvement across all shards. Note that the honest ratio $\rho_j$ has a positive correlation with the improvement in shard $j$. Adversarial miners may concentrate their hashing power in a single shard to diminish the honest ratio $\rho_j$ of that shard, subsequently undermine the throughput improvement. However, this strategy does not efficiently reduce the total throughput improvement across all shards, given that the ratio of the total adversarial hashing power is constrained by $1-\rho$. We then estimate a readily attainable throughput improvement, $\mathcal{I}np'(\prod_{mfd}, j)$, by simply setting $\rho_j = \rho$. In a large-scale blockchain system with an overwhelmingly large $m \gg \gamma_j$, we estimate the maximum available $\gamma_i$ via the following equation.

\begin{equation}
\begin{split}
\frac{1}{\gamma_j + 1}\frac{\log \frac{m\rho}{\delta(\overline{\mathcal{P}}, \kappa)(\gamma_j + m)}}{\log\frac{\rho}{\delta(\overline{\mathcal{P}},\kappa)}} \approx \frac{1}{\gamma_j + 1} \geq \frac{\Delta_j}{\Delta}.  
\end{split}
\end{equation}

Consequently, we present a readily available improvement as follows:

\begin{equation}
\begin{split}
\mathcal{I}np'(\prod_{mfd}, j) \leq& \frac{\gamma_j\rho_j}{\rho} \\
\approx& \gamma_j \\
\approx& \frac{\Delta}{\Delta_j} - 1.
\end{split}
\end{equation}

The improvement increases almost linearly with the network delay gap $\frac{\Delta}{\Delta_j}$. In a shard where $\Delta_j = \frac{1}{5}\Delta$, the protocol $\prod_{mfd}(f_{bcsf}, \cdot, \cdot)$ achieves a throughput improvement approaching $400\%$ when compared with $\prod_{mfd}(f_{usf}, \cdot, \cdot)$.

% \section{Lemmas}

% In this section, we introduce several significant lemmas that are essential for the subsequent proofs of the three fundamental security properties held by Manifoldchain. 

\section{Honest Presence}\label{honest_presence_section}

We first present a lemma utilized for calculating the error probability of honest presence.

\begin{lemma}\label{probability_of_honest_existence_lemma}
    Given $n$ miners miners randomly distributed into $m$ shards, the probability $\varepsilon(n)$ there is no honest miners in some shards is bounded by $(m-1)(\frac{m-1}{m})^{n-1}$.
\end{lemma}

\begin{proof}
    The distribution of $n$ miners into $m$ shards satisfies a Multinomial Distribution\cite{multinomial_distribution}. Specifically, it models the probability of counts for each side of a $m$-sided dice rolled $n$ times. For $n$ independent trials each of which leads to a success for exactly one of $m$ categories, with each category having a given fixed success probability, the multinomial distribution gives the probability of any particular combination of numbers of successes for the various categories. Mathematically, for each independent trial, we have $m$ possible mutually exclusive outcomes, with corresponding probabilities $p_0, ..., p_{m-1}$. Given a random distribution, $p_0=p_1=...=p_{m-1}=\frac{1}{m}=p$. The probability mass function of this multinomial distribution is:
    \begin{equation}
        \begin{split}
            f(x_0,...,x_{m-1};n,p) &= Pr[X_0=x_0, ..., X_m=x_m] \\
            &=\frac{n!}{x_0!\cdot ...\cdot x_{m-1}!}p^{x_0}\times ... \times p^{x_{m-1}},\\
            &=\frac{n!}{x_0!\cdot ...\cdot x_{m-1}!}p^{n},
        \end{split}
    \end{equation}
    where $\sum_{i=0}^{m-1}x_i = n$. 

    The probability that there is at least one honest miner in each shard is denoted as $Pr[x_0 \geq 1, x_1 \geq 1, ..., x_{n-1} \geq 1]$. Hence, 

    \begin{equation}
        \begin{split}
            Pr[x_0 \geq 1, &x_1 \geq 1, ..., x_{m-1} \geq 1] \\
            \geq& 1 - \sum_{x_1, x_2, ..., x_{m-1}} Pr[x_0=0, x_1 \geq 0, ..., x_{m-1} \geq 0] \\
            &- \sum_{x_0, x_2, ..., x_{m-1}} Pr[x_0\geq 0, x_1 = 0, ..., x_{m-1} \geq 0] \\
            &- ... \\
            &- \sum_{x_0, x_1, ..., x_{m-2}} Pr[x_0\geq 0, ..., x_{m-2} \geq 0, x_{m-1}=0] \\
            =& 1 - \sum_{x_1, x_2, ..., x_{m-1}} \frac{n!}{x_1!\cdot ... \cdot x_{m-1}}p^{n} \\
            &- ...\\
            &- \sum_{x_0, x_1, ..., x_{m-2}} \frac{n!}{x_0!\cdot ... \cdot x_{m-2}}p^{n} \\
            =& 1 - \frac{p^n}{{p'}^n}\sum_{x_1, x_2, ..., x_{m-1}} \frac{n!}{x_1!\cdot ... \cdot x_{m-1}}{p'}^{n} \\
            &- ...\\
            &- \frac{p^n}{{p'}^n}\sum_{x_0, x_1, ..., x_{m-2}} \frac{n!}{x_0!\cdot ... \cdot x_{m-2}}{p'}^{n}\\
            =& 1 - n \cdot \frac{p^n}{{p'}^n}\sum_{x_0, x_1, ..., x_{m-2}} \frac{n!}{x_0!\cdot ... \cdot x_{m-2}}{p'}^{n},
        \end{split}
    \end{equation}
    where $p' = \frac{1}{m-1}$, $\frac{n!}{x_0!\cdot ... \cdot x_{m-2}}{p'}^{n}$ represents the probability mass function of another multinomial distribution where $n$ miners are randomly distributed into $m-1$ shards. Given 
    \begin{equation}
        \begin{split}
            \sum_{x_0, x_1, ..., x_{m-2}} \frac{n!}{x_0!\cdot ... \cdot x_{m-2}}{p'}^{n} = 1,
        \end{split}
    \end{equation}we have
    \begin{equation}
        \begin{split}
            Pr[x_0 \geq 1, &x_1 \geq 1, ..., x_{m-1} \geq 1] \\
            \geq& 1-m\cdot \frac{p^n}{{p'}^n} \\
            =& 1 - m\cdot \frac{(\frac{1}{m})^n}{(\frac{1}{m-1})^n} \\
            =& 1 - (m-1)(\frac{m-1}{m})^{n-1}.
        \end{split}
    \end{equation}
    Subsequently, $\varepsilon(n) \leq (m-1)(\frac{m-1}{m})^{n-1}$ get proved.
    
\end{proof}

Now we can utilize this lemma to prove Lemma \ref{honest_presence_proof}.

\begin{proof}
In the static setting, adopting the shard partition mechanism in Section \ref{formation}, each miner has an equal probability of being assigned to any Y-region. Once assigned to a specific Y-region, it also has an equal probability of being assigned to each shard. We model this process by uniformly distributing $\rho N$ honest miners across $m = S_X \cdot S_Y$ shards. According to Lemma~\ref{probability_of_honest_existence_lemma} in Appendix \ref{honest_presence_section}, the error probability $\varepsilon'(\rho N)$ of an absence of honest miners in certain shards adheres to an upper bound $(m-1)(\frac{m-1}{m})^{\rho N-1}$. Additionally, without preemptive bandwidth information, the adversary cannot selectively corrupt miners based on specific bandwidths. Randomly selecting miners for corruption makes it nearly impossible for the adversary to corrupt all miners within a shard.%\xw{did you define negligible functions? and it should be $\varepsilon(N)$?}. 

In the dynamic setting, the protocol implements periodic rotations to generate a new credential-chain for shard re-allocation. 
% {\old Under Assumption~\ref{new_miners_assumption}, there are at least $G$ new miners in each shard formation phase. Similar to the static setting, these new miners are uniformly distributed into $m$ shards. The error probability $\varepsilon''(G)$ of an absence of honest miners in certain shards adheres to an upper bound $(m-1)(\frac{m-1}{m})^{G-1}$.} 
There are at least $\underline{\alpha} \rho N$ new honest miners in each shard formation phase. Similar to the static setting, these new miners are uniformly distributed into $m$ shards. The error probability $\varepsilon''(\underline{\alpha} \rho N)$ of an absence of honest miners in certain shards adheres to an upper bound $(m-1)(\frac{m-1}{m})^{\underline{\alpha} \rho N-1}$. The mildly adaptive adversary cannot transfer corrupted miners within a single rotation round. Moreover, during the shard formation phase, when new miners join, the adversary cannot corrupt miners with specific bandwidths due to a lack of preemptive bandwidth information about these miners. Therefore, despite the adversary's attempt to corrupt an entire shard through the transfer of corrupted miners, it ultimately fails since new miners include honest miners. 

Combined $\varepsilon'(\rho N)$ and $\varepsilon''(\underline{\alpha} \rho N)$, the error probability is bounded by $\varepsilon(\underline{\alpha} \rho N)$.
\end{proof}

\noindent {\bf Imperfect bandwidth estimation}. Given slight deviations in bandwidth estimation, the $n$ miners may not be perfectly evenly distributed across the $m$ shards. In shard $i$, it's possible that $p_i < p = \frac{1}{m}$. To compute the upper bound of the error probability in such scenarios, we consider a new case where $p_i'$ represents the minimum $p_i$ from the original distribution, i.e.,$p_i' = \min_i p_i < \frac{1}{m}$. Consequently, $m' = \frac{1}{p'} > m$, which can replace $m$ in Lemma \ref{probability_of_honest_existence_lemma} to calculate the new error probability. This error probability is still negligible with an overwhelmingly large $n$.

\section{Proof of Security Properties}\label{proof_of_security_properties}

In this section, we prove that Manifoldchain hold the three aforementioned security properties. Recalling the categorization of blocks into three types, namely exclusive block, inclusive block, and transaction block. We note that the former two only contain compressed information and hence maintain a constant block size of $\mathcal{O}(1)$, while the transaction block encompasses $\mathtt{b}$ specific transactions, resulting in a size of $\mathcal{O}(B)$. In most scenarios, the size of a transaction block is significantly larger compared to that of exclusive and inclusive blocks. Consequently, the broadcasting time for exclusive block and inclusive block is deemed negligible in comparison to that of transaction blocks. Note that a transaction block is generated concomitantly with its corresponding exclusive/inclusive block. We define a block as a full block when it is composed a matching pair of exclusive/inclusive block and transaction block. Also, we refer to the inclusive blocks from other shards as light blocks, and assume that the delivery time for a light block is negligible or zero. We use the term abstract blocks to encompass both full blocks and light blocks.

\subsection{Severe Execution Environment}\label{severe_execution_environment_section}

We consider a severe execution environment $\tilde{\mathcal{E}nv}(\Delta, \rho)$, where any delivered full block is delayed for $\Delta$ seconds and during the delivering any mining operation is considered invalid.

\begin{definition}\label{severe_execution_environment}
Given a severe execution environment $\tilde{\mathcal{E}nv}(\Delta, \rho)$ and the corresponding original environment $\mathcal{E}nv(\Delta, \rho)$, the execution of $\prod_{mfd}(f, \lambda_s, \{\lambda_i\}^m)$ in $\tilde{\mathcal{E}nv}(\Delta, \rho)$ remains the same as the execution in $\mathcal{E}nv(\Delta, \rho)$ except the following hypothetical operations:
\begin{enumerate}
    \item All honest full blocks in shard $i$ are virtually delayed for $\Delta_i$ seconds;
    \item Whenever a honest full block is virtually delayed, any PoW trials by honest miners within the same shard are considered invalid and return $\mathtt{false}$;
    \item If two or more abstract blocks are mined at the same time, ignore one of them and randomly discard the others.
\end{enumerate}
Through sharding, it satisfies that $\{\tilde{\mathcal{E}nv}(\Delta_i, \rho_i)\}^m = f(\tilde{\mathcal{E}nv}(\Delta, \rho))$ and $\{\mathcal{E}nv(\Delta_i, \rho_i)\}^m = f(\mathcal{E}nv(\Delta, \rho))$.
\end{definition}

The severe execution environment eliminates all forking situations in the original execution environment, ensuring that each honest abstract block is positioned at distinct heights. However, the power of the adversary is amplified in the severe execution environment. Because the delivery time of full block is fixed to the maximum value and honest miners freeze during the delivery time, fewer honest blocks are generated within a given time interval. Consequently, it becomes relatively easier for the adversary to generate a longer chain and violate the consistency.

We present our first lemma as follows:

\begin{lemma}\label{distinct_height}
If $\prod_{mfd}$ is executed in $\tilde{\mathcal{E}nv}(\Delta, \rho)$, then within any given shard, the following statements hold:
\begin{enumerate}
    \item Any honest abstract blocks mined after time $r$ are higher than honest blocks mined prior to time $r$.
    \item All honest abstract blocks are uniquely located at different heights.
\end{enumerate}
\end{lemma}

\begin{proof}
Assuming for any honest miners $i$ and $j$, there exist $t_1$ and $t_2$ such that $t_1 < t_2$, miner $i$ mines an abstract block at time $t_1$, and miner $j$ mines an abstract block at time $t_2$. It holds that $len(\mathcal{C}_i^{t_1}) \geq len(\mathcal{C}_j^{t_2})$. Then we consider the following cases

\begin{itemize}
\item If $\mathtt{head}(\mathcal{C}_{i}^{t_1})$ is a full block:
\begin{enumerate}
    \item If $t_1 > t_2 - \Delta_i$, it implies that $\mathtt{head}(\mathcal{C}_{j}^{t_2})$ is mined during the delivery of $\mathtt{head}(\mathcal{C}{i}^{t_1})$. This contradicts the hypothetical operations (1) and (2) in $\tilde{\mathcal{E}nv}(\Delta, \rho)$.
    \item If $t_1 = t_2 - \Delta_i$, it means that $\mathtt{head}(\mathcal{C}_{i}^{t_1})$ and $\mathtt{head}(\mathcal{C}_{i}^{t_1})$ are mined at the same time. This contradicts the hypothetical operation (3) in $\tilde{\mathcal{E}nv}(\Delta, \rho)$.
    \item If $t_1 < t_2 - \Delta_i$, $\mathcal{C}_{i}^{t_1}$ appears in the view of each honest miner since time $t_2$. However, $\mathtt{head}(\mathcal{C}_{j}^{t_2})$ is adversarial because it does not extend the longest chain. This contradicts the assumption.
\end{enumerate}
\item If $\mathtt{head}(\mathcal{C}_{i}^{t_1})$ is a light block, $\mathcal{C}_{i}^{t_1}$ appears in the view of each honest miner after time $t_1$. $\mathtt{head}(\mathcal{C}_{j}^{t_2})$ is adversarial because it does not extend the longest chain. This contradicts the assumption.
\end{itemize}

Therefore, Lemma \ref{distinct_height} is proven.
\end{proof}

% \subsection{Private-mining Attack as the Most Dangerous Attack}

It is challenging to comprehensively discuss all types of attack strategies. However, we demonstrate that the private-mining attack stands out as the most efficient strategy within a severe environment. Specifically, when targeting a transaction $\mathtt{tx}$ that emerges in a newly minted block at time $r$, the private-mining attack includes two primary phases that undermine the consistency of $\mathtt{tx}$. During the initial phase before time $r$, the adversary attempts to establish an advantage over honest miners. Subsequently, in the subsequent phase commencing at time $r$, the adversary engages in a race against honest miners, striving to mine a private chain with a larger length of at least $len(\mathcal{C}^r) + \kappa + 1$ than that of public longest chain. Here, $\mathcal{C}^r$ represents the longest publicly available chain at time $r$. Formally, we define the ``advantage'' of the adversary as follows:

\begin{definition}
The adversary's advantage $\mathcal{A}dv_t$ at time $t$ is defined as the gap between the length of the longest chain mined by time $t$ and the length of the longest chain mined by honest miners up to time $t$.
\end{definition}

$\mathcal{A}dv$ is non-negative due to its definition. For convenience, we define a chain as being \textit{public} at time $t$ if it is present in the view of all honest miners at that time. Additionally, we define a chain as \textit{trustworthy} at time $t$ if its validated sub-chain is at least as long as the validated sub-chain of any public chain present at time $t$. Specifically, the private-mining attack executes as follows:
\begin{itemize}
    \item (\textbf{Cheating Phase}.) Up to time $r$, the adversary mines on the longest chain to try to increase its advantage. 
    \item (\textbf{Race Phase}.) From time $r$ onwards, the adversary tries to mine a private chain that deliberately excludes the target transaction $\mathtt{tx}$. If $\mathtt{tx}$ becomes confirmed and the private chain is longer than a trustworthy chain, the adversary publishes the private chain, thereby successfully violating the consistency of $\mathtt{tx}$. If this event never occur, the adversary fail to violate the consistency.
\end{itemize}

It can be demonstrated straightforwardly that the private-mining attack maximizes the advantage until time $r$. Based on the given definition, the advantage increases by a maximum of 1 when an adversarial block is mined, and decreases by 1 if it is positive when a honest block is mined. The private mining ensures that the advantage increase exactly by 1 whenever an adversarial block is mined, and decrease exactly by 1 if it is positive upon the mining of a honest block. Therefore, the private-mining attack maximizes the advantage. 

We now present our theorem, which demonstrates that the private-mining attack is the most efficient strategy for compromising the consistency of a target transaction within a severe execution environment.

\begin{theorem}\label{private_mining_worse}
Consider the Manifoldchain protocol $\prod_{mfd}$ and a severe execution environment $\tilde{\mathcal{E}nv}(\Delta,\rho)$. Within any shard $i$, if there exists any attack that successfully violates the consistency of a target transaction $\mathtt{tx}$, then the private-mining attack also achieves the violation of $\mathtt{tx}$'s consistency.
\end{theorem}

\begin{proof}
Let $h_{\mathcal{B}}$ denote the height of an abstract block $\mathcal{B}$. For 
convenience, we refer to the abstract block as the block in this proof segment. Suppose there is a hypothetical attack that compromises the consistency of $\mathtt{tx}$ included in block $\mathtt{b}$. Under the hypothetical attack, we consider two \textbf{minimum} chains $\mathtt{c}$ and $\mathtt{d}$, which leads to the consistency violation event. These two chains satisfies the following conditions:

\begin{itemize}
\item \textbf{Condition 1:} chain $\mathtt{c}$ contains block $b$ that includes $\mathtt{tx}$ and has a length of $len(\mathtt{c}) \geq h_{\mathtt{b}} + \kappa - 1$. Chain $\mathtt{c}$ is trustworthy at time $t_{\mathtt{c}}$.
\item \textbf{Condition 2:} chain $\mathtt{d}$ does not contain block $\mathtt{b}$ and does not include $\mathtt{tx}$ in any block at heights up to $h_{\mathtt{b}} - 1$, while maintaining a length of $len(\mathtt{d}) \geq h_{\mathtt{b}} + \kappa - 1$. Chain $\mathtt{d}$ is trustworthy at time $t_{\mathtt{d}}$.
\end{itemize}

We define the larger length between the two as $len_+$, and the later time as $\tau_+$, which satisfies:

\begin{equation}
\begin{split}
&len_+ = \min\limits_{\mathtt{c}, \mathtt{d}}\max(len(\mathtt{c}), len(\mathtt{d})) \geq h_{\mathtt{b}} + \kappa - 1, \\
&\tau_+ = \max(t_{\mathtt{c}}, t_{\mathtt{d}}).
\end{split}
\end{equation}

Supposed that $\mathtt{b}$ is mined after $\tau$ (mined during $(\tau, \tau_+]$). Notice that $b$ may not be the highest honest block after time $\tau$. Let block $\mathtt{u}$ be the highest honest block mined by time $\tau$, and block $\mathtt{a}$ be the highest block on chain $\mathtt{d}$ at time $\tau$. By definition we have $\mathcal{A}dv_{\tau} + h_{\mathtt{u}} \geq h_{\mathtt{a}}$. We consider the following two situations, illustrated by Fig. \ref{fig:case_1}:

\begin{itemize}
    \item $h_{\mathtt{u}} \leq h_{\mathtt{b}} - 1$;
    \item $h_{\mathtt{u}} \geq h_{\mathtt{b}}$.
\end{itemize}

\begin{figure*}
  \centering
  \includegraphics[width=14cm]{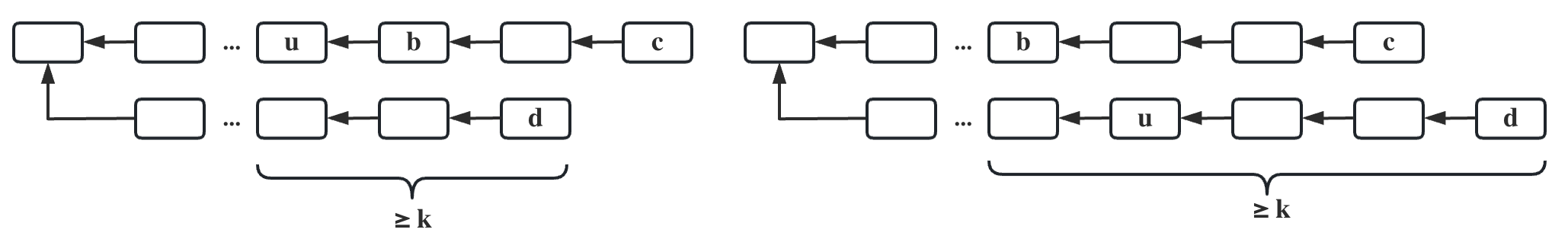}
  \caption{$h_\mathtt{u} \leq h_\mathtt{b} + 1$ or $h_\mathtt{u} \geq h_\mathtt{b}$}
  \label{fig:case_1}
\end{figure*}

Let $H_{j,k}$ be the number of honest blocks mined during $(j, k]$, and $A_{j,K}$ be the number of adversarial blocks during $(j, k]$, we now present our second lemma:

\begin{lemma}
$H_{\tau, \tau_+ - \Delta_i} \leq h_+ - h_{\mathtt{u}}$.
\end{lemma}

\begin{proof}
Assuming $H_{\tau, \tau_+ - \Delta_i} \geq h_+ - h_{\mathtt{u}} + 1$. Since $\mathtt{u}$ is the highest honest block mined by time $\tau$. According to Lemma \ref{distinct_height}, there must exist a public chain with a length of at least $h_{\mathtt{u}} + h_+ - h_{\mathtt{u}} + 1 = h_+ + 1$ at time $\tau_+$. In this case, either chain $\mathtt{c}$ or chain $\mathtt{d}$ is not trustworthy, which contradicts its definition.
\end{proof}

Then we present next lemma as follows:

\begin{lemma}\label{adversarial_base}
There exist an adversarial block mined during $(\tau, \tau_+]$ on each height from $h_{\mathtt{a}} + 1$ to $h_+$.
\end{lemma}

\begin{figure}
  \centering
  \includegraphics[width=8cm]{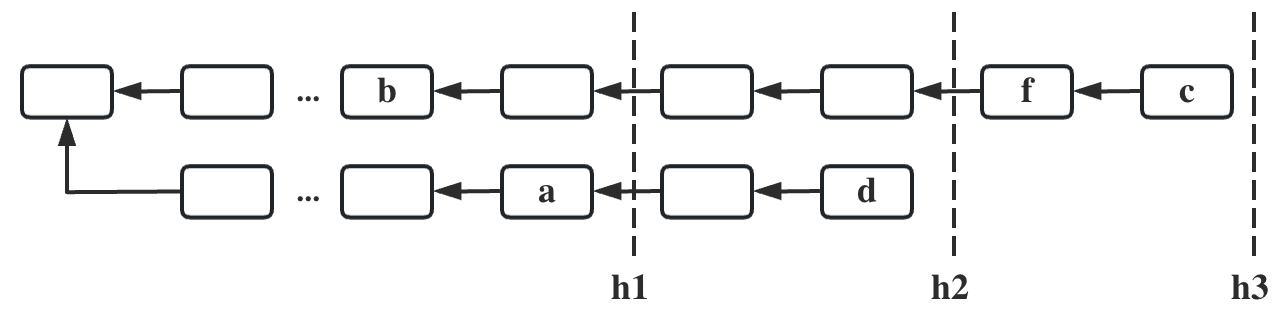}
  \caption{$h_\mathtt{a} \geq h_{\mathtt{b}} - 1, h_{\mathtt{c}} \geq h_{\mathtt{d}}$}
  \label{fig:case_2}
\end{figure}

\begin{figure}
  \centering
  \includegraphics[width=8cm]{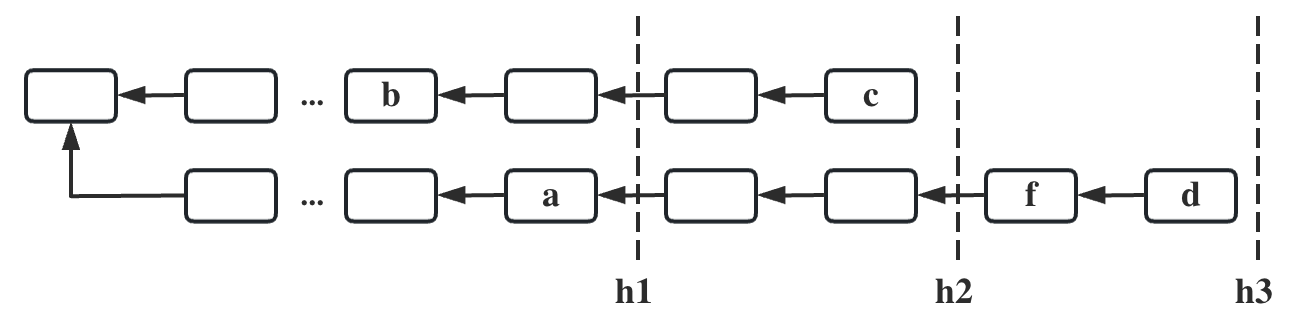}
  \caption{$h_\mathtt{a} \geq h_{\mathtt{b}} - 1, h_{\mathtt{c}} < h_{\mathtt{d}}$}
  \label{fig:case_3}
\end{figure}

\begin{figure}
  \centering
  \includegraphics[width=8cm]{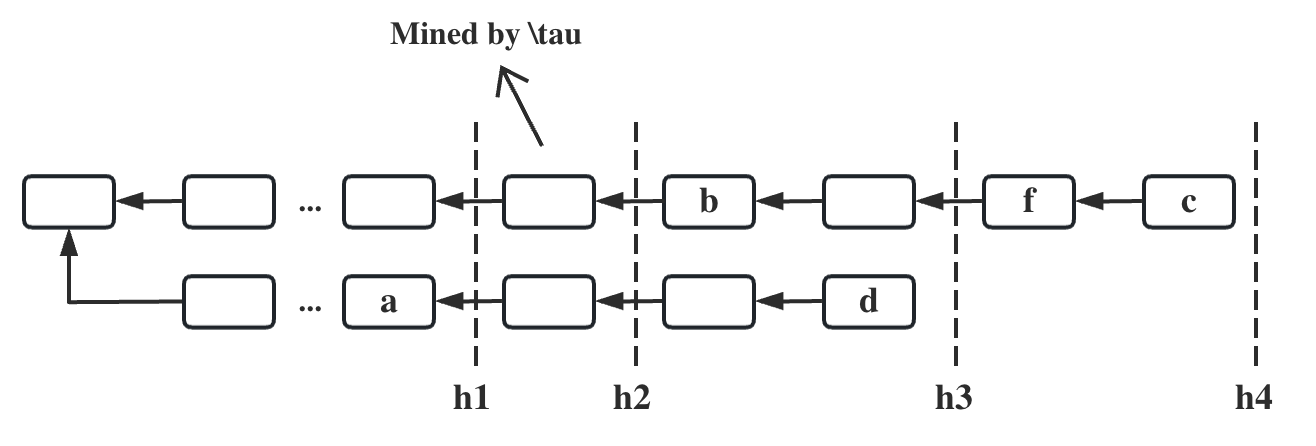}
  \caption{$h_{\mathtt{a}} < h_{\mathtt{b}} - 1$}
  \label{fig:case_4}
\end{figure}

\begin{proof}
We prove this lemma case by case. Supposed $\mathtt{e}$ represents the longer chain among $\mathtt{c}$ and $\mathtt{d}$.

\textbf{Case 1:} $h_{\mathtt{a}} \geq h_{\mathtt{b}} - 1$. The proof methodology is shown in Fig. \ref{fig:case_2} and Fig. \ref{fig:case_3}.

\begin{itemize}
    \item From $h_1$ to $h_2$, every height contains an adversarial block mined during $(\tau, \tau_+]$ because honest blocks locate at different heights.
    \item From $h_2$ to $h_3$, we assume that there is a honest block $\mathtt{f}$ locating at these heights in $\mathtt{e}$. Supposed $\mathcal{C}' = \mathcal{C}^{\rceil 1}$, where $head(\mathcal{C}) = \mathtt{f}$. As $\mathcal{C}'$ is always trustworthy before $\mathcal{C}$, $\mathcal{C}'$ becomes a shorter chain than $\mathtt{e}$ which satisfies Condition 1 and Condition 2 before $\mathtt{e}$ is trustworthy. Thus, $\mathtt{e}$ contradicts the definition of being one of the two minimum chains satisfying Condition 1 and Condition 2.
\end{itemize} 

\textbf{Case 2:} $h_{\mathtt{a}} < h_{\mathtt{b}} - 1$. Fig. \ref{fig:case_4} demonstrates the proof methodology. Through Case 1, we have proved that every height from $h_2$ to $h_4$ contains an adversarial block mined during $(\tau, \tau_+]$. From $h_1$ to $h_2$, only blocks in $\mathtt{d}$ are mined during $(\tau,\tau_+]$. These blocks are adversarial because they does not include $\mathtt{tx}$ according to the definition, while honest blocks would have included $\mathtt{tx}$.

Thus we complete the proof of Lemma \ref{adversarial_base}. 
\end{proof}

Furthermore, we present another significant lemma as follows:

\begin{lemma}\label{adv_count}
$A_{\tau, \tau_+} \geq h_+ - h_{\mathtt{a}} + \max(h_{\mathtt{u}} - h_{\mathtt{b}} + 1, 0)$.
\end{lemma}

\begin{figure}
  \centering
  \includegraphics[width=8cm]{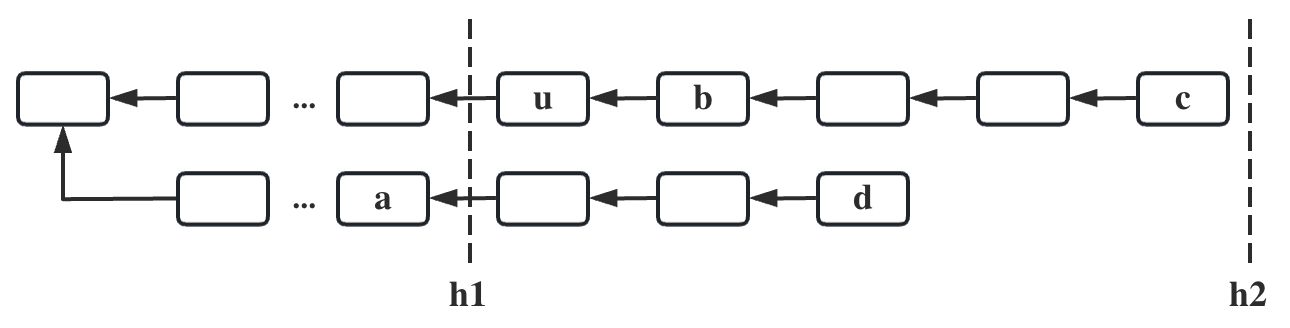}
  \caption{$h_{\mathtt{u}} \leq h_{\mathtt{b}}-1$}
  \label{fig:case_5}
\end{figure}

\begin{figure}
  \centering
  \includegraphics[width=8cm]{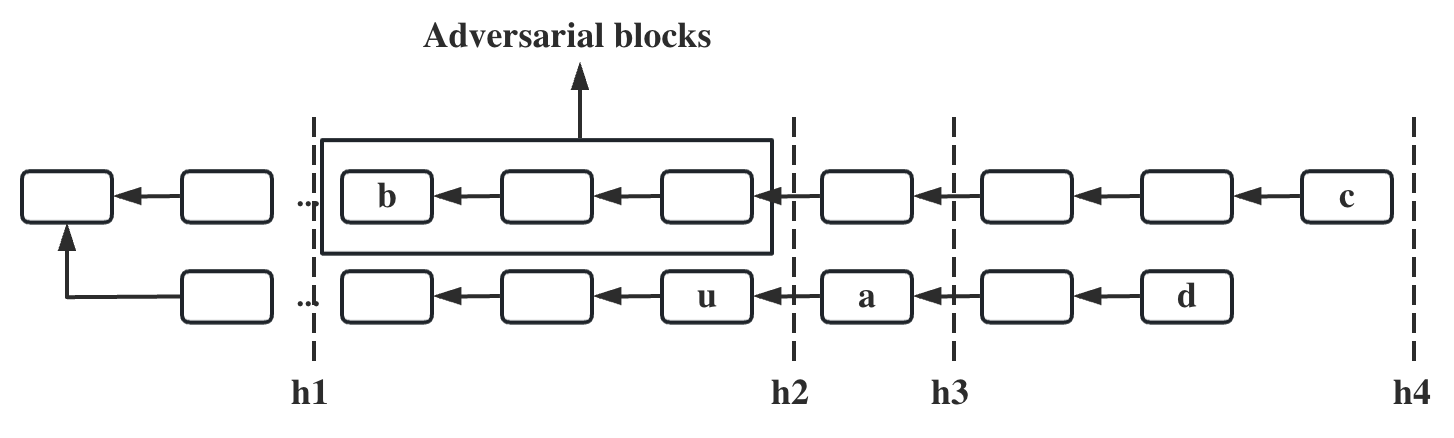}
  \caption{$h_{\mathtt{u}} \geq h_{\mathtt{b}}$ and $h_{\mathtt{u}} \leq h_{\mathtt{a}}$}
  \label{fig:case_6}
\end{figure}

\begin{figure}
  \centering
  \includegraphics[width=8cm]{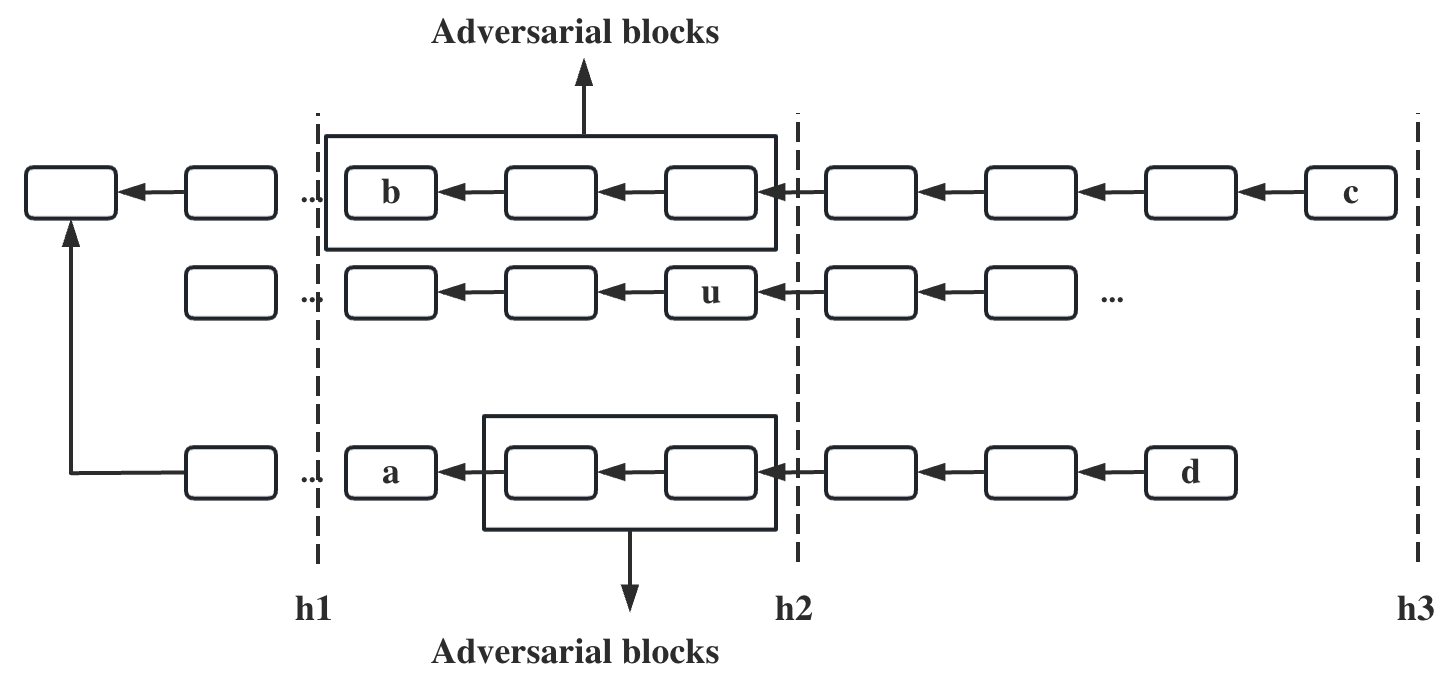}
  \caption{$h_{\mathtt{u}} \geq h_{\mathtt{b}}$ and $h_{\mathtt{u}} > h_{\mathtt{a}}$}
  \label{fig:case_7}
\end{figure}

\begin{proof}
We prove this lemma case by case.

\item \textbf{Case 1:} $h_{\mathtt{u}} \leq h_{\mathtt{b}} - 1$, which is demonstrated in Fig. \ref{fig:case_5}. According to Lemma \ref{distinct_height}, there exists an adversarial block at every height from $h_{\mathtt{a}}$ to $h_+$. There are a total of $h_+ - h_{\mathtt{a}}$ adversarial blocks mined during the interval $(\tau, \tau_+]$ from $h_1$ to $h_2$.

\item \textbf{Case 2:} $h_{\mathtt{u}} \geq h_{\mathtt{b}}$ and $h_{\mathtt{u}} \leq h_{\mathtt{a}}$, illustrated in Fig. \ref{fig:case_6}. From $h_3$ to $h_4$, there are at least $h_+ - h_{\mathtt{a}}$ adversarial blocks due to Lemma \ref{distinct_height}. Because block $\mathtt{u}$ is the highest honest block by time $\tau$, any honest block mined after $\tau$ must higher than $h_{\mathtt{u}}$. Thus the $h_{\mathtt{u}} - h_{\mathtt{b}} + 1$ blocks in $\mathtt{c}$ of height between $h_1$ and $h_2$ are all adversarial. There are totally $h_+ - h_{\mathtt{a}} + (h_{\mathtt{u}} - h_{\mathtt{b}} + 1)$ adversarial blocks mined during $(\tau, \tau_+]$.

\item \textbf{Case 3:} $h_{\mathtt{u}} \geq h_{\mathtt{b}}$ and $h_{\mathtt{u}} > h_{\mathtt{a}}$, demonstrated in Fig. \ref{fig:case_7}. From $h_2$ to $h_3$, there are at least $h_+ - h_{\mathtt{u}}$ adversarial blocks, as stated in Lemma \ref{distinct_height}. Similarly to Case 2, any block mined after $\tau$ with a height not greater than $h_{\mathtt{u}}$ is adversarial. Consequently, there are at least $h_{\mathtt{u}} - h_{\mathtt{a}} + (h_{\mathtt{u}} - h_{\mathtt{b} + 1})$ adversarial blocks from $h_1$ to $h_2$. In total, there are $h_+ - h_{\mathtt{a}} + (h_{\mathtt{u}} - h_{\mathtt{b}} + 1)$ adversarial blocks mined during the interval $(\tau, \tau_+]$.

Therefore, we have done the proof of Lemma \ref{adv_count}.
\end{proof}
We are now prepared to establish the success of the private-mining attack in this scenario. Let $\mathtt{u}'$ be the highest block by time $\tau$, and let $\mathcal{A}dv_{\tau}'$ denote the advantage obtained by time $\tau$ through the private-mining attack. The target transaction $\mathtt{tx}$ is included in the first honest block with a height of $h_{\mathtt{u}}' + 1$ after $\tau$. Let $\mathcal{C}_H$ represent the public honest chain after $\tau+$, and $\mathcal{C}_A$ represent the adversarial private chain after $\tau$. It is evident that $|\mathcal{C}_H| \leq h{\mathtt{u}}' + H{\tau, \tau_+-\Delta}$ and $|\mathcal{C}_A| \geq h{\mathtt{u}}' + \mathcal{A}dv_{\tau}' + A_{\tau, \tau_+}$. Therefore,

\begin{equation}\label{AdvIsLonger}
\begin{split}
|\mathcal{C}_A| - |\mathcal{C}_H| \geq&  h_{\mathtt{u}'} + \mathcal{A}dv_{\tau}' + A_{\tau, \tau_+} - (h_{\mathtt{u}'} + H_{\tau, \tau_+-\Delta}) \\
\geq& \mathcal{A}dv_{\tau}' + (h_+ - h_{\mathtt{a}} + \max(h_{\mathtt{u}} - h_{\mathtt{b}} + 1, 0))\\
&- (h_+ - h_{\mathtt{u}}) \\
\geq&  \mathcal{A}dv_{\tau} + h_{\mathtt{u}} - h_{\mathtt{a}}\\
\geq& 0.
\end{split}
\end{equation}
\end{proof}

Besides,

\begin{equation}\label{len_k}
\begin{split}
|\mathcal{C}_A| - h_{\mathtt{u}'} =& \mathcal{A}dv' + h_+ - h_{\mathtt{a}} + \max(h_{\mathtt{u}} - h_{\mathtt{b}} + 1, 0) \\
\geq& \mathcal{A}dv + h_+ - h_{\mathtt{a}} + h_{\mathtt{u}} - h_{\mathtt{b}} + 1 \\
\geq& (\mathcal{A}dv + h_{\mathtt{u}} - h_{\mathtt{a}}) + (h_+ - h_{\mathtt{b}} + 1) \\
\geq& \kappa.
\end{split}
\end{equation}

Ineq. \ref{AdvIsLonger} ensures that the length of the private chain is always greater than or equal to any public honest chain after $\tau_+$, while Ineq. \ref{len_k} guarantees that the private chain contains a fork that excludes $\mathtt{tx}$ with a length of at least $\kappa$. By definition the private-mining attack always succeed under these two guarantees. Hence, we have completed the proof of Theorem \ref{private_mining_worse}.

\subsection{Proof of Common Prefix Property ( Theorem \ref{cp_property})}\label{CP_proof}

Since it is more difficult to achieve consistency in a severe execution environment than in a typical execution environment, we proved that Theorem \ref{cp_property} holds true in the severe execution environment, thereby complete the proof.

Consider the Manifoldchain protocol $\prod_{mfd}(f, \lambda_s, \{\lambda_i\}^m)$ and a severe execution environment $\tilde{\mathcal{E}nv}(\Delta, \rho)$. Given a shard $i$, we consider a hypothetical block solving the PoW puzzle, which includes the following three types of labels:

\begin{itemize}
    \item \textbf{Label 1}: It is mined by honest miner or adversarial miner.
    \item \textbf{Label 2}: It is a full block or a light block.
    \item \textbf{Label 3}: It is considered invalid due to the severe environment's rule or not. For convenience, we say the block is severely invalid.
\end{itemize}

A block is considered exactly honest if it is mined by honest miner while being considering valid in the severe environment. Given the consensus mining rate $\lambda_i$ and the sharing mining rate $\lambda_s$, a total of $\lambda_s + \lambda_i$ full blocks are generated per second within shard $i$, while $(\mathtt{m}-1)$ light blocks are received from the other $\mathtt{m}-1$ shards. Therefore, the probability that a block is a full block, denoted as ${\rm Pr}_{full}^i$, is calculated as ${\rm Pr}_{full}^i = \frac{\lambda_s + \lambda_i}{\mathtt{m}\lambda_s + \lambda_i}$. Similarly, the probability that a block is a light block, denoted as ${\rm Pr}_{light}^i$, is given by ${\rm Pr}_{light}^i = \frac{(\mathtt{m}- 1)\lambda_s}{\mathtt{m}\lambda_s + \lambda_i}$. Among the full blocks, the ratio of the blocks mined by honest miners is $\rho_i$. The same honest ratio among the light blocks is $\frac{\mathtt{m}\rho - \rho_i}{\mathtt{m} -1}$. Let ${\rm Pr}_h$ represent the probability that a hypothetical block is mined by honest miner, it is calculated as 

\begin{equation}
\begin{split}
{\rm Pr}_h = {\rm Pr}_{full}^i \cdot \rho_i + {\rm Pr}_{light}^i\cdot \frac{\mathtt{m}\rho - \rho_i}{\mathtt{m} - 1}.
\end{split}
\end{equation}

Recalling that the generation of hypothetical blocks follows a Poisson process, it is important to note that each Poisson process within one second is independent and identically distributed, sharing the same rate parameter. Supposed $\rm X$ and $\rm Y$ represents the number of full blocks and light blocks generated within one second respectively, the corresponding probability mass function (PMF) are as follows:

\begin{equation}
\begin{split}
&{\rm Pr}({\rm X} = k) = \frac{e^{-(\lambda_s + \lambda_i)}(\lambda_s + \lambda_i)^k}{k!} \\
&{\rm Pr}({\rm Y} = k) = \frac{e^{-(\mathtt{m} - 1)\lambda_s}[(\mathtt{m} - 1)\lambda_s]^k}{k!}
\end{split}
\end{equation}

Whether a hypothetical block is considered severely invalid relies on the presence of at least one full block mined within the previous $\Delta_i$ seconds. We denote the probability of a hypothetical block being severely invalid as ${\rm Pr}_s$, which can be computed as follows:

\begin{equation}
\begin{split}
{\rm Pr}_s =& [{\rm Pr}({\rm X} = 0)]^{\Delta_i}\\
=& e^{-(\lambda_s + \lambda_i)\Delta_i}.
\end{split}
\end{equation}

Consequently, the probability that a hypothetical block in shard $i$ is honest, denoted by $p_i$, is calculated as follows:

\begin{equation}
\begin{split}
p_i =& {\rm Pr}_h \cdot {\rm Pr}_s \\
=& e^{-(\lambda_s + \lambda_i)\Delta_i}\frac{\lambda_i\rho_i + \mathtt{m}\lambda_s\rho}{\lambda_i + \mathtt{m}\lambda_s}.
\end{split}
\end{equation}

Hence, the probability that a hypothetical block in shard $i$ is adversarial is bounded by $1 - p_i$.

We introduce several notations in this context. Let $E_v'$ denote the event that the private-mining attack violates the consistency of transaction $\mathtt{tx}$ in the severe execution environment $\tilde{\mathcal{E}nv}(\Delta, \rho)$. Assuming that the target transaction $\mathtt{tx}$ is included in a block at time $\tau$, let $h$ represent the height of the highest block mined by time $\tau$, and $\mathcal{A}dv$ represent the adversary's advantage by time $\tau$. Let $B$ denote the number of adversarial blocks out of the first $\max(2\kappa - \mathcal{A}dv, 0)$ blocks mined after $\tau$. Let $M$ be the maximum value obtained by subtracting the number of honest blocks from the number of adversarial blocks ever since time $\tau$. The violation event occurs if one of the following events occurs:

\begin{enumerate}
    \item $E_1$: $\mathcal{A}dv \geq \kappa$. The private chain is already sufficiently long to violate $\mathtt{tx}$'s consistency. The adversary publishes it as long as the longest public chain's length exceeds $h + \kappa$.
    \item $E_2$: $\mathcal{A}dv < \kappa$ and $\mathcal{A}dv + B \geq \kappa$. Prior to the longest length of the public chain reaching $h + \kappa$, the length of the private chain increases by at least $\kappa - \mathcal{A}dv$ and reaches $h + \kappa$. Consequently, the private chain is of sufficient length to violate the consistency of $\mathtt{tx}$.
    \item $E_3$: $\mathcal{A}dv + B < k$ and $M \geq 2(\kappa - \mathcal{A}dv - B) - 1$. In this case, the private chain is not sufficiently long for violating the consistency after the first $\max(2\kappa - \mathcal{A}dv, 0)$ are mined. However, there exits a time $\tau'$ when the gap between the length of private chain and that of the longest public chain reaches 0, i.e. $(h + (2\kappa -\mathcal{A}dv - B) - 1) - (h + \mathcal{A}dv + B + M) \geq 0$. The adversary publish the private chain at time $\tau$ and thereby violates the consistency.
\end{enumerate}

The advantage $\mathcal{A}dv$ can be viewed as the state of a continuous-time birth-death process. It begins at state 0, and each mining of an adversarial block corresponds to a birth, while each mining of an honest block corresponds to a death. The advantage follows a geometric distribution with a parameter of $p_i$. 

Suppose ${\rm gr}(i;p)$ and ${\rm Gr}(i;p)$ represent the PMF and complementary Cumulative Distribution Function (CDF) respectively of a geometric distribution with parameter $p$. Similarly, we denote by ${\rm br}(j;n, p)$ and ${\rm Br}(j;n, p)$ the PMF and complementary CDF of a binomial distribution with parameter $\left<n, p\right>$. Formally,

\begin{equation}
\begin{split}
&{\rm gr}(i;p) = (\frac{1-p}{p})^{i-1}(1 - \frac{1-p}{p}),\\
&{\rm Gr}(i;p) = (\frac{1-p}{p})^i,\\
&{\rm br}(j;n, p) = \tbinom{n}{j}p^j(1-p)^{n-j},\\
&{\rm Br}(j;n, p) = \sum_{k = j+1}^n {\rm br}(k;n, p).
\end{split}
\end{equation}

Therefore, the probability of event $E_1$ can be calculated as follows:

\begin{equation}
\begin{split}
{\rm Pr}[E_1] =& {\rm Gr}(\kappa;p_i),
\end{split}
\end{equation}

Hence $B$ is a binomial random variable with parameters $1-p$, thus

\begin{equation}
\begin{split}
{\rm Pr}[E_2] =& \sum_{l=0}^{\kappa-1}{\rm Pr}(\mathcal{A}dv = l)\cdot {\rm Pr}(B >= \kappa - l | \mathcal{A}dv = l) \\
=& \sum_{l=0}^{\kappa-1}{\rm gr}(l+1;p_i)\cdot {\rm Br}(\kappa - l - 1;2\kappa - l, 1-p_i).
\end{split}
\end{equation}

The $M$ implies the maximum reach of a simple random walk defined as follows: The walk starts at 0, and at each step, it increments by one with probability $p$ (when an adversarial block is generated), and decrements by one with probability $1-p$ (when an honest block is generated). $M$ follows a geometric distribution, specifically,

\begin{equation}
\begin{split}
&{\rm Pr}(M = l) = {\rm gr}(l+1;p_i),\\
&{\rm Pr}(M \geq l) = {\rm Gr}(l;p_i).
\end{split}
\end{equation}

Therefore, we calculate the probability of $E_3$ as follows:

\begin{equation}
\begin{split}
{\rm Pr}(E_3) =& \sum_{l=0}^{\kappa-1}{\rm Pr}(\mathcal{A}dv = l)\cdot \sum_{j=0}^{\kappa-l-1}{\rm Pr}(B = j|\mathcal{A}dv = l)\cdot\\
&{\rm Pr}(M \geq 2\kappa-1-2l-2j)\\
=& \sum_{l=0}^{\kappa-1}{\rm gr}(l+1;p_i)\cdot\sum_{j=0}^{\kappa-l-1}{\rm br}(j;2\kappa-l, 1-p_i)\cdot\\
&{\rm Gr}(2\kappa-1-2l-2j;p_i).
\end{split}
\end{equation}

In summary, the probability of violation event is 

\begin{equation}
\begin{split}
{\rm Pr}[E_v'] = {\rm Pr}[E_1] + {\rm Pr}[E_2] + {\rm Pr}[E_3].
\end{split}
\end{equation}

As $\mathcal{A}dv$ and $M$ follows the same geometric distribution, the Moment Generating Functions (MGF) of $\mathcal{A}dv$ and $M$ are expressed as 

\begin{equation}
\begin{split}
{\rm E}\{e^{v\mathcal{A}dv}\} = {\rm E}\{e^{vM}\} = \frac{p_i-(1-p_i)}{p_i - (1-p_i)e^v},
\end{split} 
\end{equation}

where $v < \log\frac{p_i}{1-p_i}$. Besides, $B$ follows a binomial distribution with parameter $\left<2\kappa-l,1-p_i\right>$ with a condition of $\mathcal{A}dv = l$. The MGF of $B$ is expressed as 

\begin{equation}
\begin{split}
{\rm E}\{e^{vB} | \mathcal{A}dv = l\} = (p_i + (1-p_i)e^v)^{2\kappa-l}.
\end{split}
\end{equation}

Because $\mathcal{A}dv, B, M$ are all non-negative, by combining the three events $E_1$, $E_2$, and $E_3$, we can concisely write $E_v'$ as: $2\mathcal{A}dv + 2B + M \geq 2\kappa - 1$. Utilizing the Chernoff bound, $\forall v > 0$, the following holds:

\begin{equation}
\begin{split}
&{\rm Pr}(2\mathcal{A}dv + 2B + M\geq 2\kappa - 1) \\
=& {\rm Pr}(\mathcal{A}dv + B + \frac{M}{2} \geq \kappa - \frac{1}{2}) \\
\leq& {\rm E}\{e^v(\mathcal{A}dv + B + \frac{M}{2} - \kappa + \frac{1}{2})\}\\
=& {\rm E}\{{\rm E}\{e^{v(\mathcal{A}dv + B)} | \mathcal{A}dv\}\}{\rm E}\{e^{v\frac{M}{2}}\}e^{-v(\kappa - \frac{1}{2})}\\
=& {\rm E}\{e^{v\mathcal{A}dv}(p_i + (1-p_i)e^v)^{2\kappa-\mathcal{A}dv}\}\frac{p_i-(1-p_i)}{p_i-(1-p_i)e^{\frac{v}{2}}}e^{-v(\kappa-\frac{1}{2})}\\
=& \frac{(p_i + (1-p_i)e^v)^{2\kappa}(p_i-(1-p_i))}{p_i-(1-p_i)\frac{e^v}{p_i+(1-p_i)e^v}}\frac{p_i-(1-p_i)}{p_i-(1-p_i)e^{\frac{v}{2}}}e^{-v(\kappa-\frac{1}{2})}\\
=& (p_i+(1-p_i)e^v)^{2\kappa}e^{-v\kappa}\frac{p_i-(1-p_i)^2(p_i+(1-p_i)e^v)e^{\frac{v}{2}}}{(p_i^2-(1-p_i)^2e^v)(p_i-(1-p_i)e^{\frac{v}{2}})}.
\end{split}
\end{equation}

Notice that the exponential coefficient for $\kappa$ is $2\log(p_i+(1-p_i)e^v) - v$, hence the tightest exponent for $\kappa$ is obtained with $e^v = \frac{p_i}{1-p_i}$. Therefore it holds that $p_i + (1-p_i)e^v = 2p_i$. Finally, we complete the proof of Theorem \ref{cp_property} by 

\begin{equation}
\begin{split}
{\rm Pr}[E_v] \leq& {\rm Pr}[E_v'] \\
\leq& {\rm Pr}(2\mathcal{A}dv + 2B + M\geq 2\kappa - 1)\\
\leq& (2 + 2\sqrt{\frac{p_i}{1-p_i}})(4p_i(1-p_i))^{\kappa}.
\end{split}
\end{equation}

\subsection{Proof of Chain Growth Property (Theorem \ref{cg_property})} \label{cg_property_proof}

We first demonstrate that the chain growth speed in a typical execution environment is not lower than that in a severe execution environment when both environments share the same randomness $\sigma$. By sharing the same randomness, we mean that the generation of hypothetical blocks is identical in both environments. 

\begin{lemma}\label{typical_faster_severe}
Given a typical execution environment $\mathcal{E}nv(\Delta, \rho)$ and a severe execution environment $\tilde{\mathcal{E}nv}(\Delta, \rho)$ sharing the randomness $\sigma$, for all time $t$ and miner $j$ that is honest at time $t$ in $\mathcal{E}nv(\Delta, \rho)$, $|\mathcal{C}_j^t(\mathcal{E}nv)| \geq |\mathcal{C}_j^t(\tilde{\mathcal{E}nv})|$.
\end{lemma}

\begin{proof}
We proof Lemma \ref{typical_faster_severe} by induction.

$t=0$: All honest miners in both $\mathcal{E}nv(\Delta, \rho)$ and $\tilde{\mathcal{E}nv}(\Delta, \rho)$ only contain the genesis block, $|\mathcal{C}_j^t(\mathcal{E}nv)| = |\mathcal{C}_j^t(\tilde{\mathcal{E}nv})| = 1$.

$t = i$: Assuming $|\mathcal{C}_j^i(\mathcal{E}nv)| \geq |\mathcal{C}_j^i(\tilde{\mathcal{E}nv})|$ holds.

$t = k$, such that $k$ is the minimum time when miner $j$ is honest after $i$: If $|\mathcal{C}_j^k(\tilde{\mathcal{E}nv})| = |\mathcal{C}_j^i(\tilde{\mathcal{E}nv})|$, then
\begin{equation}
\begin{split}
|\mathcal{C}_j^k(\mathcal{E}nv)| \geq& |\mathcal{C}_j^i(\mathcal{E}nv)| \geq |\mathcal{C}_j^i(\tilde{\mathcal{E}nv})| = |\mathcal{C}_j^k(\tilde{\mathcal{E}nv})|
\end{split}
\end{equation}

holds. If $|\mathcal{C}_j^k(\tilde{\mathcal{E}nv})| > |\mathcal{C}_j^i(\tilde{\mathcal{E}nv})|$, during $(i-\Delta, k-\Delta]$ in $\tilde{\mathcal{E}nv}$, there must exists $n$ blocks (not the hypothetical blocks) mined on top of $\mathcal{C}_j^i$ and received by miner $j$ by time $k$. Because $\mathcal{E}nv$ shares the same randomness with $\tilde{\mathcal{E}nv}$, the same blocks are also mined during the same period in $\mathcal{E}nv$ and received by miner $j$ by time $k$. Conditioned the same length increment, the following holds: 
\begin{equation}
\begin{split}
|\mathcal{C}_j^k(\mathcal{E}nv)| = |\mathcal{C}_j^i(\mathcal{E}nv)| + n \geq |\mathcal{C}_j^i(\tilde{\mathcal{E}nv)}| + n = |\mathcal{C}_j^k(\tilde{\mathcal{E}nv)}|.
\end{split}
\end{equation}
Therefore, $|\mathcal{C}_j^i(\mathcal{E}nv)| \geq |\mathcal{C}_j^i(\tilde{\mathcal{E}nv})|$ also holds for $t = k > i$, thereby complete the proof of Lemma \ref{typical_faster_severe}.
\end{proof}

We proceed to establish a lower bound on the growth of the longest public chain in $\tilde{\mathcal{E}nv}$ over a period of $t$ seconds. Let $\mathtt{L}^t(\mathcal{E}nv)$ represent the length of the longest chain maintained by any honest miner at time $t$ in the typical environment $\mathcal{E}nv$. To analyze the protocol execution in $\tilde{\mathcal{E}nv}$ starting from time $r$, we define $\mathtt{L}^t(\tilde{\mathcal{E}nv})$ as the corresponding chain length in $\tilde{\mathcal{E}nv}$. Since $\mathtt{L}^r(\tilde{\mathcal{E}nv}) = \mathtt{L}^r(\mathcal{E}nv)$ by definition, we can treat the longest chain mined by $r$ as the genesis block in $\tilde{\mathcal{E}nv}$, where $\mathtt{L}^r(\tilde{\mathcal{E}nv}) = \mathtt{L}^r(\mathcal{E}nv) = 1$.

\begin{lemma}
For any $r, t \geq 0$ and any $\delta > 0$, within any shard $i$,
\begin{equation}
\begin{split}
{\rm Pr}[\mathtt{L}^{r+t}(\tilde{\mathcal{E}nv}_i) < \mathtt{L}^{r}(\tilde{\mathcal{E}nv}_i) + (1 -\delta)\mathtt{g}_it] < e^{-\Omega(\delta^2\mathtt{g}_it)},
\end{split}
\end{equation}
where $\mathtt{g}_i = \frac{p_i(\lambda_s + \lambda_i)}{1 + \Delta_ip_i(\lambda_s + \lambda_i)}$.
\end{lemma}

\begin{proof}
In $\tilde{\mathcal{E}nv}_i$, every block is delayed for $\Delta_i$ seconds and during the delivery period no other blocks are generated. As a result, given the fixed $r$ and after that all honest miners would agree on the same chain. 

Within each second, the probability of successfully generating a honest block is given by $p_i(\lambda_s + \lambda_i)$, where $(\lambda_s + \lambda_i)$ represents the probability of generating a hypothetical block, and $p_i$ represents the probability that this hypothetical block is honest. Given that each generation of a honest block corresponds to a frozen period of at most $\Delta_i$ seconds ($\Delta_i$ seconds for a full block and $1$ second for a light block), if the increase in the length of the longest chain is less than $c$ after $t$ seconds, there must be at least $t-c\Delta_i$ unfrozen seconds available. Moreover, during these $t-c\Delta_i$ unfrozen seconds, the number of honest blocks is limited to be less than $c$. Let $W^k$ represent the sum of $k$ independent binary random variables $w_i$, where each $w_i$ takes the value 1 with probability $p_i(\lambda_s + \lambda_i)$. For each possible $c$ the following holds

\begin{equation}
\begin{split}
E[W^{t-c\Delta}]  =  p_i(\lambda_s + \lambda_i)(t-c\Delta_i) = c,
\end{split}
\end{equation}

which implies that $c = \frac{p_i(\lambda_s + \lambda_i)t}{1 + \Delta_ip_i(\lambda_s + \lambda_i)} = \mathtt{g}_it$. Thus we have

\begin{equation}
\begin{split}
E[W^{t-c\Delta}] = \mathtt{g}_it.
\end{split}
\end{equation}

By the Chernoff bound, we have 

\begin{equation}
\begin{split}
{\rm Pr}[W^{t-c\Delta_i} < (1-\delta)\mathtt{g}_it] \leq e^{-\Omega(\delta^2\mathtt{g}_it)}.
\end{split}
\end{equation}

As the increase of the longest public chain is bounded by the number of generated honest blocks, we complete the proof as follows:

\begin{equation}
\begin{split}
{\rm Pr}[\mathtt{L}^{r+t}(\tilde{\mathcal{E}nv}_i)& < \mathtt{L}^{r}(\tilde{\mathcal{E}nv}_i) + (1 -\delta)\mathtt{g}_it] \\
=& {\rm Pr}[\mathtt{L}^{r+t}(\tilde{\mathcal{E}nv}_i) - \mathtt{L}^{r}(\tilde{\mathcal{E}nv}_i) < (1 -\delta)\mathtt{g}_it] \\
\leq& {\rm Pr}[W^{t-c\Delta_i} < (1-\delta)\mathtt{g}_it] \\
\leq& e^{-\Omega(\epsilon^2\mathtt{g}_it)}.
\end{split}
\end{equation}

\end{proof}

Now we extend the chain growth lower bound to $\mathcal{E}nv$.

\begin{lemma}
For any $r, t \geq 0$ and any $\epsilon > 0$, within any shard $i$,
\begin{equation}
\begin{split}
{\rm Pr}[\mathtt{L}^{r+t}(\mathcal{E}nv_i) < \mathtt{L}^{r}(\mathcal{E}nv_i) + (1 -\delta)\mathtt{g}_it] < e^{-\Omega(\epsilon^2\mathtt{g}_it)},
\end{split}
\end{equation}
where $\mathtt{g}_i = \frac{p_i(\lambda_s + \lambda_i)}{1 + \Delta_ip_i(\lambda_s + \lambda_i)}$.
\end{lemma}

\begin{proof}

By definition, $\mathtt{L}^{r}(\mathcal{E}nv_i) = \mathtt{L}^{r}(\tilde{\mathcal{E}nv}_i)$. According to Lemma \ref{typical_faster_severe}, $\mathtt{L}^{r+t}(\mathcal{E}nv_i) \geq \mathtt{L}^{r+t}(\tilde{\mathcal{E}nv}_i)$. Thus, 

\begin{equation}
\begin{split}
&{\rm Pr}[\mathtt{L}^{r+t}(\mathcal{E}nv_i) < \mathtt{L}^{r}(\mathcal{E}nv_i) + (1 -\delta)\mathtt{g}_it]\\ 
\leq& {\rm Pr}[\mathtt{L}^{r+t}(\tilde{\mathcal{E}nv}_i) < \mathtt{L}^{r}(\tilde{\mathcal{E}nv}_i) + (1 -\delta)\mathtt{g}_it] \\
<&  e^{-\Omega(\epsilon^2\mathtt{g}_it)}.
\end{split}
\end{equation}
\end{proof}

As every block mined before $r-\Delta_i$ will appear in the view of every honest miner after $r$, we can establish the following lemma:

\begin{lemma}
For any time $r$ and honest miners $i,j$ in shard $i$, $|\mathcal{C}_i^{r}| \geq \mathcal{C}_i^{r-\Delta_i}$,
\end{lemma}

thereby satisfying the consistent length condition of the anticipated growth event. Note that for any time $r$ and interval $t$, , the following holds:

\begin{equation}
\begin{split}
&\min\limits_{i,j}(|\mathcal{C}_j^{r+t}| - |\mathcal{C}_i^{r}|)\\
=& \min\limits_j|\mathcal{C}_j^{r+t}| - \max\limits_i|\mathcal{C}_i^{r}|\\
\geq& \max\limits_j|\mathcal{C}_j^{r+t-\Delta_i}| - \max\limits_i|\mathcal{C}_i^{r}| \\
=& \mathtt{L}^{r+t-\Delta_i}(\mathcal{E}nv_i) -  \mathtt{L}^{r}(\mathcal{E}nv_i).
\end{split}
\end{equation}

Therefore, 

\begin{equation}
\begin{split}
&{\rm Pr}[E_g(t,\Delta_i, \mathcal{T}) = 0] \\
=& {\rm Pr}[\min\limits_{i,j}(|\mathcal{C}_j^{r+t}| - |\mathcal{C}_i^{r}|) < \mathcal{T}] \\
=& {\rm Pr}[\mathtt{L}^{r+t-\Delta_i}(\mathcal{E}nv_i) -  \mathtt{L}^{r}(\mathcal{E}nv_i) < \mathcal{T}] \\
\leq& {\rm Pr}[\mathtt{L}^{r+t-\Delta_i}(\mathcal{E}nv_i) -  \mathtt{L}^{r}(\mathcal{E}nv_i) < g_it]\\
=& {\rm Pr}[\mathtt{L}^{r+t-\Delta_i}(\mathcal{E}nv_i) -  \mathtt{L}^{r}(\mathcal{E}nv_i) < (1-\delta)\mathtt{g}_it] \\
<& e^{-\Omega(\delta^2\mathtt{g}_i(t-\Delta_i))} = \varepsilon(\mathtt{g_i}t) < \varepsilon(\mathcal{T}).
\end{split}
\end{equation}

Consequently, ${\rm Pr}[E_g(t,\Delta_i, \mathcal{T}) = 1] = 1 - {\rm Pr}[E_g(t,\Delta_i, \mathcal{T}) = 0] > 1 - \varepsilon(\mathcal{T})$. Hence we complete the proof of Theorem \ref{cg_property}.

\subsection{Proof of Chain Quality Property (Theorem \ref{cq_property})} \label{cq_property_proof}

We present the bounded number of total blocks and adversarial blocks within shard $i$ in $\mathcal{E}nv(\Delta_i, \rho_i)$ by the following two lemmas

\begin{lemma}\label{upper_block_num}
Let $Q_t(\mathcal{E}nv)$ denote the maximum number of blocks mined within any interval of $t$ seconds in $\mathcal{E}nv$, for any $t\geq0$ and any $\delta$
\begin{equation}
\begin{split}
{\rm Pr}[Q_t(\mathcal{E}nv) > (1+\delta)(\lambda_s + \lambda_i)t] < e^{-\Omega(\delta^2(\lambda_s + \lambda_i)t)}.
\end{split}
\end{equation}
\end{lemma}

\begin{lemma}\label{upper_adv_block_num}
Let $A_t(\mathcal{E}nv)$ denote the maximum number of adversarial blocks mined within any interval of $t$ seconds in $\mathcal{E}nv$, for any $t\geq0$ and any $\delta$
\begin{equation}
\begin{split}
{\rm Pr}[A_t(\mathcal{E}nv) >& (1+\delta)(1-\rho_i)(\lambda_s + \lambda_i)t] \\
&< e^{-\Omega(\delta^2(1-\rho_i)(\lambda_s + \lambda_i)t)}.
\end{split}
\end{equation}
\end{lemma}

\begin{proof}
Since the generation of blocks in each second follows an independent Poisson process, the expected number of blocks mined in one second is $\lambda_s + \lambda_i$, out of which $1-\rho_i$ are adversarial. Over an interval of $t$ seconds, on average, $(\lambda_s + \lambda_i)t$ blocks are mined, and $(1-\rho_i)(\lambda_s + \lambda_i)t$ of them are adversarial. By applying the Chernoff bound, we can directly prove Lemma \ref{upper_block_num} and Lemma \ref{upper_adv_block_num}.
\end{proof}

Considering a block sequence $\mathcal{B}j, \mathcal{B}{j+1}, ..., \mathcal{B}{j+T}$, where $\mathcal{B}{j-1}$ and $\mathcal{B}_{j+T+1}$ are honest, as involving these two blocks will increase the fraction of honest blocks in the sequence. Let $r'$ be the time when $\mathcal{B}_{j-1}$ is mined, and $r' + t$ be the time when $\mathcal{B}_{j+T+1}$ is mined. 

We consider three expected events with negligible probability. The simultaneous occurrence of these three events ensures the chain quality property.

$\mathtt{EXP}_1$: For any $\delta \in (0,1)$, it holds that $Q_{T/(1+\delta)(\lambda_s+\lambda_i)t} < T$ except with a probability of $e^{-\Omega(\delta^2T)}$ (as stated in Lemma \ref{upper_block_num}). Consequently, $t > T/(1+\delta)(\lambda_s+\lambda_i)t$ except with a probability of $\epsilon_1(T) = e^{-\Omega(\delta^2T)}$.

$\mathtt{EXP}_2$: Due to Theorem \ref{cg_property}, the chain growth is at least $(1-\delta')\mathtt{g}$, where $\mathtt{g} = \frac{p_i(\lambda_i + \lambda_s)}{1 + \Delta(\lambda_s + \lambda_i)}$. Therefore except with probability of $\epsilon_2(T)$, we have $t \leq \frac{T}{(1-\delta')\mathtt{g}}$.

$\mathtt{EXP}_3$: According to Lemma \ref{upper_adv_block_num}, for any $\delta''\in(0,1)$, there are at most $(1+\delta'')(1-\rho_i)(\lambda_s + \lambda_i)t$ adversarial blocks mined during the interval, except with probability of $\epsilon_3(T) < e^{-\Omega(\delta''(1-\rho_i)(\lambda_s + \lambda_i)t)}$.

Assuming all expected events happens except with probability $\epsilon(T)$, the fraction of adversarial blocks $\overline{\mathtt{q}}$ is computed as follows:

\begin{equation}
\begin{split}
&\overline{\mathtt{q}} = (1+\delta'')(1-\rho_i)(\lambda_s + \lambda_i)t / T \\
\leq& \frac{1+\delta''}{1-\delta'}\cdot T\cdot \frac{[1 + \Delta_ip_i(\lambda_s + \lambda_i)(1-\rho_i)]}{p_i} / T\\
\leq& \frac{1+\delta''}{1-\delta'}\cdot \frac{[1 + \Delta_ip_i(\lambda_s + \lambda_i)(1-\rho_i)]}{p_i}.
\end{split}
\end{equation}

For every $\delta > 0$, by selecting sufficiently small $\delta'$ and $\delta''$, $\overline{\mathtt{q}}$ is bounded by
 
\begin{equation}
\begin{split}
\overline{\mathtt{q}} \leq (1+\delta)\frac{[1 + \Delta_ip_i(\lambda_s + \lambda_i)(1-\rho_i)]}{p_i}.
\end{split}
\end{equation}

thus the fraction of honest blocks is at least $1-\overline{\mathtt{q}}$, thereby completing the proof of Theorem \ref{cq_property}.

% \section{Proof of Theorem \ref{btc_tps}}\label{btc_tps_proof}

% Bitcoin's throughput

\section{Proof of Throughput (Theorem \ref{btc_tps} \& \ref{mfd_tps})}  \label{mfd_tps_proof}

We consider the Manifoldchain $\prod_{mfd}(f, \lambda_s, \{\lambda_i\}^m)$ in environment $\mathcal{E}nv(\Delta, \rho)$ that holds the efficient-CP property $\mathcal{ECQ}(\kappa, \varepsilon(\kappa), \lambda)$ and satisfies the system security requirement $\overline{\mathcal{P}}$. Assuming $\lambda_t = \lambda_s + \lambda_i$ represents the total mining rate in shard $i$. According to the Ineq.\ref{eq_pi}, $p_i$ is bounded as follows:

\begin{equation}
\begin{split}
p_i =& \frac{\lambda_i\rho_i + m\lambda_s\rho}{\lambda_i + m\lambda_s}e^{-(\lambda_i + \lambda_s)\Delta_i}\\
\leq& e^{-(\lambda_i + \lambda_s)\Delta_i}\\
=& e^{-\lambda_t\Delta_i}.
\end{split} 
\end{equation}

In a blockchain satisfies liveness, $\lambda_t > 0$. We denote by $\underline{\Delta}$, $\overline{\Delta}$ the lower bound, upper bound of $\Delta_i$. $p_i$ satisfies

\begin{equation}
\begin{split}
\frac{1}{2} < p_i \leq e^{-\underline{\lambda}\underline{\Delta}} = \overline{p} < 1,
\end{split}
\end{equation}

where $\overline{p}$ represents the upper bound of $p_i$ in a growing blockchain protocol. Consquently we can bound $\varepsilon(\kappa)$ as follows:

\begin{equation}
\begin{split}
\varepsilon(\kappa) \leq & (2 + 2\sqrt{\frac{1}{\frac{1}{p_i}} - 1})(4p_i(1-p_i))^k \\
\leq& (2 + 2\sqrt{\frac{\overline{p}}{1 - \overline{p}}})(4p_i(1-p_i))^k.
\end{split}
\end{equation}

$\prod_{mfd}(f, \lambda_s,\{\lambda_i\}^m)$ satisfies the system security requirement as long as 

\begin{equation}
\begin{split}
 (2 + 2\sqrt{\frac{\overline{p}}{1 - \overline{p}}})(4p_i(1-p_i))^k \leq \overline{\mathcal{P}},
\end{split}
\end{equation}

which requires that

\begin{equation}\label{eq2}
\begin{split}
p_i \geq \frac{1 + \sqrt{1 - \sqrt[k]{\frac{\overline{\mathcal{P}}}{2 + 2\sqrt{\frac{\overline{p}}{1-\overline{p}}}}}}}{2}.
\end{split}
\end{equation}

For simplification, we represent Ineq. \ref{eq2} equivalently as follows:

\begin{equation}%\label{eq2}
\begin{split}
&p_i \geq \delta(\overline{\mathcal{P}}, \kappa), where\\
&\delta(\overline{\mathcal{P}}, \kappa) = \frac{1 + \sqrt{1 - \sqrt[k]{\frac{\overline{\mathcal{P}}}{2 + 2\sqrt{\frac{\overline{p}}{1-\overline{p}}}}}}}{2} \in (\frac{1}{2}, 1).
\end{split}
\end{equation}

Assuming a real scenario that honest miners can not known the exact fraction of adversarial hashing power, the following inequality has to be satisfied:

\begin{equation}
\begin{split}
p_i > e^{-(\gamma_i + 1)\lambda_s\Delta_i}\frac{m\rho}{\gamma_i + m} > \delta(\overline{\mathcal{P}}, \kappa).
\end{split}
\end{equation}

Hence we have 

\begin{equation}
\begin{split}
\lambda_s < \frac{1}{(\gamma_i + 1)\Delta_i}\log\frac{m\rho}{\delta(\overline{\mathcal{P}}, \kappa)(\gamma_i + m)}.
\end{split}
\end{equation}

We propose a lemma as follows to ensure the non-negativity of the throughput.

\begin{lemma}
$\p_i \leq \rho$ as long as the following constrains hold:
\begin{equation}
\begin{split}
&\underline{\lambda} \geq \frac{1}{(\gamma_i + 1)\Delta_i}\log(\frac{\gamma_i}{\gamma_i + m}\cdot\frac{\rho_i}{\rho} + 1),\\
&\gamma_i < \frac{\rho(\frac{\rho}{\delta(\overline{\mathcal{P}}, \kappa)}-1)m}{\rho + 1}.
\end{split}
\end{equation}
\end{lemma}

\begin{proof}
Initially we prove the following formula to ensure the existence of real $\lambda_s$:
\begin{equation}\label{real_lambda_s}
    \begin{split}
         \frac{1}{(\gamma_i + 1)\Delta_i}\log(\frac{\gamma_i}{\gamma_i + m}\cdot\frac{\rho_i}{\rho} + 1) < \frac{1}{(\gamma_i + 1)\Delta_i}\log(\frac{m}{\delta(\overline{\mathcal{P}}, \kappa)}\rho)        
    \end{split}
\end{equation}

Specifically,

\begin{equation}
\begin{split}
&\gamma_i < \frac{\rho(\frac{\rho}{\delta(\overline{\mathcal{P}}, \kappa)}-1)m}{\rho + 1},\\
&\gamma_i < \frac{(\frac{\rho}{\delta(\overline{\mathcal{P}}, \kappa)}-1)m}{\frac{1}{\rho} + 1},\\
&\gamma_i < \frac{(\frac{\rho}{\delta(\overline{\mathcal{P}}, \kappa)}-1)m}{\frac{\rho_i}{\rho} + 1},\\
&(\frac{\rho_i}{\rho} + 1) \gamma_i < (\frac{\rho}{\delta(\overline{\mathcal{P}}, \kappa)} - 1)m,\\
&\gamma_i\frac{\rho_i}{\rho} + \gamma_i + m < \frac{m\rho}{\delta(\overline{\mathcal{P}}, \kappa)},\\
\end{split}    
\end{equation}
\begin{equation}
\begin{split}
&\frac{\gamma_i}{\gamma_i + m}\frac{\rho_i}{\rho} + 1 < \frac{m}{\delta(\overline{\mathcal{P}}, \kappa)(\gamma_i + m)}\rho,\\
&\frac{1}{(\gamma_i + 1)\Delta_i}\log(\frac{\gamma_i}{\gamma_i + m}\cdot\frac{\rho_i}{\rho} + 1) < \frac{1}{(\gamma_i + 1)\Delta_i}\log(\frac{m}{\delta(\overline{\mathcal{P}}, \kappa)}\rho).
\end{split}    
\end{equation}

Therefore, there exits a real $\lambda_s$ which satisfies the Ineq.\ref{real_lambda_s}. As $\lambda \geq \underline{\lambda}$, we have
\begin{equation}
\begin{split}
&\lambda_s \geq \underline{\lambda} \geq \frac{1}{(\gamma_i + 1)\Delta_i}\log(\frac{\gamma_i}{\gamma_i + m}\cdot\frac{\rho_i}{\rho} + 1),\\
&e^{(\gamma_i + 1)\lambda_s\Delta_i}\geq \frac{\gamma_i}{\gamma_i + m}\cdot\frac{\rho_i}{\rho} + 1,\\
&e^{(\gamma_i + 1)\lambda_s\Delta_i} - 1 \geq \frac{\gamma_i}{\gamma_i + m}\cdot\frac{\rho_i}{\rho},\\
&\frac{1}{e^{(\gamma_i + 1)\lambda_s\Delta_i} - 1} \leq (\frac{m}{\gamma_i} + 1)\frac{\rho}{\rho_i},\\
\end{split}
\end{equation}
\begin{equation}
\begin{split}
& \frac{\rho_i}{\rho}\frac{1}{e^{(\gamma_i + 1)\lambda_s\Delta_i} - 1} - 1\leq \frac{m}{\gamma_i},\\
& \frac{\gamma_i}{\rho}\frac{\rho_i}{e^{(\gamma_i + 1)\lambda_s\Delta_i} - 1} - \gamma_i \leq m,\\
&\frac{\gamma_i}{\rho}\frac{\rho_i - \rho(e^{(\gamma_i + 1)\lambda_s\Delta_i} - 1)}{e^{(\gamma_i + 1)\lambda_s\Delta_i} - 1} \leq m,\\
&\frac{\gamma_i}{\rho}\frac{\rho_i - \rho e^{(\gamma_i + 1)\lambda_s\lambda_i} + \rho}{e^{(\gamma_i + 1)\lambda_s\Delta_i} - 1} \leq m,\\
&\frac{\gamma_i}{\rho}\frac{\rho_i - \rho e^{(\gamma_i + 1)\lambda_s\lambda_i}}{e^{(\gamma_i + 1)\lambda_s\Delta_i} - 1} \leq m,\\
&\gamma_i\rho_i - \gamma_i\rho e^{(\gamma_i + 1)\lambda_s\lambda_i} \leq m\rho(e^{(\gamma_i + 1)\lambda_s\lambda_i} - 1), \\
&\frac{\gamma_i\rho_i + m\rho}{\gamma_i + m} \leq \rho e^{(\gamma_i + 1)\lambda_s\lambda_i},\\
&e^{-(\gamma_i + 1)\lambda_s\lambda_i}\frac{\gamma_i\rho_i + m\rho}{\gamma_i + m} = p_i \leq \rho.
\end{split}
\end{equation}

\end{proof}

Therefore, holding the following constrains:

\begin{equation}
\begin{split}
&\gamma_i < \frac{\rho(\frac{\rho}{\delta(\overline{\mathcal{P}}, \kappa)}-1)m}{\rho + 1},\\
&\frac{\log(\frac{\gamma_i}{\gamma_i + m}\cdot\frac{\rho_i}{\rho} + 1)}{(\gamma_i + 1)\Delta_i} \leq\lambda_s < \frac{\log\frac{m\rho}{\delta(\overline{\mathcal{P}}, \kappa)(\gamma_i + m)}}{(\gamma_i + 1)\Delta_i},\\    
\end{split}
\end{equation}

$\prod_{mfd}$ guarantees security and achieves a non-negative throughput of 

\begin{equation}
\begin{split}
\mathcal{T}PS(\prod_{mfd}, \mathcal{E}nv, \mathcal{ECP}) = \lambda_s + \lambda_i = (\gamma_i + 1)\lambda_s.
\end{split}
\end{equation}

We need to determine the optimal values of $\{\lambda\} = \{\lambda_s, \lambda_0, ...\lambda_{m-1}\}$ to maximize the throughput with the following constrain:

\begin{equation}
\begin{split}
\lambda_s + \lambda_i = (\gamma_i + 1)\lambda_s \leq \frac{1}{\Delta_i}\log\frac{m\rho}{\delta(\overline{\mathcal{P}}, \kappa)(\gamma_i + m)}.
\end{split}
\end{equation}

We now present our result as follows:

\begin{lemma}\label{appendix_max_tps_lemma}
The throughput of shard $i$ is maximized when the $\{\lambda\}$ is chosen as $\{\lambda\}^*$, where 
\begin{equation}
\begin{split}
&\lambda_s^* = \frac{1}{\overline{\Delta}}\log\frac{\rho}{\delta(\overline{\mathcal{P}}, \kappa)},\\
&\lambda_i^* = \gamma_i^*\lambda_s^*,\\
\end{split}
\end{equation}
and $\gamma_i$ satisfies the constrain:
\begin{equation}
\begin{split}
\frac{1}{\gamma_i^* + 1} \frac{\log \frac{m\rho}{\delta(\overline{\mathcal{P}}, \kappa)(\gamma_i^* + m)}}{\log \frac{\rho}{\delta(\overline{\mathcal{P}}, \kappa)}} = \frac{\Delta_i}{\overline{\Delta}}.
\end{split}
\end{equation}
\end{lemma}

\begin{proof}
Assuming that there exits another $\{\lambda\}'$ which achieves a greater throughput. Given any possible $\lambda_s$, to ensure the security in the shard with maximum network delay, it must satisfies $\lambda_s \leq \frac{1}{\overline{\Delta}(\gamma_i + 1)}\log \frac{m\rho}{\delta(\overline{\mathcal{P}}, \kappa)(\gamma_i + m)} \leq \frac{1}{\Delta_i}\log\frac{\rho}{\delta(\overline{\mathcal{P}}, \kappa)} = \lambda_s^*$. Besides, $\lambda_i^*$ is the maximum value to be chosen when fixing $\lambda_s = \lambda_s^*$. If there exits $\lambda_i > \lambda_i^*$, then

\begin{equation}
\begin{split}
&\frac{1}{\gamma_i + 1} \frac{\log \frac{m\rho}{\delta(\overline{\mathcal{P}}, \kappa)(\gamma_i + m)}}{\log \frac{\rho}{\delta(\overline{\mathcal{P}}, \kappa)}} < \frac{\Delta_i}{\overline{\Delta}},\\
&\frac{1}{(\gamma_i + 1)\Delta_i}\log\frac{m\rho}{\delta(\overline{\mathcal{P}}, \kappa)(\gamma_i + m)} < \frac{1}{\overline{\Delta}}\log\frac{\rho}{\delta(\overline{\mathcal{P}} = \lambda_{s}^*, \kappa)},\\
&\lambda_s + \lambda_i = (\gamma_i + 1)\lambda_s > \frac{1}{\Delta_i}\log\frac{m\rho}{\delta(\overline{\mathcal{P}}, \kappa)(\gamma_i + m)},
\end{split}
\end{equation}

which is in contradiction with the security requirement. Therefore, $\lambda_s' < \lambda_s^*$ and $\lambda_i' > \lambda_i^*$. In this case, $\gamma_i^* < \gamma_i'$ for any $i$. Consequently,
\begin{equation}
\begin{split}
\mathcal{T}PS(\{\lambda\}') &= \lambda_s' + \lambda_i' = \frac{1}{\Delta_i}\log\frac{m\rho}{\delta(\overline{\mathcal{P}}, \kappa)(\gamma_i' + m)}\\
&< \frac{1}{\Delta_i}\log\frac{m\rho}{\delta(\overline{\mathcal{P}}, \kappa)(\gamma_i^* + m)} = \mathcal{T}PS(\{\lambda\}^*),
\end{split}
\end{equation}
which is in contradiction with the assumption, thereby complete the proof.
\end{proof}

Hence, we can obtain the maximum total throughput over $m$ shards and complete the proof of Theorem \ref{mfd_tps}:

\begin{equation}
\begin{split}
\mathcal{T}PS^* &= m\rho\lambda_s + \rho_i\sum_{i}^m\lambda_i\\
&= m\rho\lambda_s + \lambda_s\sum_i^m \gamma_i\rho_i\\
&= \lambda_s(m\rho + \sum_i^m \gamma_i\rho_i)\\
&= \frac{1}{\Delta}\log\frac{\rho}{\delta(\overline{\mathcal{P}}, \kappa)}(m\rho + \sum_i^m \gamma_i\rho_i), 
\end{split}
\end{equation}

where 

\begin{equation}
\begin{split}
&\frac{1}{\gamma_i + 1} \frac{\log \frac{m\rho}{\delta(\overline{\mathcal{P}}, \kappa)(\gamma_i + m)}}{\log \frac{\rho}{\delta(\overline{\mathcal{P}}, \kappa)}} \geq \frac{\Delta_i}{\Delta},\\
&\underline{\lambda} \geq \frac{1}{(\gamma_i + 1)\Delta_i}\log(\frac{\gamma_i}{\gamma_i + m}\cdot\frac{\rho_i}{\rho} + 1),\\
&\gamma_i < \frac{\rho(\frac{\rho}{\delta(\overline{\mathcal{P}}, \kappa)}-1)m}{\rho + 1}.
\end{split}
\end{equation}

% that's all folks
\end{document}